\documentclass[a4paper,USenglish, cleveref]{lipics-v2021}

\usepackage{microtype}

\usepackage{enumerate}
\usepackage{comment}
\usepackage{amsmath,amssymb}
\usepackage[linesnumbered,noend,ruled,vlined]{algorithm2e}
\SetKwInput{KwInput}{Input}                
\SetKwInput{KwOutput}{Output}              

\usepackage{graphicx}
\usepackage{amsthm}
\usepackage{amsmath}
\usepackage{amssymb}
\usepackage{todonotes}
\usepackage{regexpatch}
\usepackage{url}
\usepackage{mathtools}
\usepackage{enumitem}

\usepackage{xstring}

\usepackage{thm-restate}

\usepackage{hyperref}

\usepackage{tikz}
\usetikzlibrary{positioning, fit, calc, arrows.meta}
\usetikzlibrary{shapes.arrows}
\usetikzlibrary{backgrounds}
\usetikzlibrary{decorations, decorations.markings, decorations.pathreplacing, decorations.pathmorphing}

\input{tikz/cases/basetree.tex}

\tikzset{vertexSec4/.style={draw=black, fill=white, auto=left,circle,minimum size=14pt,inner sep=0pt}} 
\tikzset{vertexSpSec4/.style={draw=blue, very thick, fill=white, auto=left,circle,minimum size=14pt,inner sep=0pt}} 
\newcommand{\convexpath}[2]{
  [   
  create hullcoords/.code={
    \global\edef\namelist{#1}
    \foreach [count=\counter] \nodename in \namelist {
      \global\edef\numberofnodes{\counter}
      \coordinate (hullcoord\counter) at (\nodename);
    }
    \coordinate (hullcoord0) at (hullcoord\numberofnodes);
    \pgfmathtruncatemacro\lastnumber{\numberofnodes+1}
    \coordinate (hullcoord\lastnumber) at (hullcoord1);
  },
  create hullcoords
  ]
  ($(hullcoord1)!#2!-90:(hullcoord0)$)
  \foreach [
  evaluate=\currentnode as \previousnode using \currentnode-1,
  evaluate=\currentnode as \nextnode using \currentnode+1
  ] \currentnode in {1,...,\numberofnodes} {
    let \p1 = ($(hullcoord\currentnode) - (hullcoord\previousnode)$),
    \n1 = {atan2(\y1,\x1) + 90},
    \p2 = ($(hullcoord\nextnode) - (hullcoord\currentnode)$),
    \n2 = {atan2(\y2,\x2) + 90},
    \n{delta} = {Mod(\n2-\n1,360) - 360}
    in 
    {arc [start angle=\n1, delta angle=\n{delta}, radius=#2]}
    -- ($(hullcoord\nextnode)!#2!-90:(hullcoord\currentnode)$) 
  }
}

\newcommand{\yes}{{\sf YES}}
\newcommand{\no}{{\sf NO}}

\newcommand{\ini}{{\sf s}} 
\newcommand{\tar}{{\sf t}} 

\newcommand{\diam}[1]{{\sf diam}(#1)} 
\newcommand{\dg}[2]{d_{#1}(#2)} 
\newcommand{\neig}[2]{N_{#1}(#2)} 
\newcommand{\inc}[2]{\delta_{#1}(#2)} 

\newcommand{\spt}{T} 
\newcommand{\thr}{d} 

\newcommand{\unip}[3]{{#1}[#2,#3]} 

\newcommand{\inner}[1]{{\sf center}(#1)} 

\Crefname{algorithm}{Algorithm}{Algorithms}
\Crefname{section}{Sect.}{Sects.}
\Crefname{lemma}{Lemma}{Lemmas}
\Crefname{claim}{Claim}{Claims}
\Crefname{theorem}{Theorem}{Theorems}
\Crefname{proposition}{Proposition}{Propositions}
\crefname{enumi}{}{}
\crefname{caseenumi}{}{}
\crefname{caseenumii}{}{}
\crefname{enumcondi}{}{}
\Crefname{figure}{Fig.}{Figs.}
\Crefname{prop}{Prop.}{Props.}

\newcommand{\reconf}[3]{#1 \overset{#3}{\longleftrightarrow} #2}
\newcommand{\vlabel}[2]{{\sf label}_{#2}(#1)}
\newcommand{\distorig}[3]{\bar{\ell}_{#3}(#1, #2)}
\newcommand{\dist}[3]{\ell_{#3}(#1, #2)}

\newcommand{\f}[1]{f\left(#1\right)}

\newcommand{\UniC}[1]{C_{#1}}

\newcommand{\ecc}[2]{\epsilon_{#2}\left(#1\right)}
\newcommand{\uvec}[1]{\chi_{#1}}
\newcommand{\idx}[1]{i(#1)}

\newcommand{\GSTRwithsmallDaim}{\textsc{RST with Small Diameter}}
\newcommand{\recgraph}{\mathcal{G}}
\newcommand{\newweight}[1]{\ell(#1)}

\newcommand{\subT}[2]{#1_{#2}}
\newcommand{\weight}[1]{L(#1)}

\newcommand{\largedeg}[1]{{\sf large}({#1}) }

\newcommand{\set}[1]{\left\{#1\right\}}
\newcommand{\inset}[2]{\left\{#1 \;\middle|\; #2\right\}}
\newcommand{\len}[1]{\left|#1\right|}
\newcommand{\size}[1]{\len{#1}}
\newcommand{\tuple}[1]{\left(#1\right)}

\usepackage {environ}

\makeatletter
\newcommand{\problemtitle}[1]{\gdef\@problemtitle{#1}}
\newcommand{\probleminput}[1]{\gdef\@probleminput{#1}}
\newcommand{\problemoutput}[1]{\gdef\@problemoutput{#1}}
\NewEnviron{problem}{
  \problemtitle{}\probleminput{}\problemoutput{}
  \BODY
  \par\addvspace{.5\baselineskip}
  \noindent{
    \framebox[\textwidth][c]{
        \begin{tabularx}{\textwidth}{@{\hspace{\parindent}} l X c}
            \multicolumn{2}{@{\hspace{\parindent}}l}{\textsc{\@problemtitle}} \\
            \textbf{Input:} & \@probleminput \\
            \textbf{Question:} & \@problemoutput
        \end{tabularx}
        }
      }
  \par\addvspace{.5\baselineskip}
}
\makeatother

\hideLIPIcs

\title{Reconfiguration of Spanning Trees with Degree Constraint or Diameter Constraint}

\titlerunning{Spanning tree reconfiguration} 

\author{Nicolas Bousquet}{CNRS, LIRIS, Universit\'e de Lyon, France
}{nicolas.bousquet@liris.cnrs.fr}{https://orcid.org/0000-0003-0170-0503}{This work was supported by ANR project GrR (ANR-18-CE40-0032).}

\author{Takehiro Ito}{Graduate School of Information Sciences, Tohoku University, Japan}{takehiro@tohoku.ac.jp}{https://orcid.org/0000-0002-9912-6898}{Partially supported by JSPS KAKENHI Grant Numbers JP18H04091, JP19K11814 and JP20H05793, Japan.}

\author{Yusuke Kobayashi}{Research Institute for Mathematical Sciences, Kyoto University, Japan}{yusuke@kurims.kyoto-u.ac.jp}{https://orcid.org/0000-0001-9478-7307}{Partially supported by JSPS KAKENHI Grant Numbers 18H05291, JP20K11692, and 20H05795, Japan.}

\author{Haruka Mizuta}{Graduate School of Information Sciences, Tohoku University, Japan}{haruka.mizuta.s4@dc.tohoku.ac.jp}{}{}

\author{Paul Ouvrard}{Universit\'e de Bordeaux, France
}{paul.ouvrard@u-bordeaux.fr}{}{This work was supported by ANR project GrR
(ANR-18-CE40-0032).}

\author{Akira Suzuki}{Graduate School of Information Sciences, Tohoku University, Japan}{akira@tohoku.ac.jp}{https://orcid.org/0000-0002-5212-0202}{Partially supported by JSPS KAKENHI Grant Numbers JP18H04091, JP20K11666 and JP20H05794, Japan.}

\author{Kunihiro Wasa}{Toyohashi University of Technology, Japan}{wasa@cs.tut.ac.jp}{https://orcid.org/0000-0001-9822-6283}{Partially supported by JST CREST Grant Numbers JPMJCR18K3 and JPMJCR1401, and JSPS KAKENHI Grant Numbers 19K20350 and JP20H05793, Japan.}

\authorrunning{N.\,Bousquet et al.} 

\Copyright{Nicolas Bousquet, Takehiro Ito, Yusuke Kobayashi, Haruka Mizuta, Paul Ouvrard, Akira Suzuki, and Kunihiro Wasa} 


\category{} 

\relatedversion{} 

\supplement{}

\funding{This work is partially supported by JSPS and MEAE-MESRI under the Japan-France Integrated Action Program (SAKURA).}

\nolinenumbers 

\EventEditors{John Q. Open and Joan R. Access}
\EventNoEds{2}
\EventLongTitle{42nd Conference on Very Important Topics (CVIT 2016)}
\EventShortTitle{CVIT 2016}
\EventAcronym{CVIT}
\EventYear{2016}
\EventDate{December 24--27, 2016}
\EventLocation{Little Whinging, United Kingdom}
\EventLogo{}
\SeriesVolume{42}
\ArticleNo{23}

\begin{document}

\maketitle

\begin{abstract}
We investigate the complexity of finding a transformation from a given spanning tree in a graph to another given spanning tree in the same graph via a sequence of edge flips. The exchange property of the matroid bases immediately yields that such a transformation always exists if we have no constraints on spanning trees. In this paper, we wish to find a transformation which passes through only spanning trees satisfying some constraint. Our focus is bounding either the maximum degree or the diameter of spanning trees, and we give the following results. The problem with a lower bound on maximum degree is solvable in polynomial time, while the problem with an upper bound on maximum degree is PSPACE-complete. The problem with a lower bound on diameter is NP-hard, while the problem with an upper bound on diameter is solvable in polynomial time. 
\end{abstract}

\newpage

\section{Introduction}

Given an instance of some combinatorial search problem and two of its
feasible solutions, a \emph{reconfiguration problem} asks whether one solution can be transformed into the other in a step-by-step fashion, such that each intermediate solution is also feasible.
Reconfiguration problems capture dynamic situations, where some
solution is in place and we would like to move to a desired alternative
solution without becoming infeasible. A systematic study of the complexity of reconfiguration problems was initiated in~\cite{Ito11}. Recently the topic has gained a lot of attention in the context of CSP and graph problems, such as the independent set problem, the matching problem, and the dominating set problem. 
For an overview of recent results on reconfiguration problems, the reader is referred to the surveys of van den Heuvel~\cite{vHeuvel13} and Nishimura~\cite{Nishimura18}.

In this paper, our reference problem is the spanning tree problem. Let $G=(V,E)$ be a connected graph on $n$ vertices. A \emph{spanning tree} of $G$ is a subgraph of $G$ which is a tree (connected acyclic subgraph) and includes all the vertices in $G$. 
Spanning trees naturally arise in various situations such as routing or discrete geometry. 
In order to define a valid step-by-step transformation, an adjacency relation on the set of feasible solutions is needed. 
 Let $T_1$ and $T_2$ be two spanning trees of $G$. 
 We say that $T_1$ and $T_2$ are \emph{adjacent} by an \emph{edge flip} if there exist $e_1 \in E(T_1)$ and $e_2 \in E(T_2)$ such that $E(T_2) = (E(T_1) \setminus \{e_1\} ) \cup \{e_2\}$. 
For two spanning trees $T_\ini$ and $T_\tar$ of $G$, a \emph{reconfiguration sequence} (or simply a \emph{transformation}) from $T_\ini$ to $T_\tar$ is a sequence of spanning trees $\langle T_0 :=T_\ini,T_1,\ldots,T_\ell:=T_\tar \rangle$ such that two consecutive spanning trees are adjacent. 
Ito et al.~\cite{Ito11} remarked that any spanning tree can be transformed into any other via a sequence of edge flips, which easily follows from the exchange property of the matroid bases.
 
 In practice, we often need that spanning trees satisfy some additional desirable properties. Even if finding a spanning tree can be done in polynomial time, the problem becomes often NP-complete when additional constraints are added.
 In this paper, we consider spanning tree reconfiguration with additional constraints.
 More formally, we study the following questions: 1) does a transformation always exist when we add some constraints on the spanning trees all along the transformation? 
2)  If not, is it possible to decide efficiently if such a transformation exists? This question was already studied for spanning trees with restrictions on the number of leaves~\cite{BousquetIKMOSW20} or vertex modification between Steiner trees~\cite{DBLP:conf/mfcs/MizutaHIZ19} for instance.
If the answer to the first question is positive, it means that we can sample uniformly at random constrained spanning trees via a simple Monte Carlo Markov Chain. When the answer is negative, we might still want to find a transformation if possible between a fixed pair of solutions, for instance for updating a routing protocol in a network step by step without breaking the network and not over-requesting nodes during the transformation. 

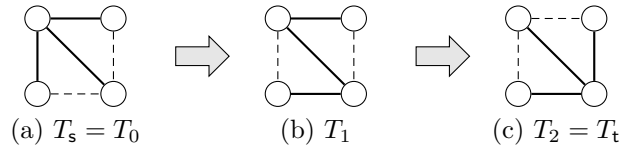
\begin{figure}[tb]
      \centering
    \begin{tikzpicture}

    \begin{scope}
    \node[circle, draw] (v1) at (0, 1) {}; 
    \node[circle, draw] (v2) at (1, 1) {}; 
    \node[circle, draw] (v3) at (0, 0) {}; 
    \node[circle, draw] (v4) at (1, 0) {}; 
    \node at (0.5,-0.5) {(a) $T_\ini = T_0$};
    \draw (v1) -- (v2) [thick];
    \draw (v1) -- (v3) [thick];
    \draw (v1) -- (v4) [thick];
    \draw (v2) -- (v4) [densely dashed];
    \draw (v3) -- (v4) [densely dashed];
    \end{scope}

    \begin{scope}[xshift=6em]
        \node[single arrow,draw=black,fill=black!10,minimum height=2em,single arrow head extend=0.4em] (arr) at (0, 0.5) {};
    \end{scope}

    \begin{scope}[xshift=9em]
    \node[circle, draw] (v1) at (0, 1) {}; 
    \node[circle, draw] (v2) at (1, 1) {}; 
    \node[circle, draw] (v3) at (0, 0) {}; 
    \node[circle, draw] (v4) at (1, 0) {}; 
    \draw (v1) -- (v2) [thick];
    \draw (v1) -- (v3) [densely dashed];
    \draw (v1) -- (v4) [thick];
    \draw (v2) -- (v4) [densely dashed];
    \draw (v3) -- (v4) [thick];
    \node at (0.5,-0.5) {(b) $T_1$};
    \end{scope}
    
    \begin{scope}[xshift=15em]
        \node[single arrow,draw=black,fill=black!10,minimum height=2em,single arrow head extend=0.4em] (arr) at (0, 0.5) {};
    \end{scope}

    \begin{scope}[xshift=18em]
    \node[circle, draw] (v1) at (0, 1) {}; 
    \node[circle, draw] (v2) at (1, 1) {}; 
    \node[circle, draw] (v3) at (0, 0) {}; 
    \node[circle, draw] (v4) at (1, 0) {}; 
    \draw (v1) -- (v2) [densely dashed];
    \draw (v1) -- (v3) [densely dashed];
    \draw (v1) -- (v4) [thick];
    \draw (v2) -- (v4) [thick];
    \draw (v3) -- (v4) [thick];
    \node at (0.5,-0.5) {(c) $T_2 = T_\tar$};
    \end{scope}

\end{tikzpicture}%

      \caption{A reconfiguration sequence from $T_\ini$ to $T_\tar$ (with no constraint on spanning trees). 
      There is no reconfiguration sequence from $T_\ini$ to $T_\tar$ if we restrict spanning trees either with maximum degree at least three or with diameter at most two.}
      \label{fig:example}
\end{figure}

In this paper, we study {\sc Reconfiguration of Spanning Trees (RST)} with degree constraints or with diameter constraints (See \figurename~\ref{fig:example}.)
We first describe the problem with degree constraints. 
\begin{problem}
\problemtitle{RST with Small (resp.~Large) Maximum Degree}
\probleminput{A graph $G$, a positive integer $d$, and two spanning trees $\spt_\ini$ and $\spt_\tar$ in $G$ with maximum degree at most (resp.~at least) $d$. }
\problemoutput{Is there a reconfiguration sequence from $\spt_\ini$ to $\spt_\tar$ such that any spanning tree in the sequence is of maximum degree at most (resp.~at least) $d$?}
\end{problem}

\noindent
Bounding the maximum degree of spanning trees has applications for routing problems 
when we send data (i.e., a flow) along a spanning tree in a communication network. 
In this setting, the degree of a node is a measure of its load, and hence it is natural to bound the maximum degree in the spanning tree. 
In a complex dynamic networks, we want to reconfigure spanning trees on the fly to keep this property on the dynamic setting, which motivates us to study the reconfiguration problem. 

The problem of finding a spanning tree with degree bounds is studied also from the theoretical point of view.  
Notice that spanning trees with bounds on the maximum degree include Hamiltonian paths that are spanning trees of maximum degree two.
This implies that finding a spanning tree with maximum degree at most $d$ is NP-hard. 
For restricted graph classes, this search problem is investigated in~\cite{Czumaj:Strothmann:ESA:1997}. 
It is shown in~\cite{FURER1994409} that if we relax the degree bound by one, then the search problem can be solved in polynomial time. 
Its optimization variants are also studied in \cite{DBLP:conf/focs/Goemans06,10.1145/2629366}. 

We also study the problem with diameter constraints, which is formally stated as follows.  

\begin{problem}
\problemtitle{RST with Small (resp.~Large) Diameter}
\probleminput{A graph $G$, a positive integer $d$, and two spanning trees $\spt_\ini$ and $\spt_\tar$ in $G$ with diameter at most (resp.~at least) $d$. }
\problemoutput{Is there a reconfiguration sequence from $\spt_\ini$ to $\spt_\tar$ such that any spanning tree in the sequence is of diameter at most (resp.~at least) $d$?}
\end{problem}

\noindent
Spanning trees with largest possible diameter are Hamiltonian paths which receive a considerable attention.
Spanning trees with upper bound on the diameter are for instance desirable in high-speed networks like optical networks since they minimize the worst-case propagation delay to all the nodes of the graphs, see e.g.~\cite{ItalianoR98}. 
We can find a spanning tree with minimum diameter in polynomial time~\cite{10.1016/0020-0190(94)00183-Y}, and 
some related problems have been studied in the literature~\cite{HASSIN2004343,DBLP:journals/algorithmica/SpriggsKBSS04}



The problem of updating minimum spanning trees to maintain a valid spanning tree in dynamic networks is an important problem that received a considerable attention in the last decades, see for instance~\cite{BarjonCCJN14,HuangFW15}.
In this situation, the graph is dynamic and is dynamically updated at each time step. The solution at time $t$, which might not be a solution anymore at time $t+1$ (e.g. if edges of the spanning has been deleted from the graph), has to be modified with as few modifications as possible into a valid solution as good as possible. Spanning tree reconfiguration lies between the static situation (since the graph is fixed) and the dynamic situation (since the solution has to be modified).

\subsection*{Our Results}

    The contribution of this paper is to study the computational complexity of \textsc{RST with Small ({\rm or} Large) Maximum Degree} and {\sc RST with Small ({\rm or} Large) Diameter}. 
    
	\begin{restatable}{theorem}{lmd}
		\label{thm:lmd_P}
		{\sc RST with Large Maximum Degree} can be decided in polynomial time.
	\end{restatable}

\noindent
    Our proof for \cref{thm:lmd_P} is in two steps. First we show that if there exists a vertex that has degree at least $d$ in both $T_\ini$ and $T_\tar$, then there is a reconfiguration sequence between them. Then, for two vertices $u$ and $v$, we prove that we can decide in polynomial time if there exists a pair of adjacent spanning trees $T$ and $T'$ such that $u$ has degree at least $d$ in $T$ and $v$ has degree at least $d$ in $T'$. These results together will imply \cref{thm:lmd_P}.

     
    While the existence of a spanning tree with maximum degree at least $d$ can be decided in polynomial time, it is NP-complete to find a spanning tree of maximum degree at most $2$ (that is a Hamiltonian path). 
    A similar behavior holds for RST with degree constraints. 
	\begin{restatable}{theorem}{smd}
		\label{thm:smd_hard}
		For every $d \ge 3$, {\sc RST with Small Maximum Degree} is PSPACE-complete.
	\end{restatable}
	
	\noindent
	The proof for \cref{thm:smd_hard} consists of a reduction from NCL (Nondeterministic Constraint Logic), known to be PSPACE-complete~\cite{HD05}. 
	This result is tight in the following sense: if at least one of $T_\ini$ and $T_\tar$ has maximum degree at most $d-1$, then the problem becomes polynomial-time solvable (shown in Theorem~\ref{thm:smd_d-1}).
    It is worth noting that this behavior is similar to the result for the search problem shown in~\cite{FURER1994409}; 
    while finding a spanning tree with maximum degree at most $d$ is NP-hard, 
    if we relax the degree bound by one, then the problem can be solved in polynomial time. 

	\smallskip

In the second part of the paper, we study {\sc RST with Small {\rm or} Large Diameter}. 
	\begin{restatable}{theorem}{ldiam}
		\label{thm:ldiam_hard}
		{\sc RST with Large Diameter} is NP-hard even restricted to planar graphs.
	\end{restatable}

\noindent
    The proof for \cref{thm:ldiam_hard} consists of a reduction from the {\sc Hamiltonian Path} problem, which is not a reconfiguration problem but the original search problem. 
    We note that since the length of a reconfiguration sequence is not necessarily bounded by a polynomial in the input size, it is unclear whether {\sc RST with Large Diameter} belongs to the class NP.
    In a similar way to {\sc RST with Small Maximum Degree}, we conjecture that {\sc RST with Large Diameter} is PSPACE-complete. 

    Finally, the main technical result of the paper is the following positive result.
	\begin{restatable}{theorem}{sdiam}
		\label{thm:sdiam_P}
		{\sc RST with Small Diameter} is polynomial-time solvable. 
	\end{restatable}

\noindent

The proof for \cref{thm:sdiam_P} follows a similar scheme to \cref{thm:lmd_P}. First we show that all the spanning trees with the same ``center'' can be transformed into any other. Therefore, it suffices to consider the transformation of the centers. However, for two vertices $u$ and $v$, it is hard to determine whether there exists a pair of adjacent spanning trees $T$ and $T'$ such that $u$ and $v$ are centers of $T$ and $T'$, respectively. Indeed, we do not know whether it can be done in polynomial time.  
The core of the proof is to focus on only ``good'' pairs of centers for which the existence of a desired pair of spanning trees can be tested in polynomial time (see \cref{prop:reconf:good}). 
A key ingredient of our proof consists in proving that if there is a reconfiguration sequence between the spanning trees, then there exists a sequence of centers from the initial center to the final center in which any consecutive centers form a good pair (see \cref{thm:good:sequence}).




\subsection*{Organization}
The rest of this paper is organized as follows. 
We first give some preliminaries in Section~\ref{sec:preliminaries}. 
Next, Sections~\ref{sec:lmd} and \ref{sec:small_deg} are devoted to {\sc RST with Large Maximum Degree} (Theorem~\ref{thm:lmd_P}) and {\sc RST with Small Maximum Degree} (Theorem~\ref{thm:smd_hard}), respectively. 
Then, Sections~\ref{sec:diam_large} and \ref{sec:small_diam} are devoted to {\sc RST with Large Diameter} (Theorem~\ref{thm:ldiam_hard}) and {\sc RST with Small Diameter} (Theorem~\ref{thm:sdiam_P}), respectively. 
Some technical parts in the proof of Theorem~\ref{thm:sdiam_P} are deferred to Sections~\ref{sec:findgoodtriple} and \ref{sec:goodsequence}.
Finally, we conclude this paper by giving some remarks in Section~\ref{sec:conclusion}. 


\section{Preliminaries}
\label{sec:preliminaries}

Throughout this paper, we consider graphs that are simple and loopless.
Let $G = (V, E)$ be a graph.   
For a vertex $v \in V$, we denote by $\dg{G}{v}$ the \emph{degree} of $v$ in $G$, by $\neig{G}{v}$ the (open) \emph{neighborhood} of $v$ in $G$, and by $\inc{G}{v}$ the set of edges incident to $v$ in $G$.
Since $G$ is simple, $\dg{G}{v} = |\neig{G}{v}| = |\inc{G}{v}|$. 
For a tree $T$, a vertex $v$ is a \emph{leaf} if its degree is one, and is an \emph{internal node} otherwise. 
A \emph{branching node} is a vertex of degree at least three. 

For a subgraph $H$ of $G$ and an $F \subseteq E$, we denote by
$H - F$ the graph $(V(H), E(H) \setminus F)$ and by $H + F$ the graph $(V(H), E(H) \cup F)$. 
To avoid cumbersome notation, if $e \in E$, $H -\{e\}$ and $H + \{e\}$ will be denoted by $H-e$ and $H+e$, respectively. 

For $u, v \in V$, the \emph{distance} $\distorig{u}{v}{G}$ between $u$ and $v$ is defined as the minimum number of edges in a shortest $u$-$v$ path. 
For $v \in V$, the \emph{eccentricity} $\ecc{v}{G}$ of $v$ in $G$ is the maximum distance between $v$ and any vertex in $G$, that is, $\ecc{v}{G} := \max\inset{\distorig{v}{u}{G}}{u \in V}$. 
The \emph{diameter} $\diam{G}$ of $G$ is the maximum eccentricity among $V$. 
That is, $\diam{G} := \max\inset{\ecc{v}{G}}{v \in V} = \max\inset{\distorig{u}{v}{G}}{u, v \in V}$.

For two spanning trees $T$ and $T'$, we denote $T \leftrightarrow T'$ if $|E(T) \setminus E(T')|=|E(T') \setminus E(T)|\le 1$, that is, either $T=T'$ or $T$ and $T'$ are adjacent.  
We say that $\spt_\ini$ is {\em reconfigurable} to $\spt_\tar$ if there exists a reconfiguration sequence from $\spt_\ini$ to $\spt_\tar$ such that any spanning tree in the sequence satisfies a given degree/diameter constraint. 
When we have no degree/diameter constraints,  
since spanning trees form a base family of a matroid, 
the exchange property of the matroid bases ensures that 
there always exists a reconfiguration sequence between any pair of spanning trees. 
		\begin{lemma}[see Proposition~1 in \cite{Ito11}]\label{lem:known}
			Let $G$ be a graph and $\spt$ and $\spt^\prime$ be  two spanning trees of $G$. There exists a reconfiguration sequence $\langle \spt = \spt_0,\spt_1,\ldots,\spt_\ell=\spt^\prime \rangle$ between $\spt$ and $\spt^\prime$ such that for all $i \in \{ 0,1,\ldots,\ell \}$, the spanning tree $\spt_i$ contains all the edges in $E(\spt) \cap E(\spt^\prime)$.
		\end{lemma}


\section{Large Maximum Degree (Proof of Theorem \ref{thm:lmd_P})}
\label{sec:lmd}

In this section, we prove \cref{thm:lmd_P}, which we restate here. 

\lmd*

Let $(G,\thr,\spt_\ini,\spt_\tar)$ be an instance of {\sc RST with Large Maximum Degree}.
For a spanning tree $\spt$ in $G$, let $\largedeg{\spt} \subseteq V$ be the 
set of all the vertices of degree at least $d$ in $\spt$, that is, $\largedeg{\spt} := \{v\in V \mid  \dg{T}{v} \ge d\}$. 
Note that $\spt$ has maximum degree at least $\thr$ if and only if $\largedeg{\spt} \not= \emptyset$. 
The following lemma is easy but is essential to prove \cref{thm:lmd_P}. 

\begin{lemma}\label{lmd_same}
      Let $T_1$ and $T_2$ be spanning trees in $G$ with maximum degree at least $d$.
      If there exists a vertex $u \in \largedeg{T_1} \cap \largedeg{T_2}$, then $T_1$ is reconfigurable to $T_2$. 
\end{lemma}
\begin{proof}
We show that $T_1$ is reconfigurable to $T_2$ by induction on $\thr - |\inc{T_1}{u} \cap \inc{T_2}{u}|$. 
		
Suppose that $\thr - |\inc{T_1}{u} \cap \inc{T_2}{u}| \le 0$ holds. 
By \cref{lem:known}, there exists a reconfiguration sequence from $T_1$ to $T_2$ in which all the spanning trees contain $\inc{T_1}{u} \cap \inc{T_2}{u}$. 
This shows that, for any spanning tree $T'$ in the sequence, $|\inc{T'}{u}| \ge |\inc{T_1}{u} \cap \inc{T_2}{u}| \ge \thr$. 
Hence, $T_1$ is reconfigurable to $T_2$. 
 
Suppose that $\thr - |\inc{T_1}{u} \cap \inc{T_2}{u}| \ge 1$ holds. 
Since $|\inc{T_2}{u}| \ge \thr$ and $|\inc{T_1}{u} \cap \inc{T_2}{u}| \le d-1$, 
there exists an edge $e \in \inc{T_2}{u} \setminus \inc{T_1}{u}$. 
Since $T_1+e$ contains a unique cycle $C$ and $T_2$ contains no cycle, there exists an edge $f \in E(C) \setminus E(T_2)$. 
Then, we have that 
$f \in E(T_1) \setminus E(T_2)$ and 
$T'_1 := T_1 + e - f$ is a spanning tree in $G$. 
Observe that $|\inc{T'_1}{u}| \ge |\inc{T_1}{u} \cup \{e\}| - 1 \ge  |\inc{T_1}{u}| \ge \thr$, 
which shows that $u \in \largedeg{T'_1}$. 
We also see that 
$\thr - |\inc{T'_1}{u} \cap \inc{T_2}{u}| = \thr - |\inc{T_1}{u} \cap \inc{T_2}{u}| - 1$. 
Therefore, by the induction hypothesis, $T'_1$ is reconfigurable to $T_2$. 
This shows that $T_1$ is reconfigurable to $T_2$ as $T_1$ and $T'_1$ are adjacent. 
\end{proof}

Our algorithm is based on testing the reachability in an auxiliary graph $\recgraph$, which is defined as follows.
The vertex set of $\recgraph$ is defined as $V$, 
where each vertex $v$ in $V(\recgraph)$ corresponds to the set of spanning trees $T$ with $v \in \largedeg{T}$. 
For any pair $u, v$ of distinct vertices in $V(\recgraph)$,
there is an edge $uv \in E(\recgraph)$ if and only if
there exist spanning trees $\spt$ and $\spt'$ such that 
$u \in \largedeg{\spt}$, $v \in \largedeg{\spt'}$, and
$T \leftrightarrow T'$ (possibly $T=T'$).
Then, by definition of the auxiliary graph and Lemma~\ref{lmd_same}, we have the following lemma.

\begin{lemma}
\label{lem:lmd_reconf}
Let $\spt_\ini$ and $\spt_\tar$ be spanning trees with maximum degree at least $d$. 
Then, $T_\ini$ is reconfigurable to $T_\tar$ if and only if $\recgraph$ contains a path from $\largedeg{\spt_\ini}$ to $\largedeg{\spt_\tar}$. 
\end{lemma}

\begin{proof}
    We first show the ``only if'' part. 
    Suppose that there exists a reconfiguration sequence $\langle T_\ini = T_0, T_1, \dots, T_k = T_\tar \rangle$ from $T_\ini$ to $T_\tar$, where 
    $T_i$ is a spanning tree of maximum degree at least $d$ for any $i \in \{0, 1, \dots , k\}$ and
    $T_i$ and $T_{i+1}$ are adjacent for any $i  \in \{0, 1, \dots, k-1\}$.
    For each $i$, let $v_i$ be a vertex in $\largedeg{T_i}$. 
    By the definition of $\recgraph$, for $i \in \{0, 1, \dots, k-1\}$, we have either $v_i = v_{i+1}$ or $\recgraph$ contains an edge $v_i v_{i+1}$. 
    Since $v_0 \in \largedeg{T_\ini}$ and $v_{k} \in \largedeg{T_\tar}$, $\recgraph$ contains a path  
    from $\largedeg{T_\ini}$ to $\largedeg{T_\tar}$. 
    
    To show the ``if'' part, suppose that $\recgraph$ contains a path $(v_0, v_1, \dots , v_k)$ from $\largedeg{T_\ini}$ to $\largedeg{T_\tar}$. 
    For $i \in \{0, 1, \dots, k-1\}$, $v_i v_{i+1} \in E(\recgraph)$ means that 
    there exist two spanning trees $T^+_i$ and $T^-_{i+1}$ such that $v_i \in \largedeg{T^+_i}$, $v_{i+1} \in \largedeg{T^-_{i+1}}$, and $T^+_i \leftrightarrow T^-_{i+1}$. 
    Let $T^-_0 := T_\ini$ and $T^+_k := T_\tar$. Then, for $i \in \{0, 1, \dots, k\}$, since $v_i \in \largedeg{T^-_i} \cap \largedeg{T^+_i}$, 
    $T^-_i$ is reconfigurable to $T^+_i$ by \cref{lmd_same}. This together with $T^+_i \leftrightarrow T^-_{i+1}$ shows that $T_\ini$ is reconfigurable to $T_\tar$. 
\end{proof}

By this lemma, we can solve {\sc RST with Large Maximum Degree} 
by detecting a path from $\largedeg{\spt_\ini}$ to $\largedeg{\spt_\tar}$ in $\recgraph$
(see \cref{alg:lmd} for a pseudocode of our algorithm).
Our remaining task is 
to construct the auxiliary graph $\recgraph$ in polynomial time which is possible by the following lemma. 

\begin{algorithm}[t]
    \KwInput{A graph $G$ and two spanning trees $T_\ini$ and $T_\tar$ in $G$ with max.~degree $\ge d$.}
  \KwOutput{ Is $T_\ini$ reconfigurable to $T_\tar$?} 
Compute  $\largedeg{T_\ini}$ and $\largedeg{T_\tar}$, and construct $\recgraph$\; 
    \lIf{there is a path between $\largedeg{T_\ini}$ and $\largedeg{T_\tar}$ in $\recgraph$}{\Return{\yes}}
    \lElse{\Return{\no}}
  \caption{Algorithm for {\sc RST with Large Maximum Degree}}\label{alg:lmd}
\end{algorithm}

\begin{lemma}
\label{lem:lmd_edge}
     For two distinct vertices $u, v \in V$, there exists an edge $uv \in E(\recgraph)$ if and only if $|\neig{G}{u}| \ge d$, $|\neig{G}{v}| \ge d$, and 
	\begin{equation}\label{eq:lmd_01}
		|\neig{G}{u} \cup \neig{G}{v}| \geq
		\begin{cases}
			2\thr - 1 & \mbox{if $uv \in E(G)$,} \\
			2\thr - 2 & \mbox{otherwise}.
 		\end{cases}
	\end{equation}
\end{lemma}
\begin{proof}
We first prove the ``only-if'' direction.
Suppose that $\recgraph$ contains an edge $uv$, that is, 
there exist spanning trees $\spt$ and $\spt'$ such that $u \in \largedeg{\spt}$, $v \in \largedeg{\spt'}$, and
$T \leftrightarrow T'$ (possibly $T=T'$). 
Then, $|\neig{G}{u}| \ge d$ and $|\neig{G}{v}| \ge d$ are obvious. 
Since $T$ contains no cycle, we know that $\neig{T}{u}$ and $\neig{T}{v}$ contain at most one common vertex. 
Then, we obtain
		\begin{eqnarray}
			|\neig{G}{u} \cup \neig{G}{v}| & \geq & |\neig{T}{u} \cup \neig{T}{v}|  \notag \\
			& \geq & |\neig{\spt}{u}| + |\neig{\spt}{v}| - 1         \qquad \qquad \ \  \mbox{(by $|\neig{\spt}{u} \cap \neig{\spt}{v}| \le 1$)}  \notag \\
			& \geq & |\neig{\spt}{u}| + (|\neig{\spt'}{v}| -1) - 1  \qquad \mbox{(by $|E(T') \setminus E(T)| \le 1$)} \notag \\
			& \geq & 2d - 2. \label{eq:lmd_02}
		\end{eqnarray}
Similarly, if $uv \in E(G) \setminus E(T)$, then we obtain
\begin{equation}
\label{eq:lmd_03}
|\neig{G}{u} \cup \neig{G}{v}| \geq  |\neig{T}{u} \cup \neig{T}{v} \cup \{u, v\}| \geq  |\neig{\spt}{u}| + |\neig{\spt}{v}| + 1 \ge 2d. 
\end{equation}
If $uv \in E(T)$, then $\neig{T}{u} \cap \neig{T}{v} = \emptyset$ holds, and hence 
we obtain
\begin{equation}
\label{eq:lmd_04}
|\neig{G}{u} \cup \neig{G}{v}| \geq  |\neig{T}{u} \cup \neig{T}{v}| =  |\neig{\spt}{u}| + |\neig{\spt}{v}| \ge 2d - 1. 
\end{equation}
By (\ref{eq:lmd_02}), (\ref{eq:lmd_03}), and (\ref{eq:lmd_04}), we obtain (\ref{eq:lmd_01}). 

\begin{figure}[t!]
      \newcommand{\picscale}{.7}
      \centering
      \begin{minipage}[t]{0.45\hsize}
            \centering
            \scalebox{\picscale}{%
    \begin{tikzpicture}[mynode/.style = {circle, fill=black, inner sep=0pt, minimum size=5pt }]
    \node[mynode, label={90:$u$}] (u) at (5.2, 1) {}; 
    \node[mynode, label={-90:$v$}]  (v) at (4.2, -1) {}; 
    \foreach \x in {0, 1, ..., 7}
        {
            \node[mynode] (v\x) at (\x, 0) {};
        }

    \foreach \x in {3, 4, ..., 7}
    {
        \draw (v\x) -- (u); 
    }

    \foreach \x in {0, 1, ..., 5}
    {
        \draw (v\x) -- (v); 
    }

    \draw[blue] \convexpath{v7,v3}{6pt}; 
    \draw[green] \convexpath{v0,v5}{8pt}; 
\end{tikzpicture}%
}
            \caption{Case when $uv \not\in E(G)$.}
            \label{fig:301}
      \end{minipage}
      \begin{minipage}[t]{0.45\hsize}
            \centering
            \scalebox{\picscale}{%
    \begin{tikzpicture}[mynode/.style = {circle, fill=black, inner sep=0pt, minimum size=5pt }]

    \node[mynode] (u1) at (0, 0) {}; 
    \node[mynode] (u2) at (0.5, 1) {}; 
    \node[mynode] (u3) at (1, 1.5) {}; 
    \node[mynode, label={2:$v$}] (v) at (1.5, 1) {}; 
    \node[mynode] (u4) at (1, 0) {}; 

    \node (u1i) at (0.25, 0.15) {}; 
    \node (u2i) at (0.7, 1) {}; 
    \node (u3i) at (1, 1.3) {}; 
    \node (vi) at (1.3, 1) {}; 
    \node (u4i) at (0.85, 0.15) {}; 

    \node[text=blue] (c) at (0.8, 0.7) {$C$}; 

    \draw (u1) -- (u2);
    \draw[red] (u2) -- node[pos=0.8, left, text=red]{$e'$} (u3);
    \draw (u3) -- (v);
    \draw[dashed, red] (v) -- node[midway, right, text=red]{$e$} (u4);
    \draw (u4) -- (u1);

    \draw (u1) -- +(-90:2em); 
    \draw (u1) -- +(-135:2em); 
    \draw (u2) -- +(170:2em); 
    \draw (u3) -- +(160:2em); 
    \draw (v) -- +(40:2em); 
    \draw (v) -- +(-20:2em); 
    \draw (u4) -- +(-100:2em); 

    \draw[blue] \convexpath{u1i, u2i, u3i, vi, u4i}{1pt}; 
\end{tikzpicture}%
}
            \caption{Cycle $C$ and edge $e'$.}
            \label{fig:302}
      \end{minipage}
\end{figure}

We next prove the ``if'' direction.
Suppose that $|\neig{G}{u}| \ge d$, $|\neig{G}{v}| \ge d$, and (\ref{eq:lmd_01}) hold. 
For each of the following two cases, we define an edge set $F \subseteq E$. 
\begin{itemize}
\item
Suppose that $uv \not\in E(G)$ holds (Figure~\ref{fig:301}). 
Let $S_u \subseteq \neig{G}{u}$ be a vertex set with $|S_u| = d$ that maximizes $|S_u \setminus \neig{G}{v}|$. 
Then, we have either $S_u \subseteq \neig{G}{u} \setminus \neig{G}{v}$ or $S_u \supsetneq \neig{G}{u} \setminus \neig{G}{v}$. 
If $S_u \subseteq \neig{G}{u} \setminus \neig{G}{v}$, then let $S_v \subseteq \neig{G}{v}$ be a vertex set with $|S_v| =d-1$. 
Otherwise, let $S_v \subseteq \neig{G}{v}$ be a vertex set such that $|S_v| =d-1$ and $|S_u \cap S_v|  = 1$, 
where such $S_v$ exists because $|\neig{G}{v} \setminus S_u| = |(\neig{G}{u} \cup \neig{G}{v}) \setminus S_u| \ge d -2$ and $|\neig{G}{v} \cap S_u| \ge 1$. 
In either case, we obtain $S_u \subseteq \neig{G}{u}$ and $S_v \subseteq \neig{G}{v}$ 
such that $|S_u| = d$, $|S_v| = d-1$, and $|S_u \cap S_v|  \le 1$. 
Define $F := \{ u w \mid w \in S_u \} \cup \{v w \mid w \in S_v\}$.

\item
Suppose that $uv \in E(G)$ holds. 
Since $|\neig{G}{u} \setminus \{v\}| \ge d-1$, $|\neig{G}{v} \setminus \{u\}| \ge d-1$, and 
$|(\neig{G}{u} \setminus \{v\}) \cup (\neig{G}{v} \setminus \{u\})| \ge 2d-3$, 
by the same argument as above, we can take 
$S_u \subseteq \neig{G}{u} \setminus \{v\}$ and $S_v \subseteq \neig{G}{v} \setminus \{u\}$   
such that $|S_u| = d-1$, $|S_v| = d-2$, and $S_u \cap S_v = \emptyset$. 
Define $F := \{ u w \mid w \in S_u \} \cup \{v w \mid w \in S_v\} \cup \{uv\}$. 
\end{itemize}
In both cases, it holds that $|F \cap \inc{G}{u}| = d$, $|F \cap \inc{G}{v}| = d-1$, and $F$ contains no cycle. 
Therefore, there exists a spanning tree $T$ with $E(T) \supseteq F$ 
such that 
$|\inc{T}{u}| \ge |F \cap \inc{G}{u}| = d$ and $|\inc{T}{v}| \ge |F \cap \inc{G}{v}| = d-1$. 
If $|\inc{T}{v}| \ge d$, then we obtain $\{u, v\} \subseteq \largedeg{T}$, which shows that $uv \in E(\recgraph)$. 
Therefore, it suffices to consider the case when $|\inc{T}{v}| = d-1$.
Since $|\inc{G}{v}| \ge d$, there exists an edge $e \in \inc{G}{v} \setminus \inc{T}{v}$. 
Let $C$ be the unique cycle in $T + e$ and $e'$ be an edge in $E(C) \setminus \inc{T}{v}$ (see Figure~\ref{fig:302}). 
Then, $T' := T + e - e'$ is a spanning tree such that $|\inc{T'}{v}| = |\inc{T}{v} \cup \{e\}| = d$, which means that $v \in \largedeg{T'}$. 
Since $T$ and $T'$ are adjacent, we obtain $uv \in E(\recgraph)$. 
\end{proof}

Since we can easily check the inequality (\ref{eq:lmd_01})
for each pair of vertices $u$ and $v$, 
Lemma~\ref{lem:lmd_edge} ensures that
the auxiliary graph $\recgraph$ can be constructed in polynomial time. 
Therefore, \cref{alg:lmd} correctly decides {\sc RST with Large Maximum Degree} in polynomial time, which completes the proof of Theorem~\ref{thm:lmd_P}.
Note that all the proofs are constructive, and hence we can find a desired reconfiguration sequence from $T_\ini$ to $T_\tar$ in polynomial time if it exists.

	
\section{Small Maximum Degree}\label{sec:small_deg}


In this section, we consider {\sc RST with Small Maximum Degree}. 
We first show the PSPACE-completeness in Section~\ref{sec:smd_hard}. 
In contrast, we show in Section~\ref{sec:smd_d-1} that 
if at least one of $T_\ini$ and $T_\tar$ has 
maximum degree at most $d-1$, then 
an instance $(G,\thr,\spt_\ini,\spt_\tar)$ of {\sc RST with Small Maximum Degree}
is a {\yes}-instance.

\subsection{PSPACE-Completeness (Proof of Theorem \ref{thm:smd_hard})}
\label{sec:smd_hard}

In this subsection, we prove \cref{thm:smd_hard}, i.e., we show that {\sc RST with Small Maximum Degree} is PSPACE-complete. 
The problem is indeed in PSPACE. We prove the PSPACE-hardness by giving a polynomial reduction from 
{\em Reconfiguration of Nondeterministic Constraint Logic on {\sc and}/{\sc or} graphs}, 
which we call {\sc NCL Reconfiguration} for short. 

Suppose that we are given a cubic graph with edge-weights such that each vertex is 
either incident to three weight-2 edges (``{\sc or} vertex'') or one weight-2 edge and two weight-1 edges (``{\sc and} vertex''), 
which we call an {\em {\sc and}/{\sc or} graph}. 
An {\em NCL configuration} is an orientation of the edges in the graph such that the total weights of incoming arcs at each vertex is at least two.  
Two NCL configurations are {\em adjacent} if they differ in a single edge direction. 
In {\sc NCL Reconfiguration}, we are given an {\sc and}/{\sc or} graph and its two NCL configurations, and the objective is to 
determine whether there exists a sequence of adjacent NCL configurations that transforms one into the other. 
It is shown in~\cite{HD05} that  {\sc NCL Reconfiguration} is PSPACE-complete. 
In what follows, we give a polynomial reduction from {\sc NCL Reconfiguration} to {\sc RST with Small Maximum Degree}.

\medskip
\noindent
\textbf{Construction of the graph. }
Suppose that we are given an instance of {\sc NCL Reconfiguration}, that is, an {\sc and}/{\sc or} graph $H=(V(H), E(H))$ and two configurations $\sigma_\ini$ and $\sigma_\tar$ of $H$. 
Fix $d \ge 3$. 
We first construct a graph $G'=(V', E')$, a vertex subset $L \subseteq V'$, and an integer $b(v) \in \{1, 2, 3\}$ for each $v \in V'$, and then 
construct a graph $G = (V, E)$  by using $G'$, $L$, and $b$.  
We consider an instance $(G,\thr,\spt_\ini,\spt_\tar)$ of {\sc RST with Small Maximum Degree}, 
where $\spt_\ini$ and $\spt_\tar$ will be defined later. 
The construction of $G'$, $L$, and $b$ is described as follows. 

\begin{itemize}
\item
We initialize $G'=(V',E')$ and $L$ as the empty graph and the empty set, respectively. 
\item
For a vertex $u \in V(H)$ and an edge $e \in \inc{H}{u}$, we introduce a vertex $v_{u, e}$ in $V'$. 
Let $b(v_{u,e})=2$. 
\item
For an edge $e \in E(H)$ connecting $u$ and $u'$, we introduce a vertex $v_{e}$ in $V'$ and two edges 
$v_e v_{u, e}$ and $v_e v_{u', e}$ in $E'$ (Figure~\ref{fig:402}). Let $b(v_e)= 1$. 
\item
For an {\sc or} vertex $u \in V(H)$ with $\inc{H}{u} = \{e_1, e_2, e_3\}$, we introduce a vertex $r_u$ in $V'$ and an edge 
$r_u v_{u, e_i}$ in $E'$ for $i \in \{1, 2, 3\}$ (Figure~\ref{fig:403}). Let $b(r_u)= 1$. 
Add $v_{u, e_i}$ to $L$ for $i \in \{1, 2, 3\}$.
\item
For an {\sc and} vertex $u \in V(H)$ with $\inc{H}{u} = \{e_0, e_1, e_2\}$, where $e_0$ is a weight-$2$ edge and $e_1$ and $e_2$ are weight-$1$ edges, 
we introduce four vertices $r_u, w_u, x_u$, and $y_u$ in $V'$, and seven edges 
$v_{u,e_0} r_u, r_u w_u, w_u x_u, w_u y_u, x_u v_{u,e_1}, y_u v_{u,e_2}$, and $v_{u,e_1} v_{u,e_2}$ in $E'$ (Figure~\ref{fig:404}).  
We denote by $E'_u$ the set of these seven edges.
Let $b(r_u)= 1$, $b(w_u)=3$, and $b(x_u)=b(y_u)=2$.
Add $v_{u, e_0}$ and $w_u$ to $L$.  
\end{itemize}

\begin{figure}[t!]
      \newcommand{\picscale}{.55}
      \centering
      \begin{minipage}{0.45\hsize}
          \centering
          \scalebox{\picscale}{%

\begin{tikzpicture}

    \begin{scope}
    \node[vertexSec4, label={270:$u$}] (u){};
    \node[vertexSec4, label={270:$u'$}, right=2.8em of u] (up){};

    \node[left=1em of u] (ulc) {}; 
    \node[above=1em of ulc] (ula) {}; 
    \node[below= 1em of ulc] (ulb) {}; 

    \node[right=1em of up] (uprc) {}; 
    \node[above=1em of uprc] (upra) {}; 
    \node[below= 1em of uprc] (uprb) {}; 

    \draw (u) -- (up) node[midway,above] {$e$}; 
    \draw (u) -- (ula); 
    \draw (u) -- (ulb); 
    \draw (up) -- (upra); 
    \draw (up) -- (uprb); 
    \end{scope}

    \begin{scope}[xshift=9em]
        \node[single arrow,draw=black,fill=black!10,minimum height=3em,single arrow head extend=0.4em] (arr) {};
    \end{scope}

    \begin{scope}[xshift=15em]
    \node[vertexSec4, label={90:$v_{u, e}$}] (vue){};
    \node[vertexSec4, label={90:$v_e$}, right=1.4em of vue] (ve){};
    \node[vertexSec4, label={90:$v_{u', e}$}, right=1.4em of ve] (vupe){};

    \draw (vue) -- (ve); 
    \draw (ve) -- (vupe); 

    \node[left = 1em of vue ] (vuec) {};
    \node[vertexSec4, above= 1em of vuec] (vuea) {};
    \node[vertexSec4, below= 1em of vuec] (vueb) {};

    \node[right= 1em of vupe ] (vupec) {};
    \node[vertexSec4, above= 1em of vupec] (vupea) {};
    \node[vertexSec4, below= 1em of vupec] (vupeb) {};

    \draw[dashed] \convexpath{vue,vueb,vuea}{1.5em}; 
    \draw[dashed] \convexpath{vupe,vupea,vupeb}{1.5em}; 
    \end{scope}
\end{tikzpicture}%
}
          \caption{The gadget for an edge.}
          \label{fig:402}
          
          \scalebox{\picscale}{%

\begin{tikzpicture}

    \begin{scope}
    \node[vertexSec4, label={180:$u$}] at (0, 0) (u){}; 

    \draw (u) -- +(150:3em)node[midway, above] {$e_1$};
    \draw (u) -- +(30:3em)node[midway, above] {$e_2$}; 
    \draw (u) -- +(270:3em)node[midway, right] {$e_3$}; 
    \end{scope}

    \begin{scope}[xshift=7em]
        \node[single arrow,draw=black,fill=black!10,minimum height=3em,single arrow head extend=0.4em] (arr) {};
    \end{scope}

    \begin{scope}[xshift=15em]

    \node[vertexSec4, label={335:$r_u$}] at(0, 0) (ru){$1$};

    \draw (ru) -- +(30:3em)  node[vertexSpSec4, label={90:$v_{u, e_1}$}] (ve1) {$2$}; 
    \draw (ru) -- +(150:3em) node[vertexSpSec4, label={90:$v_{u, e_2}$}] (ve2) {$2$} -- +(150:3em);
    \draw (ru) -- +(270:3em) node[vertexSpSec4, label={0:$v_{u, e_3}$}] (ve3)  {$2$} -- +(270:3em);
    \draw (ve1) -- +(30:2em);
    \draw (ve2) -- +(150:2em);
    \draw (ve3) -- +(270:2em);
    \end{scope}

\end{tikzpicture}%
}
          \caption{The gadget for an {\sc or} vertex.}
          \label{fig:403}
          
          \scalebox{\picscale}{%

\begin{tikzpicture}

    \begin{scope}
    \node[vertexSec4, label={180:$u$}] at (0, 0) (u){};

    \draw (u) -- +(150:3em)node[midway, above] {$e_1$};
    \draw (u) -- +(30:3em)node[midway, above] {$e_2$}; 
    \draw (u) -- +(270:3em)node[midway, right] {$e_0$}; 

    \end{scope}

    \begin{scope}[xshift=7em]
        \node[single arrow,draw=black,fill=black!10,minimum height=3em,single arrow head extend=0.4em] (arr) {};
    \end{scope}

    \begin{scope}[xshift=15em]

    \node (c) at (0, 0) {};
    \node (vc) at ([shift=({90:3.5em})]c) {}; 
    \node[vertexSec4, label={90:$v_{u, e_1}$}, left = 1em of vc] (vue1){$2$};
    \node[vertexSec4, label={90:$v_{u, e_2}$}, right = 1em of vc] (vue2){$2$};
    \node[vertexSec4, label={0:$x_u$}, below = 1em of vue1] (xu){$2$};
    \node[vertexSec4, label={0:$y_u$}, below = 1em of vue2] (yu){$2$};
    \node[vertexSpSec4, label={0:$w_u$}, below = 3em of vc]  (wu){$3$};
    \node[vertexSec4, label={0:$r_u$}, below = 1em of wu] (ru){$1$};
    \node[vertexSpSec4, label={0:$v_{u, e_0}$}, below = 1em of ru] (vue0){$2$};

    \draw (vue1) -- (vue2); 
    \draw (vue1) -- (xu); 
    \draw (vue2) -- (yu); 
    \draw (xu) -- (wu); 
    \draw (yu) -- (wu); 
    \draw (wu) -- (ru); 
    \draw (ru) -- (vue0); 

    \draw (vue1) -- +(150:2em); 
    \draw (vue2) -- +(30:2em); 
    \draw (vue0) -- +(270:2em); 
    \end{scope}

\end{tikzpicture}%
}
          \caption{The gadget for an {\sc and} vertex.}
          \label{fig:404}
      \end{minipage}
      \begin{minipage}{0.45\hsize}
          \centering
          \scalebox{\picscale}{%

\begin{tikzpicture}
    \begin{scope}
    \node[vertexSec4] at (1, 3) (v13){$2$};
    \node[vertexSec4] at (1, 1) (v11){$2$};
    \node[vertexSec4] at (2, 3) (v23){$2$};
    \node[vertexSec4] at (2, 1) (v21){$2$};
    \node[vertexSpSec4] at (3, 2) (v32){$3$};
    \node[vertexSec4] at (4, 2) (v42){$1$};
    \node[vertexSpSec4] at (5, 2) (v52){$2$};
    \node[vertexSec4] at (6, 2) (v62){$1$};
    \node[vertexSpSec4] at (7, 2) (v72){$2$};
    \node[vertexSec4] at (8, 2) (v82){$1$};
    \node[vertexSpSec4] at (9, 3) (v93){$2$};
    \node[vertexSpSec4] at (9, 1) (v91){$2$};

    \draw (v13) -- (v11); 
    \draw (v13) -- (v23); 
    \draw (v11) -- (v21); 
    \draw (v23) -- (v32); 
    \draw (v21) -- (v32); 
    \draw (v32) -- (v42); 
    \draw (v42) -- (v52); 
    \draw (v52) -- (v62); 
    \draw (v62) -- (v72); 
    \draw (v72) -- (v82); 
    \draw (v82) -- (v93); 
    \draw (v82) -- (v91); 

    \draw (v13) -- +(145:3em); 
    \draw (v11) -- +(215:3em); 
    \draw (v93) -- +(45:3em); 
    \draw (v91) -- +(315:3em); 

        \node[single arrow,draw=black,fill=black!10,minimum height=3em,single arrow head extend=0.4em, rotate=-90] at ([shift=({270:6em})]v52) (arr) {};
    \end{scope}

    \begin{scope}[yshift=-14em]
    \node[vertexSec4] at (1, 3) (v13){$2$};
    \node[vertexSec4] at (1, 1) (v11){$2$};
    \node[vertexSec4] at (2, 3) (v23){$2$};
    \node[vertexSec4] at (2, 1) (v21){$2$};
    \node[vertexSpSec4] at (3, 2) (v32){$3$};
    \node[vertexSec4] at (4, 2) (v42){$1$};
    \node[vertexSpSec4] at (5, 2) (v52){$2$};
    \node[vertexSec4] at (6, 2) (v62){$1$};
    \node[vertexSpSec4] at (7, 2) (v72){$2$};
    \node[vertexSec4] at (8, 2) (v82){$1$};
    \node[vertexSpSec4] at (9, 3) (v93){$2$};
    \node[vertexSpSec4] at (9, 1) (v91){$2$};

    \draw (v13) -- (v11); 
    \draw (v13) -- (v23); 
    \draw (v11) -- (v21); 
    \draw (v23) -- (v32); 
    \draw (v21) -- (v32); 
    \draw (v32) -- (v42); 
    \draw (v42) -- (v52); 
    \draw (v52) -- (v62); 
    \draw (v62) -- (v72); 
    \draw (v72) -- (v82); 
    \draw (v82) -- (v93); 
    \draw (v82) -- (v91); 

    \draw (v13) -- +(145:3em); 
    \draw (v11) -- +(215:3em); 
    \draw (v93) -- +(45:3em); 
    \draw (v91) -- +(315:3em); 

    \newcommand{\expen}[3]{ \draw (#1) -- +(#2:#3) node[circle, fill=red, minimum size=5pt,inner sep=0pt] (#1p) {}; }
    \expen{v13}{60}{2em}
    \expen{v11}{60}{2em}
    \expen{v23}{60}{2em}
    \expen{v21}{60}{2em}
    \expen{v42}{60}{2em}
    \expen{v42}{120}{2em}
    \expen{v52}{60}{2em}
    \expen{v62}{60}{2em}
    \expen{v62}{120}{2em}
    \expen{v72}{60}{2em}
    \expen{v82}{60}{2em}
    \expen{v82}{120}{2em}
    \expen{v93}{120}{2em}
    \expen{v91}{60}{2em}

\tikzset{tvertex/.style={auto=left,circle,fill=blue,minimum size=5pt,inner sep=0pt}} 
\tikzset{tleaf/.style={}} 
\tikzset{tedge/.style={blue}} 

    \node[tvertex] at (6.9, -0.5) (tv5)     {};
        \node[tvertex] at (5, 0) (tv10)   {};
            \node[tleaf] at (4, -0.3) (tv6)    {};
            \node[tleaf] at (4.4, 0.4) (tv7)     {};
        \node[tvertex] at (5.7, -0.7) (tv8)    {};
            \node[tleaf] at (4.4, -1) (tv9)    {};
            \node[tleaf] at (4.8, -1.3) (tv11) {};

        \node[tvertex] at (7.4, -0.4) (tv4)   {};
            \node[tvertex] at (7.5, 0.2) (tv2)    {};
            \node[tvertex] at (6, 0.4) (tv3)      {};
                \node[tvertex] at (5.5, 1) (tv1)      {};

    \draw[tedge] (v32)  -- (tv1);
    \draw[tedge] (v52)  -- (tv1);
    \draw[tedge] (v72)  -- (tv3);
    \draw[tedge] (v91)  -- (tv2);
    \draw[tedge] (v93)  -- (tv2);

    \draw[tedge] (tv5)  -- (tv10);
        \draw[tedge] (tv10) -- (tv6);
        \draw[tedge] (tv10) -- (tv7);
    \draw[tedge] (tv5)  -- (tv8);
        \draw[tedge] (tv8)  -- (tv11);
        \draw[tedge] (tv8)  -- (tv9);
    \draw[tedge] (tv5)  -- (tv4);
        \draw[tedge] (tv4)  -- (tv2);
        \draw[tedge] (tv4)  -- (tv3);
            \draw[tedge] (tv3)  -- (tv1);

    \end{scope}
\end{tikzpicture}%
}
          \caption{Construction of $G$ from $G'$.}
      \label{fig:405}
      \end{minipage}
\end{figure}

We next construct $G = (V, E)$ by adding new vertices and edges to $G'=(V', E')$ as follows (see Figure~\ref{fig:405} for an illustration). 
\begin{itemize}
\item
We construct a tree $T^*=(V(T^*), E(T^*))$ of maximum degree at most three such that $V(T^*) \cap V' = L$, $E(T^*) \cap E' = \emptyset$, and 
$L$ is the set of all the leaves of $T^*$. 
Then, we attach $T^*$ to $G'$. We denote the obtained graph by $G' + T^*$. 
\item
For each vertex $v \in V'$, we add $d - b(v)$ new vertices $\bar{v}_1, \dots , \bar{v}_{d-b(v)}$ 
and new edges $v \bar{v}_1, \dots ,  v \bar{v}_{d-b(v)}$. 
\end{itemize}


\medskip
\noindent
\textbf{Correspondence between solutions. }
In order to see the correspondence between 
NCL configurations in $H$ and 
spanning trees in $G$ with maximum degree at most $d$, 
we begin with the following easy lemma. 

\begin{lemma}
\label{obs:NCL01}
Any spanning tree $T$ in $G$ with maximum degree at most $d$
satisfies the following properties:  
(a) $v \bar{v}_i \in E(T)$ for $v \in V'$ and for $i \in \{1, 2, \dots , d-b(v)\}$;  
(b) $|\inc{G'+T^*}{v} \cap E(T)| \le b(v)$ for $v \in V'$;  
(c) $T$ contains exactly one of $v_e v_{u,e}$ and $v_e v_{u',e}$ for $e = u u' \in E(H)$; 
and 
(d) $E(T^*) \subseteq E(T)$. 
\end{lemma}

\begin{proof}
Since $T$ is a spanning tree with maximum degree at most $d$, (a), (b), and (c) are obvious. 
By (b), $T - E(T^*)$ contains no path connecting two distinct components of 
$G' - \{v \in V' \mid b(v) = 1 \}$. 
Since each connected component of $G' - \{v \in V' \mid b(v) = 1 \}$ contains exactly one vertex in $L$, 
for any pair of vertices $v_1, v_2 \in L$, the unique $v_1$-$v_2$ path in $T^*$ must be contained in $T$. 
This shows that $E(T^*) \subseteq E(T)$, because $L$ is the set of all the leaves of $T^*$. 
\end{proof}

For a spanning tree $T$ in $G$ with maximum degree at most $d$, 
we define an orientation $\sigma_T$ of $H$ as follows: 
an edge $e = u u' \in E(H)$ is inward for $u$ if $v_e v_{u',e} \in E(T)$, 
and it is outward for $u$ if $v_e v_{u,e} \in E(T)$. 
This defines an orientation of $H$ by~\cref{obs:NCL01} (c). 
The following two lemmas show the correspondence between 
NCL configurations in $H$ and 
spanning trees in $G$ with maximum degree at most $d$.

\begin{lemma}
\label{obs:NCL02}
For any spanning tree $T$ in $G$ with maximum degree at most $d$, 
the orientation $\sigma_T$ is an NCL configuration of $H$. 
\end{lemma}

\begin{proof}
It suffices to show that, for any $u \in V(H)$, 
the total weights of incoming arcs at $u$ is at least two in $\sigma_T$. 

First, let $u \in V(H)$ be an {\sc or}-vertex with $\inc{H}{u} = \{e_1, e_2, e_3\}$. Since $T$ is a spanning tree, 
it holds that $r_u v_{u, e_i} \in E(T)$ for some $i \in \{1, 2, 3\}$. 
Then, since $|\inc{G'}{v_{u, e_i}} \cap E(T)| \le b(v_{u, e_i}) - 1 = 1$ 
by (b) and (d) in~\cref{obs:NCL01}, 
it holds that $v_{e_i} v_{u,e_i} \not\in E(T)$. 
This means that $e_i$ is inward for $u$ in $\sigma_T$, and hence 
the total weights of incoming arcs at $u$ is at least two. 

Second, let $u \in V(H)$ be an {\sc and}-vertex with $\inc{H}{u} = \{e_0, e_1, e_2\}$, where $e_0$ is a weight-$2$ edge and $e_1$ and $e_2$ are weight-$1$ edges. 
Since $T$ is a spanning tree, 
we have either $r_u v_{u, e_0} \in E(T)$ or $r_u w_u \in E(T)$. 
If $r_u v_{u, e_0} \in E(T)$, then $e_0$ is inward for $u$ in $\sigma_T$, 
which implies that the total weights of incoming arcs at $u$ is at least two. 
Therefore, it suffices to consider the case when $r_u w_u \in E(T)$. 
Since $|\inc{G'}{v} \cap E(T)| \le 2$ for $v \in \{w_u, x_u, y_u, v_{u,e_1}, v_{u,e_2}\}$ by (b) and (d) in~\cref{obs:NCL01}, 
we have either $\{r_u w_u, w_u x_u, x_u v_{u, e_1}, v_{u, e_1} v_{u, e_2}, v_{u, e_2} y_u \} \subseteq E(T)$ or 
$\{r_u w_u, w_u y_u, y_u v_{u, e_2}, v_{u, e_2} v_{u, e_1}, v_{u, e_1} x_u \} \subseteq E(T)$.  
In either case, $v_{e_i} v_{u,e_i} \not\in E(T)$ for $i \in \{1, 2\}$, 
because $|\inc{G'}{v_{u, e_i}} \cap E(T)| \le 2$. 
This means that $e_i$ is inward for $u$ in $\sigma_T$ for $i \in \{1, 2\}$, and hence 
the total weights of incoming arcs at $u$ is at least two. 

Therefore, $\sigma_T$ is an NCL configuration of $H$. 
\end{proof}

\begin{lemma}
\label{obs:NCL03}
For any NCL configuration $\sigma$ of $H$, we can construct
a spanning tree $T$ in $G$ with maximum degree at most $d$
such that $\sigma_T = \sigma$ in polynomial time. 
\end{lemma}

\begin{proof}
Given an NCL configuration $\sigma$ of $H$,
we construct a spanning subgraph $T$ of $G$ such that
$$
E(T) := E(T^*) \cup \{v \bar{v}_i \mid  v \in V',\ i \in \{1, 2, \dots , d-b(v)\}\} \cup \{f_e \mid e \in E(H)\} \cup \bigcup_{u \in V(H)} F_u, 
$$
where an edge $f_e$ for $e \in E(H)$ and an edge set $F_u$ for $u \in V(H)$ are defined as follows.  
\begin{itemize}
\item
For an edge $e = u u' \in E(H)$, 
define $f_e := v_e v_{u',e}$ if $e$ is inward for $u$ in $\sigma$ and 
define $f_e := v_e v_{u,e}$ otherwise. 
\item
For an {\sc or}-vertex $u \in V(H)$ with $\inc{H}{u} = \{e_1, e_2, e_3\}$, 
choose an arbitrarily edge $e_i$ that is inward for $u$ in $\sigma$ and 
define $F_u := \{r_u v_{u, e_i}\}$. 
Note that such $e_i$ exists as $\sigma$ is an NCL configuration. 
\item
For an {\sc and}-vertex $u \in V(H)$ with $\inc{H}{u} = \{e_0, e_1, e_2\}$, where $e_0$ is a weight-$2$ edge and $e_1$ and $e_2$ are weight-$1$ edges, 
define $F_u = E'_u \setminus \{ r_u w_u, v_{u, e_1} v_{u, e_2} \}$ if $e_0$ is inward for $u$ in $\sigma$, and 
define $F_u := E'_u \setminus \{ r_u v_{u, e_0}, w_u y_u \}$ otherwise.
\end{itemize}
Then, $T$ is a spanning tree in $G$ with maximum degree at most $d$ such that $\sigma_T = \sigma$, which completes the proof. 
\end{proof}

For two NCL configurations $\sigma_\ini$ and $\sigma_\tar$ of $H$, by~\cref{obs:NCL03}, 
we can construct spanning trees $T_\ini$ and $T_\tar$ in $G$ with maximum degree at most $d$
such that $\sigma_{T_\ini} = \sigma_\ini$ and $\sigma_{T_\tar} = \sigma_\tar$.  
This yields an instance $(G,\thr,\spt_\ini,\spt_\tar)$ of {\sc RST with Small Maximum Degree}.

\medskip
\noindent
\textbf{Correctness. }
In order to show the PSPACE-hardness of {\sc RST with Small Maximum Degree}, 
we show that the original instance $(H, \sigma_\ini, \sigma_\tar)$ of {\sc NCL Reconfiguration} 
is equivalent to the obtained instance $(G,\thr,\spt_\ini,\spt_\tar)$ of {\sc RST with Small Maximum Degree}, 
that is,  
we prove that $(H, \sigma_\ini, \sigma_\tar)$ is a $\yes$-instance if and only if 
$(G,\thr,\spt_\ini,\spt_\tar)$ is a $\yes$-instance. 
To this end, we use the following lemma. 

\begin{lemma}
\label{lem:NCL04}
Let $T_1$ and $T_2$ be spanning trees in $G$ with maximum degree at most $d$. 
If $\sigma_{T_1}$ and $\sigma_{T_2}$ are adjacent, then there is a reconfiguration sequence from $T_1$ to $T_2$
in which all the spanning trees have maximum degree at most $d$.  
\end{lemma}

\begin{proof}
Let $e^* \in E(H)$ be the unique edge in $H$ whose direction is different in $\sigma_1$ and $\sigma_2$. 
We prove the existence of a reconfiguration sequence by induction on $|E(T_1) \setminus E(T_2)|$. 
If $|E(T_1) \setminus E(T_2)| =1$, then $T_1$ and $T_2$ are adjacent, and hence the claim is obvious. 
Suppose that $|E(T_1) \setminus E(T_2)| \ge 2$. 
Since 
$\sigma_{T_1}$ and $\sigma_{T_2}$ are adjacent, 
there exists a vertex $u \in V(H)$ such that 
$T_1$ and $T_2$ contain different edge sets in 
the gadget corresponding to $u$.  
That is, $(E(T_1) \setminus E(T_2)) \cap \inc{G'}{r_u}\not= \emptyset$ for an {\sc or}-vertex $u \in V(H)$ or
$(E(T_1) \setminus E(T_2)) \cap E'_u \not= \emptyset$ for an {\sc and}-vertex $u \in V(H)$. 
We fix such a vertex $u \in V(H)$.

Suppose that $u$ is an {\sc or}-vertex such that $(E(T_1) \setminus E(T_2)) \cap \inc{G'}{r_u}\not= \emptyset$.  
In this case, $E(T_1) \cap \inc{G'}{r_u} = \{ r_u v_{u, e_i}\}$ and $E(T_2) \cap \inc{G'}{r_u} = \{ r_u v_{u, e_j} \}$ for some distinct $i, j \in \{1, 2, 3\}$. 
By changing the roles of $T_1$ and $T_2$ if necessary, we may assume that either $e^* \not\in \inc{H}{u}$ or $e^*$ is inward for $u$ in $\sigma_1$. 
Then, $T'_1 := T_1 - r_u v_{u, e_i} + r_u v_{u, e_j}$ is a spanning tree with maximum degree at most $d$ 
such that  $T'_1$ is adjacent to $T_1$, $\sigma_{T'_1} = \sigma_{T_1}$, and $|E(T'_1) \setminus E(T_2)| = |E(T_1) \setminus E(T_2)| - 1$. 
By induction hypothesis, $T'_1$ is reconfigurable to $T_2$, and hence $T_1$ is reconfigurable to $T_2$. 

Suppose that $u$ is an {\sc and}-vertex such that 
$|(E(T_1) \setminus E(T_2)) \cap E'_u| = 1$. 
By changing the roles of $T_1$ and $T_2$ if necessary, we may assume that either $e^* \not\in \inc{H}{u}$ or $e^*$ is inward for $u$ in $\sigma_{T_1}$. 
Then, $T'_1 := T_1 - (E(T_1) \cap E'_u) + (E(T_2) \cap E'_u)$ is a spanning tree with maximum degree at most $d$ 
such that  $T'_1$ is adjacent to $T_1$, $\sigma_{T'_1} = \sigma_{T_1}$, and $|E(T'_1) \setminus E(T_2)| = |E(T_1) \setminus E(T_2)| - 1$. 
By induction hypothesis, $T'_1$ is reconfigurable to $T_2$, and hence $T_1$ is reconfigurable to $T_2$. 

The remaining case is that $u$ is an {\sc and}-vertex such that 
$|(E(T_1) \setminus E(T_2)) \cap E'_u| \ge 2$. 
Since each of $T_1$ and $T_2$ contains exactly one edge in $\inc{G'}{r_u}$ and exactly four edges in $E'_u \setminus \inc{G'}{r_u}$,  
we have that $|(E(T_1) \setminus E(T_2)) \cap \inc{G'}{r_u}| = 1$ and $|(E(T_1) \setminus E(T_2)) \cap (E'_u \setminus \inc{G'}{r_u})| = 1$. 
By changing the roles of $T_1$ and $T_2$ if necessary, we may assume that 
$E(T_1) \cap \inc{G'}{r_u} = \{ r_u v_{u, e_0} \}$ and $E(T_2) \cap \inc{G'}{r_u} = \{ r_u w_u \}$. 
This implies that $v_{u, e_0} v_{e_0} \not\in E(T_1)$, and hence $e_0$ is inward for $u$ in $\sigma_{T_1}$. 
We consider the following two cases separately. 
\begin{itemize}
\item
Suppose that $e_0$ is inward for $u$ in $\sigma_{T_2}$.
In this case, $T'_2 := T_2 - r_u w_u + r_u v_{u, e_0}$ is a spanning tree with maximum degree at most $d$ 
such that  $T'_2$ is adjacent to $T_2$, $\sigma_{T'_2} = \sigma_{T_2}$, and $|E(T_1) \setminus E(T'_2)| = |E(T_1) \setminus E(T_2)| - 1$. 
By induction hypothesis, $T_1$ is reconfigurable to $T'_2$, and hence $T_1$ is reconfigurable to $T_2$. 
\item
Suppose that $e_0$ is outward for $u$ in $\sigma_{T_2}$.
In this case, $e_1$ and $e_2$ are inward for $u$ in $\sigma_{T_2}$ by \cref{obs:NCL02}. 
Furthermore, since $e^* = e_0$ holds, $e_1$ and $e_2$ are inward for $u$ also in $\sigma_{T_1}$, 
that  is, $v_{u, e_1} v_{e_1}, v_{u, e_2} v_{e_2} \not\in E(T_1)$. 
Then, $T'_1 := T_1 - (E(T_1) \cap (E'_u \setminus \inc{G'}{r_u})) + (E(T_2) \cap (E'_u \setminus \inc{G'}{r_u}))$ is a spanning tree with maximum degree at most $d$ 
such that  $T'_1$ is adjacent to $T_1$, $\sigma_{T'_1} = \sigma_{T_1}$, and $|E(T'_1) \setminus E(T_2)| = |E(T_1) \setminus E(T_2)| - 1$. 
By induction hypothesis, $T'_1$ is reconfigurable to $T_2$, and hence $T_1$ is reconfigurable to $T_2$. 
\end{itemize}
By the above argument, there is a reconfiguration sequence from $T_1$ to $T_2$. 
\end{proof}

We are now ready to show the equivalence of $(H, \sigma_\ini, \sigma_\tar)$ and $(G,\thr,\spt_\ini,\spt_\tar)$. 
\begin{lemma}
Let $(H, \sigma_\ini, \sigma_\tar)$ be an instance of {\sc NCL Reconfiguration}
and $(G,\thr,\spt_\ini,\spt_\tar)$ be an instance of {\sc RST with Small Maximum Degree}
obtained by the above construction. 
Then, $(H, \sigma_\ini, \sigma_\tar)$ is a $\yes$-instance if and only if 
$(G,\thr,\spt_\ini,\spt_\tar)$ is a $\yes$-instance. 
\end{lemma}

\begin{proof}
    We first show the ``if'' part. 
    Suppose that there exists a reconfiguration sequence $\langle T_\ini = T_0, T_1, \dots, T_k = T_\tar \rangle$ from $T_\ini$ to $T_\tar$, where 
    $T_i$ is a spanning tree in $G$ with maximum degree at most $d$ for any $i \in \{0, 1, \dots , k\}$ and
    $T_i$ and $T_{i+1}$ are adjacent for any $i  \in \{0, 1, \dots, k-1\}$.
    Then, $\sigma_{T_i}$ is an NCL configuration of $H$ for $i \in \{0, 1, \dots , k\}$ by \cref{obs:NCL02}, and 
    we have either $\sigma_{T_i} = \sigma_{T_{i+1}}$ or $\sigma_{T_i}$ and $\sigma_{T_{i+1}}$ are adjacent for $i \in \{0, 1, \dots, k-1\}$ as $|E(T_i) \setminus E(T_{i+1})| \le 1$. 
    Since $\sigma_{T_0} = \sigma_{T_\ini} = \sigma_\ini$ and  $\sigma_{T_k} = \sigma_{T_\tar} = \sigma_\tar$, there exists a sequence of adjacent NCL configurations from $\sigma_\ini$ to $\sigma_\tar$. 
    
    To show the ``only-if'' part, suppose that there exists a reconfiguration sequence $\langle \sigma_\ini = \sigma_0, \sigma_1, \dots, \sigma_k = \sigma_\tar \rangle$, where 
    $\sigma_i$ is an NCL configuration of $H$ for any $i \in \{0, 1, \dots , k\}$ and $\sigma_i$ and $\sigma_{i+1}$ are adjacent for any $i  \in \{0, 1, \dots, k-1\}$.
    For $i \in \{1, 2,  \dots, k-1\}$, let $T_i$ be a spanning tree in $G$ with maximum degree at most $d$ such that $\sigma_{T_i} = \sigma_i$, 
	whose existence is guaranteed by Lemma~\ref{obs:NCL03}. 
	Let $T_0 := T_\ini$ and $T_k := T_\tar$. 
	Since \cref{lem:NCL04} shows that there is a reconfiguration sequence from $T_i$ to $T_{i+1}$ for $i \in \{0, 1, \dots, k-1\}$, 
	$T_\ini$ is reconfigurable to $T_\tar$. 
\end{proof}

This lemma shows that the above construction gives a polynomial reduction from
{\sc NCL Reconfiguration} to {\sc RST with Small Maximum Degree}. 
Therefore,  {\sc RST with Small Maximum Degree} is PSPACE-hard, 
which completes the proof of \cref{thm:smd_hard}.

\subsection{A Solvable Special Case}
\label{sec:smd_d-1}

In this subsection, we show a sufficient condition for the reconfigurability of instances.  
The condition is as follows; 
at least one of $\spt_\ini$ and $\spt_\tar$ has maximum degree at most $\thr-1$.
Without loss of generality, we may assume that $\spt_\tar$ satisfies the condition. 
Under this assumption, we have the following lemma. 

    \begin{lemma}
    \label{clm:swaponeedge}
        Suppose that $(G,\thr,\spt_\ini,\spt_\tar)$ is an instance of {\sc RST with Small Maximum Degree} 
        such that $\spt_\tar$ has maximum degree at most $\thr-1$.
        There exists an edge $e=xy \in E(T_\tar) \setminus E(T_\ini)$ such that 
        $d_{T_\ini}(x) \le d-1$ and $d_{T_\ini}(y) \le d-1$. 
    \end{lemma}
    \begin{proof}
        To derive a contradiction, assume that Lemma~\ref{clm:swaponeedge} does not hold, that is, 
        for any $e=xy \in E(T_\tar) \setminus E(T_\ini)$, we have  
        $d_{T_\ini}(x) = d$ or $d_{T_\ini}(y) = d$.
        Let $T^*_\ini := T_\ini -  E(T_\tar)$ and $T^*_\tar := T_\tar - E(T_\ini)$. 
        Note that $|E(T^*_\ini)| = |E(T^*_\tar)|$, because both $T_\ini$ and $T_\tar$ are spanning trees in $G$.  
        Let $S := \{ v \in V \mid d_{T_\ini} (v) =d,\ d_{T^*_\tar} (v) \ge 1 \}$. 
        With this notation, the assumption means that $S$ forms a vertex cover of $T^*_\tar$. 
        In what follows, we compare $|E(T^*_\ini)|$ and $|E(T^*_\tar)|$. 

        Let 
        $X_1 := \{v \in V \mid d_{T^*_\ini}(v)=1\}$ and 
        $X_{\ge 2} := \{v \in V \mid d_{T^*_\ini}(v)\ge 2\}$. 
        Then, we see that
        \begin{equation}
            \frac{1}{2} \sum_{v \in V} d_{T^*_\ini}(v) = |E(T^*_\ini)| < |X_1 \cup X_{\ge 2}|,     \label{eq:minusone01}
        \end{equation}
        because $T^*_\ini$ is a forest.  
        We also see that, for any $v \in S$, 
        \begin{equation}
            d_{T^*_\ini}(v) = d - |\delta_{T_\ini}(v) \cap \delta_{T_\tar}(v)|  
            \ge d_{T_\tar}(v) + 1 - |\delta_{T_\ini}(v) \cap \delta_{T_\tar}(v)| = d_{T^*_\tar}(v) + 1     \label{eq:minusone02}
        \end{equation}
        holds, and hence $S \subseteq X_{\ge 2}$. 
        With these observations, we obtain 
       \begin{align*}
            |E(T^*_\ini)| &= \sum_{v \in V} d_{T^*_\ini}(v) - \frac{1}{2} \sum_{v \in V} d_{T^*_\ini}(v) & & \\
                          &> \sum_{v \in V} d_{T^*_\ini}(v) - |X_1 \cup X_{\ge 2}|    & & \mbox{(by (\ref{eq:minusone01}))} \\
                          &= \sum_{v \in X_{\ge 2}} (d_{T^*_\ini}(v) - 1) & & \\
                          &\ge \sum_{v \in S} (d_{T^*_\ini}(v) - 1)       & & \mbox{(by $S \subseteq X_{\ge 2}$)} \\
                          &\ge \sum_{v \in S} d_{T^*_\tar}(v)             & & \mbox{(by (\ref{eq:minusone02}))} \\
                          &\ge |E(T^*_\tar)|,                             & & \mbox{(because $S$ is a vertex cover of $T^*_\tar$)}
        \end{align*}
        which is a contradiction to $|E(T^*_\ini)| = |E(T^*_\tar)|$.
        Therefore, Lemma~\ref{clm:swaponeedge} holds. 
    \end{proof}
    
        Let $e \in E(T_\tar) \setminus E(T_\ini)$ be the edge as in the lemma and let $e' \in  E(T_\ini) \setminus E(T_\tar)$ be an edge such that 
        $T'_\ini := T_\ini + e - e'$ is a spanning tree in $G$. 
        Note that the maximum degree of $T'_\ini$ is at most $\thr$. 
        Since $|E(T'_\ini) \setminus E(T_\tar)| = |E(T_\ini) \setminus E(T_\tar)|-1$, 
        $(G,\thr,\spt'_\ini,\spt_\tar)$ is a {\yes}-instance by induction. 
        This implies that $\spt_\ini$ is reconfigurable to $\spt_\tar$, and thus the following theorem holds.

	\begin{theorem}\label{thm:smd_d-1}
		If at least one of $\spt_\ini$ and $\spt_\tar$ has maximum degree at most $\thr-1$, then an instance $(G,\thr,\spt_\ini,\spt_\tar)$ of {\sc RST with Small Maximum Degree} is a {\yes}-instance.
	\end{theorem}
        


Note that the above discussion shows that we can find a reconfiguration sequence 
$\langle \spt_\ini = \spt_0,\spt_1,\ldots,\spt_k=\spt_\tar \rangle$
with $k = |E(T_\ini) \setminus E(T_\tar)|$, which is a shortest reconfiguration sequence, in polynomial time. 
Moreover, \Cref{thm:smd_d-1} implies the following corollary. 

	\begin{corollary}
		Let $G$ be a graph and $d$ be a positive integer. 
		If $G$ contains a spanning tree with maximum degree at most $\thr-1$, 
		then any instance $(G,\thr,\spt_\ini,\spt_\tar)$ of {\sc RST with Small Maximum Degree} is 
		a {\yes}-instance. 
	\end{corollary}
	\begin{proof}
	    Let $T^*$ be a spanning tree in $G$ with maximum degree at most $\thr-1$. 
	    Then, Theorem~\ref{thm:smd_d-1} shows that $T_\ini$ is reconfigurable to $T^*$ and 
	    $T^*$ is reconfigurable to $T_\tar$. 
	    Hence, $T_\ini$ is reconfigurable to $T_\tar$, which completes the proof. 
	\end{proof}

We note that it is not easy to determine whether or not $G$ contains a spanning tree with maximum degree at most $\thr-1$ even when $d=3$, because 
finding a Hamiltonian path in cubic graphs is NP-hard. 


\section{Large Diameter (Proof of Theorem \ref{thm:ldiam_hard})}\label{sec:diam_large}
	
	
	In this section, we prove \cref{thm:ldiam_hard}, which we restate here. 
	
	\ldiam*
	
	To prove the theorem, we give a polynomial reduction from {\sc Hamiltonian Path} problem to {\sc RST with Large Diameter}.
	A {\em Hamiltonian path} of a graph $G$ is a path that visits each vertex of $G$ exactly once.
	Given a graph $G=(V, E)$ and two vertices $s,t \in V$, the {\sc Hamiltonian Path} problem asks to determine whether or not $G$ has a Hamiltonian path whose endpoints are $s$ and $t$, which is known to be NP-hard~\cite{Karp72}.
	

\medskip
\noindent
\textbf{Reduction. }
		Let $(G^\prime,s^\prime,t^\prime)$ be an instance of {\sc Hamiltonian Path}. 
		We may assume that $G'$ is connected, since otherwise $(G^\prime,s^\prime,t^\prime)$ is trivially a \no-instance.  
		We construct a corresponding instance $(G,\thr,\spt_\ini,\spt_\tar)$ of {\sc RST with Large Diameter} as follows
		(see \figurename~\ref{fig:reduction_hamipath}). 

		Let $n^\prime$ be the number of vertices in $G^\prime$, that is,  $n^\prime = |V(G^\prime)|$.
		We first add three vertices $t_1$, $t_2$, and $t_3$, and five edges ${t^\prime}{t_1}$, ${t^\prime}{t_2}$, ${t_1}{t_2}$, ${t_1}{t_3}$, and ${t_2}{t_3}$ to $G'$. 
		Let $D$ be the subgraph induced by $\{t', t_1, t_2, t_3\}$, which is isomorphic to the so-called diamond graph. 
		We then add three paths $P_x = (x_{3n^\prime}, x_{3n^\prime -1}, \ldots, x_1, s^\prime)$, $P_y = (s^\prime, y_1, y_2, \ldots, y_{n^\prime -3}, t^\prime )$, 
		and $P_z = ( t_3, z_1, z_2, \ldots, z_{3n^\prime} )$, where all the vertices in $P_x$, $P_y$, and $P_z$ except for $s^\prime$, $t^\prime$, and $t_3$ are distinct new vertices.
        Note that $|E(P_x)| = |E(P_z)| = 3n^\prime$ and $|E(P_y)| = n^\prime-2$.
		Let $G=(V, E)$ be the obtained graph and set $\thr = 7n^\prime + 1$. 

		Let $F^\prime$ be an arbitrary spanning forest in $G^\prime$ such that $F^\prime$ consists of two connected components (trees) 
		of which one contains $s^\prime$ and the other contains $t^\prime$.
		Then, define spanning trees $\spt_\ini$ and $\spt_\tar$ in $G$ by 
		\begin{align*}
		    & E(\spt_\ini) = E(P_x) \cup E(P_y) \cup E(P_z) \cup E(F^\prime) \cup \{ {t^\prime}{t_1}, {t_1}{t_2}, {t_2}{t_3} \},   \\
			& E(\spt_\tar) =  E(P_x) \cup E(P_y) \cup E(P_z) \cup E(F^\prime) \cup \{ {t^\prime}{t_2}, {t_1}{t_2}, {t_1}{t_3} \}.
		\end{align*}
		We notice that $E(\spt_\ini) \setminus E(F^\prime)$ forms a path in $\spt_\ini$ of length $\thr = 7n^\prime +1$, and hence $\diam{\spt_\ini} \geq \thr$.
		Similarly, $E(\spt_\tar) \setminus E(F^\prime)$ forms a path in $\spt_\tar$ of length $\thr$, and hence $\diam{\spt_\tar} \geq \thr$.
		This completes the construction of the instance $(G,\thr,\spt_\ini,\spt_\tar)$ of {\sc RST with Large Diameter}.
	

\medskip
\noindent
\textbf{Correctness. }
        In the following, we show that $G'$ contains a Hamiltonian path from $s'$ to $t'$ if and only if $(G,\thr,\spt_\ini,\spt_\tar)$ is a \yes-instance. 
        The following lemma shows that the diameter of a spanning tree is dominated by the distance between $x_{3n^\prime}$ and $z_{3n^\prime}$. 
		\begin{lemma}
		\label{ld_diam}
			For any spanning tree $\spt$ in $G$, 
			$\diam{T} = \distorig{x_{3n^\prime}}{z_{3n^\prime}}{T}$. 
		\end{lemma}
		\begin{proof}
			Let $\spt$ be a spanning tree in $G$ and $P^*$ be a longest path in $\spt$.
            For $x, y \in V$ and for a spanning tree $\spt$ in $G$, we denote by $\unip{\spt}{x}{y}$ the unique path between $x$ and $y$ in $\spt$.
			Since $\unip{\spt}{x_{3n^\prime}}{z_{3n^\prime}}$ contains all the edges in $P_x$ and $P_z$, the length of $\unip{\spt}{x_{3n^\prime}}{z_{3n^\prime}}$ is at least $6n^\prime$, and hence $|E(P^*)| \ge |E(\unip{\spt}{x_{3n^\prime}}{z_{3n^\prime}})| \ge  6n^\prime$. Since each of $G-\{x_1, \dots , x_{3n'}\}$ and $G-\{z_1, \dots , z_{3n'}\}$ contains at most $5n'$ vertices, we obtain $V(P^*) \cap \{x_1, \dots , x_{3n'}\} \not= \emptyset$ and $V(P^*) \cap \{z_1, \dots , z_{3n'}\} \not= \emptyset$. 
            This shows that $P^* = \unip{\spt}{x_{i}}{z_{j}}$ for some $i, j \in \{1,2, \dots , 3n'\}$. 
            Since $\unip{\spt}{x_{i}}{z_{j}}$ is a subpath of $\unip{\spt}{x_{3n^\prime}}{z_{3n^\prime}}$,  
            $P^*$ must be equal to $\unip{\spt}{x_{3n^\prime}}{z_{3n^\prime}}$, that is, $\diam{T} = \distorig{x_{3n^\prime}}{z_{3n^\prime}}{T}$. 
		\end{proof}
		\begin{figure}[t]
			\centering
			\includegraphics[width=0.5\linewidth]{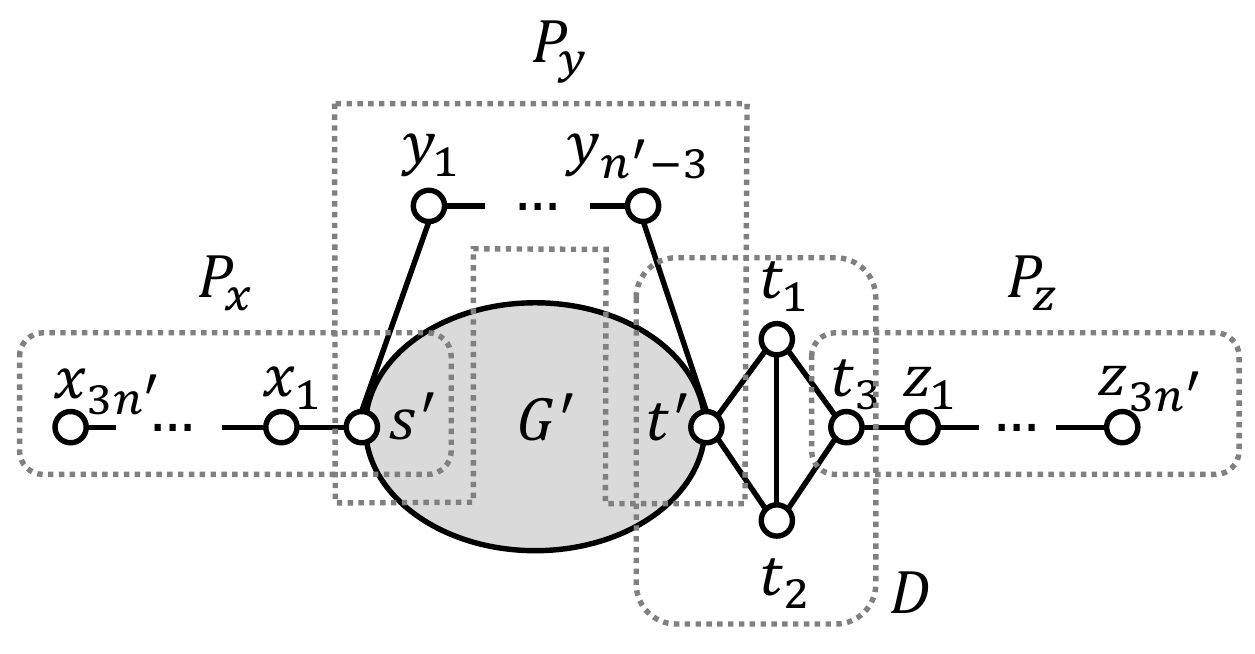}
			\caption{The graph $G$ in the corresponding instance.} 
			\label{fig:reduction_hamipath}
		\end{figure}
    	Thus, intuitively speaking, to keep the diameter and modify a spanning tree in $D$, we need to replace $P_y$ with a slightly longer path in $G^\prime$. 
    	Moreover, such a path must be a Hamiltonian path of $G^\prime$. This observation yields the following lemma and completes the proof of~\cref{thm:ldiam_hard}. 
		\begin{lemma}
			$(G^\prime,s^\prime,t^\prime)$ is a {\yes}-instance of {\sc Hamiltonian Path} if and only if $(G,\thr,\spt_\ini,\spt_\tar)$ is a {\yes}-instance of {\sc RST with Large Diameter}.
		\end{lemma}
		\begin{proof}
			We first prove the ``if'' direction.
			Suppose that $(G,\thr,\spt_\ini,\spt_\tar)$ is a {\yes}-instance.
			Then there is a reconfiguration sequence $\langle \spt_\ini = \spt_0,\spt_1,\ldots,\spt_k=\spt_\tar \rangle$ between $\spt_\ini$ and $\spt_\tar$ in which all the spanning trees have diameter at least $d$.
			Let $\spt_i$ be the first spanning tree in the sequence such that $\spt_i$ is obtained from $\spt_{i-1}$ by exchanging an edge in $D$, that is, $E(\spt_j) \cap E(D) = \{ {t^\prime}{t_1}, {t_1}{t_2}, {t_2}{t_3} \}$ for all $j \in \{ 0,1,\ldots,i-1 \}$ and $E(\spt_i) \cap E(D) \not= \{ {t^\prime}{t_1}, {t_1}{t_2}, {t_2}{t_3} \}$. 
			Note that such $i$ exists, because $E(\spt_k) \cap E(D) \not= \{ {t^\prime}{t_1}, {t_1}{t_2}, {t_2}{t_3} \}$. 
			Note also that $\distorig{t'}{t_3}{T_i}=2$ by the definition of $T_i$.
			Then, by Lemma~\ref{ld_diam}, we obtain 
			\begin{align*}
            7n^\prime +1 &\le \diam{T_i} \\
			             &= \distorig{x_{3n^\prime}}{z_{3n^\prime}}{T_i} \\
			             &=	\distorig{x_{3n^\prime}}{s^\prime}{T_i} + \distorig{s'}{t'}{T_i} + \distorig{t'}{t_3}{T_i} + \distorig{t_3}{z_{3n^\prime}}{T_i} \\
			             &= 3n' + \distorig{s'}{t'}{T_i} + 2 + 3n',
			\end{align*}
			and hence $\distorig{s'}{t'}{T_i} \ge n'-1$. 
			Since $P_y$ contains only $n^\prime -2$ edges, all the edges in $\unip{\spt_i}{s^\prime}{t^\prime}$ are contained in $G^\prime$.
			We thus conclude that $\unip{\spt_i}{s^\prime}{t^\prime}$ is a Hamiltonian path between $s^\prime$ and $t^\prime$ in $G^\prime$, and hence the ``if'' direction follows.
			
			We now prove the ``only-if'' direction.
			Suppose that $(G^\prime,s^\prime,t^\prime)$ is a {\yes}-instance, that is, $G^\prime$ contains a Hamiltonian path $P^*$ between $s^\prime$ and $t^\prime$.
			Let $e^*$ be any edge in $P^*$ and $e_y$ be any edge in $P_y$.
			We define five spanning trees $\spt_1$, $\spt_2$, $\spt_3$, $\spt_4$, and $\spt_5$ in $G$ as follows:
            \begin{align*}
                &E(\spt_1) = E(P_x) \cup E(P_y) \cup E(P_z) \cup E(P^* - e^*)  \cup \{ {t^\prime}{t_1}, {t_1}{t_2}, {t_2}{t_3} \},  \\				
				&E(\spt_2) = E(P_x) \cup E(P_y - e_y) \cup E(P_z) \cup E(P^*) \cup \{ {t^\prime}{t_1}, {t_1}{t_2}, {t_2}{t_3} \}, \\
				&E(\spt_3) = E(P_x) \cup E(P_y - e_y) \cup E(P_z) \cup E(P^*) \cup \{ {t^\prime}{t_2}, {t_1}{t_2}, {t_2}{t_3} \}, \\
				&E(\spt_4) = E(P_x) \cup E(P_y - e_y) \cup E(P_z) \cup E(P^*) \cup \{ {t^\prime}{t_2}, {t_1}{t_2}, {t_1}{t_3} \}, \\
				&E(\spt_5) = E(P_x) \cup E(P_y) \cup E(P_z) \cup E(P^* - e^*) \cup \{ {t^\prime}{t_2}, {t_1}{t_2}, {t_1}{t_3} \}.
            \end{align*}
			We observe that $\langle \spt_1, \spt_2, \spt_3, \spt_4, \spt_5 \rangle$ is a reconfiguration sequence from $\spt_1$ and $\spt_5$ 
			in which all the spanning trees have diameter at least $d=7n'+1$. 
			Thus, in order to show that $\spt_\ini$ is reconfigurable to $\spt_\tar$, it suffices to show that 
			$\spt_\ini$ is reconfigurable to $\spt_1$ and $\spt_5$ is reconfigurable to $\spt_\tar$. 
			Since $\unip{\spt_\ini}{x_{3n^\prime}}{z_{3n^\prime}} = \unip{\spt_1}{x_{3n^\prime}}{z_{3n^\prime}}$, 
			\cref{lem:known} shows that there is a reconfiguration sequence from $\spt_\ini$ to $\spt_1$ in which all the spanning trees contain $E(\unip{\spt_\ini}{x_{3n^\prime}}{z_{3n^\prime}}) \subseteq E(\spt_\ini) \cap E(\spt_1)$.
			Therefore, every spanning tree in the sequence has diameter at least $\thr$, and hence $\spt_\ini$ is reconfigurable to $\spt_1$. 
			Similarly, $\spt_5$ is reconfigurable to $\spt_\tar$. 
			By combining them, we have that $\spt_\ini$ is reconfigurable to $\spt_\tar$, which completes the proof of the ``only-if'' direction. 
		\end{proof}


\section{Small Diameter (Proof of Theorem \ref{thm:sdiam_P})}\label{sec:small_diam}

In this section, we prove \cref{thm:sdiam_P}, which we restate here. 

\sdiam*


After giving some preliminaries for the proof in Section~\ref{sec:061}, 
we describe a naive algorithm for the problem in Section~\ref{sec:062}, 
which does not necessarily run in polynomial time. 
Then, by modifying it, we give a polynomial-time algorithm in Section~\ref{sec:063}. 
Some technical parts in the proof are deferred to Sections~\ref{sec:findgoodtriple} and \ref{sec:goodsequence}.

\subsection{Preliminaries for the Proof}
\label{sec:061}

Throughout the proof of~\cref{thm:sdiam_P},  
we fix a positive integer $d$. 
For each edge $e \in E$, we denote the middle point of $e$ by $p_e$. 
We denote $R(H) := \{p_e \mid e \in E(H)\}$ for a subgraph $H$ of $G$ and let $R:=R(G)$. 
Let $\hat G$ be the graph on $V \cup R$ that is obtained from $G$ by subdividing each edge. 
Then, since $\distorig{u}{v}{G} = \frac{1}{2} \distorig{u}{v}{\hat G}$ for $u, v \in V$, 
we can naturally extend the domain of the distance to $V \cup R$ by setting
$\distorig{u}{v}{G} := \frac{1}{2} \distorig{u}{v}{\hat G}$ for $u, v \in V \cup R$. 
We also define $\ecc{v}{G} := \max\inset{ \distorig{v}{u}{G}}{u \in V}$ for $v \in R$.
If no confusion may arise, for $u, v \in V \cup R$, 
a $u$-$v$ path in $\hat G$ is sometimes called a $u$-$v$ path in $G$. 
We can see that spanning trees with diameter at most $d$ are characterized as follows (see also~\cite{10.1016/0020-0190(94)00183-Y}). 

\begin{lemma}
\label{lem:001}
     For any spanning tree $T$ in $G=(V,E)$, 
    $\diam{T} \le d$ if and only if there exists $r \in V \cup R(T)$ such that 
    $\ecc{r}{T} \le \frac{d}{2}$.    
\end{lemma}

\begin{proof}
    To show the ``if'' part, suppose that there exists $r \in V \cup R(T)$ such that $\ecc{r}{T} \le \frac{d}{2}$.   
    Then, for any $u, v \in V$, $\distorig{u}{v}{T} \le \distorig{u}{r}{T} + \distorig{r}{v}{T} \le 2 \ecc{r}{T} \le d$, 
    which shows that $\diam{T} \le d$. 

    To show the ``only-if'' part, suppose that $\diam{T} \le d$. 
    Let $d^* := \diam{T}$ and let $u, v \in V$ be a pair of vertices such that $\distorig{u}{v}{T} = d^*$. 
    Let $r \in V \cup R(T)$ be the middle point of $u$ and $v$ in $T$, that is, $\distorig{u}{r}{T}=\distorig{r}{v}{T}= \frac{d^*}{2}$. 
    Since $T$ is a spanning tree, for any $x \in V$, $\frac{d^*}{2} + \distorig{r}{x}{T} = \max\{\distorig{u}{x}{T},\ \distorig{v}{x}{T} \} \le d^*$. 
    This shows that $\distorig{r}{x}{T} \le \frac{d^*}{2}$, that is, $\ecc{r}{T} \le \frac{d}{2}$.    
\end{proof}


We say that a subgraph $Q$ of $G$ is a \emph{spanning pseudotree} if it is a connected spanning subgraph containing at most one cycle. 
In other words, a spanning pseudotree is obtained from a spanning tree by adding at most one edge.  
For brevity, a spanning pseudotree is simply called a \emph{pseudotree}.
For a pseudotree $Q$, let $\UniC{Q}$ denote the unique cycle in $Q$ if it exists.
We can easily see that, 
for two spanning trees $T_1$ and $T_2$ with diameter at most $d$, 
$T_1 \leftrightarrow T_2$ if and only if $T_1 \cup T_2$ forms a pseudotree. 
For a pseudotree $Q$, 
we refer a point $r \in V \cup R(Q)$ as a \emph{center point} of $Q$ if $\ecc{r}{Q} \le \frac{d}{2}$.
Note that a center point is not necessarily unique even if $Q$ is a spanning tree. 
For a pseudotree $Q$, let $\inner{Q} \subseteq V \cup R(Q)$ be the set of all center points of $Q$.

\subsection{Algorithm Using Center Points: First Attempt}
\label{sec:062}

In this subsection, as a first step, we give an algorithm for \GSTRwithsmallDaim \ whose running time is not necessarily polynomial. 
In the same say as {\sc RST with Large Maximum Degree} (Section~\ref{sec:lmd}), 
the proposed algorithm is based on testing the reachability in an auxiliary graph $\recgraph$, which is defined as follows.
The vertex set of $\recgraph$ is defined as $V \cup R$, 
where each vertex $v$ in $V(\recgraph)$ corresponds to the set of all the spanning trees containing $v$ as a center point.
For any pair $u, v$ of distinct vertices in $V(\recgraph)$,
there is an edge $uv \in E(\recgraph)$ if and only if there is a pseudotree $Q$ with $u, v \in \inner{Q}$. 
As we will see in \cref{prop:correctness1} later, for two spanning trees $T_u$ and $T_v$ having center points $u$ and $v$, respectively, 
$\recgraph$ contains a $u$-$v$ path if and only if $T_u$ and $T_v$ are reconfigurable to each other. 
Thus, to solve \GSTRwithsmallDaim, it is enough to find a path from a center point of $T_\ini$ to a center point of $T_\tar$ on $\recgraph$. 
See Algorithm~\ref{alg:01} for a pseudocode of our algorithm. 


\begin{algorithm}[t]
    \KwInput{A graph $G$ and two spanning trees $T_\ini$ and $T_\tar$ in $G$ with diameter at most $d$. }
    \KwOutput{Is $T_\ini$ reconfigurable to $T_\tar$?}
    Compute  $\inner{T_\ini}$ and $\inner{T_\tar}$, and construct $\recgraph$\; 
    \lIf{there is a path between $\inner{T_\ini}$ and $\inner{T_\tar}$ in $\recgraph$}{\Return{\yes}}
    \lElse{\Return{\no}}
    \caption{First algorithm for \GSTRwithsmallDaim}\label{alg:01}
\end{algorithm}

To show the correctness of Algorithm~\ref{alg:01}, we begin with easy but important lemmas. 

\begin{lemma}
      \label{lem:share:inner:vertex}
      Let $T_1$ and $T_2$ be spanning trees in $G$ with diameter at most $d$.
      If there exists a point $r \in \inner{T_1} \cap \inner{T_2}$, then $T_1$ is reconfigurable to $T_2$. 
\end{lemma}
\begin{proof}
      Let $T^*$ be the spanning tree that is obtained by applying the breadth first search from $r$ in $G$. 
      Here, if $r \in R$ is the middle point of $uv \in E$, then the breadth first search is started from $\{u, v\}$. 
      Since $\distorig{r}{v}{T^*} \le \distorig{r}{v}{T_1} \le \frac{d}{2}$ for any $v \in V$, 
      the diameter of $T^*$ is at most $d$.
      For $v \in V$, let $P_{T^*}(v)$ (resp.~$P_{T_1}(v)$) denote the unique path from $r$ to $v$ in $T^*$ (resp.~$T_1$). 
      In $T^*$, we say that a vertex $u \in V$ is the \emph{parent} of $v$ if $uv \in E(T^*)$ and 
      $\distorig{r}{v}{T^*} = \distorig{r}{u}{T^*} + 1$.  
      The parent in $T_1$ is defined in the same way.  

      In order to show that $T_1$ is reconfigurable to $T_2$, 
      it suffices to show that $T_i$ is reconfigurable to $T^*$ for $i \in \{1, 2\}$.
      Suppose that $T_1 \not= T^*$ and let $xy$ be an edge in $E(T^*) \setminus E(T_1)$ that minimizes $\min\{\distorig{r}{x}{T^*}, \distorig{r}{y}{T^*}\}$. 
      Without loss of generality, we assume that $x$ is the parent of $y$ in $T^*$. 
      Let $w \in V$ be the parent of $y$ in $T_1$ and define $T'_1 := T_1 + \set{xy} - \set{wy}$, which is a spanning tree in $G$.  
      By the choice of $xy$, 
      we obtain $P_{T_1}(x) = P_{T^*}(x)$, and hence $P_{T'_1}(y) = P_{T^*}(y)$ and $\distorig{r}{y}{T'_1} = \distorig{r}{y}{T^*} \le \distorig{r}{y}{T_1}$. 
      Since this shows that $\distorig{r}{v}{T'_1} \le \distorig{r}{v}{T_1} \le \frac{d}{2}$ for any $v \in V$, the diameter of $T'_1$ is at most $d$ by \cref{lem:001}. 
      We observe that replacing $T_1$ with $T'_1$ increases $|\{v \in V \mid P_{T_1}(v) = P_{T^*}(v)\}|$ by at least one, because $P_{T_1}(y) \not= P_{T'_1}(y) = P_{T^*}(y)$. 
      Therefore, by applying this procedure at most $|V|$ times, we obtain a reconfiguration sequence from $T_1$ to $T^*$. 
      We can also obtain a reconfiguration sequence from $T_2$ to $T^*$ in the same way.
      Hence, the statement holds.
\end{proof}

\begin{lemma}
      \label{lem:pseudotree:and:two:center}
      Let $r_1, r_2 \in V \cup R$ (possibly $r_1=r_2$). 
      There exists a pseudotree $Q$ with $r_1, r_2 \in \inner{Q}$ 
      if and only if
      there exist two spanning trees $T_1$ and $T_2$ such that $r_i \in \inner{T_i}$ for $i=1,2$ and $T_1 \leftrightarrow T_2$ (possibly $T_1=T_2$).
\end{lemma}
\begin{proof}
      We first consider the ``if'' part. Suppose that there exist two spanning trees $T_1$ and $T_2$ such that $r_i \in \inner{T_i}$ for $i=1,2$ and $T_1 \leftrightarrow T_2$.
      Then $Q := T_1 \cup T_2$ is a desired pseudotree as $\ecc{r_i}{Q} \le \ecc{r_i}{T_i} \le \frac{d}{2}$ for $i = 1, 2$. 

      We next consider the ``only-if'' part. Suppose that $Q$ is a pseudotree with $r_1, r_2 \in \inner{Q}$. 
      For $i=1, 2$, let $T_i$ be the spanning tree that is obtained by applying the breadth first search from $r_i$ in $Q$.
      Then, we obtain $\ecc{r_i}{T_i} = \ecc{r_i}{Q} \le \frac{d}{2}$. 
      Furthermore, since $|E(T_1) \setminus E(T_2)| \le |E(Q) \setminus E(T_2)| = 1$, it holds that $T_1 \leftrightarrow T_2$. 
\end{proof}

By these lemmas, we can show the correctness of Algorithm~\ref{alg:01}.

\begin{proposition}
\label{prop:correctness1}
    Let $T_\ini$ and $T_\tar$ be spanning trees with diameter at most $d$. 
    Then, $T_\ini$ is reconfigurable to $T_\tar$ if and only if $\recgraph$ contains a path from $\inner{T_\ini}$ to $\inner{T_\tar}$. 
\end{proposition}

\begin{proof}
    We first show the ``only if'' part. 
    Suppose that there exists a reconfiguration sequence $\langle T_\ini = T_0, T_1, \dots, T_k = T_\tar \rangle$ from $T_\ini$ to $T_\tar$, where 
    $T_i$ is a spanning tree of diameter at most $d$ for any $i \in \{0, 1, \dots , k\}$ and
    $T_i \leftrightarrow T_{i+1}$ for any $i  \in \{0, 1, \dots, k-1\}$.
    Let $r_i$ be a center point of $T_i$, where its existence is guaranteed by \cref{lem:001}. 
    For $i \in \{0, 1, \dots, k-1\}$, by \cref{lem:pseudotree:and:two:center}, there exists a pseudotree $Q_i$ having both $r_i$ and $r_{i+1}$ as center points. 
    This means that either $r_i = r_{i+1}$ or $\recgraph$ contains an edge $r_i r_{i+1}$. 
    Since $r_0 \in \inner{T_\ini}$ and $r_{k} \in \inner{T_\tar}$, $\recgraph$ contains a path 
    from $\inner{T_\ini}$ to $\inner{T_\tar}$. 
    
    To show the ``if'' part, suppose that $\recgraph$ contains a path $(r_0, r_1, \dots , r_k)$ from $\inner{T_\ini}$ to $\inner{T_\tar}$. 
    For $i \in \{0, 1, \dots, k-1\}$, since $r_i r_{i+1} \in E(\recgraph)$ implies the existence of a pseudotree $Q_i$ with $r_i, r_{i+1} \in \inner{Q_i}$, 
    \cref{lem:pseudotree:and:two:center} shows that 
    there exist two spanning trees $T^+_i$ and $T^-_{i+1}$ such that $r_i \in \inner{T^+_i}$, $r_{i+1} \in \inner{T^-_{i+1}}$, and $T^+_i \leftrightarrow T^-_{i+1}$. 
    Let $T^-_0 := T_\ini$ and $T^+_k := T_\tar$. Then, for $i \in \{0, 1, \dots, k\}$, since $T^-_i$ and $T^+_i$ share $r_i$ as a center point, 
    $T^-_i$ is reconfigurable to $T^+_i$ by \cref{lem:share:inner:vertex}. This together with $T^+_i \leftrightarrow T^-_{i+1}$ shows that $T_\ini$ is reconfigurable to $T_\tar$. 
\end{proof}

Although this proposition shows the correctness of Algorithm~\ref{alg:01}, 
it does not imply a polynomial-time algorithm for \GSTRwithsmallDaim, because
it is not easy to construct $\recgraph$ efficiently. 
Indeed, for $u, v \in V(\recgraph)$, we do not know how to decide whether $uv \in E(\recgraph)$ or not in polynomial time.  
To avoid this problem, we efficiently construct a subgraph ${\recgraph}'$ of $\recgraph$ such that 
the reachability of $\recgraph'$ is equal to that of $\recgraph$, which is a key ingredient of our algorithm and discussed in the next subsection.

\subsection{Modified Algorithm}
\label{sec:063}

In this subsection, we give a polynomial-time algorithm for \GSTRwithsmallDaim. In our algorithm, it is important to uniquely determine a shortest path between two points. To achieve this, we use a \emph{perturbation technique} (see e.g.,~\cite{Cabello:Chambers:Erickson:SIAMJCOMP:2013}).

For each edge $e$ in $G$, we give a unique index $\idx{e} \in \{1, 2, \dots , |E|\}$ to $e$.
For $j \in \{1, 2, \dots , |E|\}$, 
let $\uvec{j} \in \mathbb{R}^{|E|}$ be the unit vector such that the $j$th coordinate is one and the other coordinates are zero.
For $e \in E$, define $\newweight{e}\coloneqq (1, \uvec{\idx{e}}) \in \mathbb{R} \times \mathbb{R}^{|E|}$. 
For two vectors $x, y  \in \mathbb{R}^k$, 
we denote $x < y$ if $x$ is lexicographically smaller than $y$. 
For two paths $P_1$ and $P_2$ in $G$, 
we say that $P_1$ is \emph{shorter than} $P_2$ if $\newweight{P_1}:= \sum_{e \in E(P_1)} \newweight{e}$ 
is lexicographically smaller than $\newweight{P_2}:= \sum_{e \in E(P_2)} \newweight{e}$. 
Since the first coordinate of $\newweight{P_i}$ is $|E(P_i)|$ for $i=1, 2$, 
if $|E(P_1)| < |E(P_2)|$, then $P_1$ is shorter than $P_2$. 
When $|E(P_1)| = |E(P_2)|$, we use the other coordinates to break ties. 
For $u, v \in V$, we define $\dist{u}{v}{G} := \min_P \sum_{e \in E(P)} \newweight{e}$, where 
the minimum is taken over all the $u$-$v$ paths $P$. 
Since $P_1 \not= P_2$ implies that $\newweight{P_1} \not= \newweight{P_2}$, 
the \emph{shortest path} between two vertices is uniquely determined.
We note that the unique shortest paths between two given vertices can be computed by using 
a standard shortest path algorithm. The running time is increased by the perturbation, but it is still polynomial. 

For an edge $e = uv \in E$ of length $\newweight{e} \in \mathbb{R} \times \mathbb{R}^{|E|}$, we regard $e$ as a curve connecting $u$ and $v$. 
An interior point $p$ on $e$ is represented by a triplet $(u, v, \alpha)$ with $\alpha \in \mathbb{R} \times \mathbb{R}^{|E|}$
such that ${\bf 0} \le \alpha \le \newweight{e}$, where $\le$ means the lexicographical order. 
Here, $\alpha$ represents the length between $u$ and $p$, and hence 
$(u, v, \alpha)$ and $(v, u, \newweight{e} - \alpha)$ represent the same point. 
For two points $p_1= (u_1, v_1, \alpha_1)$ and $p_2= (u_2, v_2, \alpha_2)$ in $G$, 
consider a curve $C$ connecting $p_1$ and $p_2$ that consists of a $u_1$-$u_2$ path $P$, a curve in $u_1 v_1$ between $u_1$ and $p_1$, and 
a curve in $u_2 v_2$ between $u_2$ and $p_2$. 
Such a curve $C$ is called a \emph{$p_1$-$p_2$ path} in $G$, and its length is defined as
$\newweight{C} := \sum_{e \in E(P)} \newweight{e} + \alpha_1 + \alpha_2$.

For a point $r \in V \cup R$, the \emph{shortest path tree from $r$} is the spanning tree in $G$ 
that contains the unique shortest $r$-$v$ path for any $v \in V$. 
For a pseudotree $Q$ and for two points $x$ and $y$ on $Q$, 
let $Q[x,y]$ denote the shortest $x$-$y$ path in $Q$, 
where we use this notation only when the shortest $x$-$y$ path is uniquely determined. 
For $\alpha \in \mathbb{R} \times \mathbb{R}^{|E|}$, 
let $\bar{\alpha}$ denote the first coordinate of $\alpha$, that is, 
$\bar \alpha$ is the length before the perturbation. 

We denote $\reconf{r_1}{r_2}{Q}$ if $Q$ is a pseudotree and $r_1, r_2 \in \inner{Q}$ with $r_1 \not= r_2$. 
For any pseudotree $Q$ and any points $r_1$ and $r_2$ with $\reconf{r_1}{r_2}{Q}$,
we say that a triplet $\tuple{r_1, r_2, Q}$ is \emph{good} if
\begin{enumerate}
      \item  $\vlabel{v}{r_1, r_2, Q} \le \vlabel{u}{r_1, r_2, Q} + \newweight{uv}$ for any $uv \in E$, and \label{enum:cond:good:a}
      \item  $\UniC{Q}$ contains both $r_1$ and $r_2$ if $\UniC{Q}$ exists, \label{enum:cond:good:b}
\end{enumerate}
where $\vlabel{v}{r_1, r_2, Q} := \max \{ \dist{r_1}{v}{Q},\ \dist{r_2}{v}{Q} \}$. 
Roughly speaking, the first condition means that $\vlabel{v}{r_1, r_2, Q}$ can be seen as the
distance from a certain point to $v$ in an auxiliary graph. 
If $r_1$ and $r_2$ are clear from the context, $\vlabel{v}{r_1, r_2, Q}$ is simply denoted by $\vlabel{v}{Q}$. 
We define the graph $\recgraph'$ as follows:
$V(\recgraph') = V \cup R$ and 
$\recgraph'$ contains an edge $r_1 r_2$ if and only if there is a pseudotree $Q$ such that $\reconf{r_1}{r_2}{Q}$ and $\tuple{r_1, r_2, Q}$ is good. 
Clearly, $\recgraph'$ is a subgraph of $\recgraph$.

The following theorem shows that we can determine whether $r_1 r_2 \in E(\recgraph')$ or not in polynomial time, whose proof is given in Section~\ref{sec:findgoodtriple}. 

\begin{theorem}
     \label{prop:reconf:good}
      Let $r_1$ and $r_2$ be points in $V \cup R$ with $r_1\not= r_2$.
      We can find in polynomial time a pseudotree $Q$ such that $\reconf{r_1}{r_2}{Q}$ and $\tuple{r_1, r_2, Q}$ is good if it exists. 
\end{theorem}

The next theorem shows that the reachability of $\recgraph'$ is equal to that of $\recgraph$, which is a key property of $\recgraph'$. 
A proof is given in Section~\ref{sec:goodsequence}.

\begin{theorem}
      \label{thm:good:sequence}
     For any $r_1, r_2 \in V \cup R$ with $r_1 r_2 \in E(\recgraph)$, $\recgraph'$ contains an $r_1$-$r_2$ path. 
\end{theorem}

We are now ready to prove \cref{thm:sdiam_P}. 
By \cref{prop:correctness1} and \cref{thm:good:sequence}, two spanning trees $T_\ini$ and $T_\tar$ are reconfigurable to each other if and only if $\recgraph'$ contains a path from $\inner{T_\ini}$ to $\inner{T_\tar}$. 
Since we can construct $\recgraph'$ in polynomial time by \cref{prop:reconf:good}, this can be tested in polynomial time. 
Therefore, \GSTRwithsmallDaim \ can be solved in polynomial time, which completes the proof of \cref{thm:sdiam_P}. 
A pseudocode of our algorithm is given in Algorithm~\ref{alg:02}. 
Note that all the proofs are constructive, and hence we can find a desired reconfiguration sequence from $T_\ini$ to $T_\tar$ in polynomial time if it exists.

\begin{algorithm}[t]
    \KwInput{A graph $G$ and two spanning trees $T_\ini$ and $T_\tar$ in $G$ with diameter at most $d$. }
    \KwOutput{Is $T_\ini$ reconfigurable to $T_\tar$?}
    Compute  $\inner{T_\ini}$ and $\inner{T_\tar}$, and construct $\recgraph' = (V \cup R, \emptyset)$\; 
    \For{$r_1, r_2 \in V \cup R$ with $r_1 \not= r_2$}{
        \If{there is a pseudotree $Q$ such that $\reconf{r_1}{r_2}{Q}$ and $(r_1, r_2, Q)$ is good}{
            Add an edge $r_1 r_2$ to $\recgraph'$\; 
        }
    }
    \lIf{there is a path between $\inner{T_\ini}$ and $\inner{T_\tar}$ in $\recgraph'$}{\Return{\yes}}
    \lElse{\Return{\no}}
  \caption{Modified algorithm for \GSTRwithsmallDaim}\label{alg:02}
\end{algorithm}

\section{Proof of Theorem~\ref{prop:reconf:good}}
\label{sec:findgoodtriple}

In this section, we give an algorithm for finding 
a pseudotree $Q$ such that 
$\reconf{r_1}{r_2}{Q}$ and $\tuple{r_1, r_2, Q}$ is good (if it exists), and prove Theorem~\ref{prop:reconf:good}. 
Our algorithm consists of two procedures: one finds a desired spanning tree $Q$ and the other
finds a desired pseudotree $Q$ with a cycle. 
We describe these two procedures in Sections~\ref{sec:Qhasnocycle} and \ref{sec:Qhascycle}, respectively. 
In what follows in this section, we fix $r_1, r_2 \in V \cup R$.

\subsection{\texorpdfstring{Finding a spanning tree $Q$}{Finding a spanning tree Q}}
\label{sec:Qhasnocycle}

We first consider the case when $Q$ is a spanning tree. 
Suppose that $Q$ is a spanning tree such that $\reconf{r_1}{r_2}{Q}$ and $\tuple{r_1, r_2, Q}$ is good. 
Let $e= v_1 v_2 \in E$ be the unique edge in $Q$ such that 
$\dist{r_1}{v_1}{Q} < \dist{r_2}{v_1}{Q}$ and $\dist{r_1}{v_2}{Q} > \dist{r_2}{v_2}{Q}$. 
That is, $e$ is the edge containing the middle point of $r_1$ and $r_2$. 
Note that such $e$ exists, because $\dist{r_1}{v}{Q} \not= \dist{r_2}{v}{Q}$ for each $v\in V$ as $\ell$ is the perturbed length. 
Define $G^+$ as the graph obtained from $G$ 
by adding a new vertex $r$ together with two new edges $e_1 = r v_1$ and $e_2 = r v_2$ (Figure~\ref{fig:701}). 
Set $\newweight{e_1} := \dist{r_2}{v_1}{Q}$ and $\newweight{e_2} := \dist{r_1}{v_2}{Q}$. 
We now show the following claim. 

\begin{figure}[t!]
      \centering

\begin{tikzpicture}

\tikzset{tvertex/.style={auto=left,circle,fill=blue,minimum size=5pt,inner sep=0pt}} 
\tikzset{tleaf/.style={}} 
\tikzset{tedge/.style={blue}} 

    \node[tvertex, label={0:$r$}] at (0, 0) (r)  {};
    \node[tvertex, label={90:$v_1$}] at (-0.7, 1) (v1)  {};
    \node[tvertex, label={90:$v_2$}] at (0.7, 1) (v2)  {};

    \node[tvertex] at (-2, 1.5) (v11)  {};
        \node[tvertex] at (-2.5, 2) (v111)  {};
        \node[tvertex] at (-3, 1.5) (v112)  {};
    \node[tvertex] at (-2, 0.5) (v12)  {};

    \node[tvertex] at (1.5, 0.6) (v21) {}; 
    \node[tvertex] at (1.5, 1.6) (v22) {}; 
    \node[tvertex, label={90:$r_2$}] at (2.5, 1.1) (r2) {};

    \draw[tedge] (r) -- node[midway, below left=-1mm] {$e_1$} (v1); 
    \draw[tedge] (r) -- node[midway, below right=-1mm] {$e_2$} (v2); 
    \draw (v1) -- node[midway,above] {$e$} (v2); 

    \draw (v1) -- (v11); 
    \draw (v1) -- (v12); 
    \draw (v11) -- (v111); 
    \draw (v11) -- (v112); 

    \path (v11) edge node[black, anchor=center, pos=0.5,font=\bfseries,label={[label distance=-1mm] 270:$r_1$}] {$\times$} (v112);

    \draw (v111) -- +(90:1em); 
    \draw (v111) -- +(145:1em); 

    \draw (v112) -- +(150:1em); 
    \draw (v112) -- +(180:1em); 
    \draw (v112) -- +(210:1em); 

    \draw (v12) -- +(160:1em); 
    \draw (v12) -- +(240:1em); 

    \draw (v2) -- (v21); 
    \draw (v2) -- (v22); 
    \draw (v22) -- (r2); 

    \draw (v22) -- +(45:1em); 

    \draw (r2) -- +(30:1em); 
    \draw (r2) -- +(-30:1em); 

    \draw (v21) -- +(30:1em); 
    \draw (v21) -- +(0:1em); 
    \draw (v21) -- +(-30:1em); 
\end{tikzpicture}%

      \caption{Construction of $G^+$.}
      \label{fig:701}
\end{figure}

\begin{claim}
For each $v\in V$, it holds that $\vlabel{v}{Q} = \dist{r}{v}{G^+}$. 
Furthermore, $Q - e + \{e_1, e_2\}$ is the shortest path tree starting from $r$ in $G^+$.  
\end{claim}

\begin{proof}
Let $Q_1$ and $Q_2$ be the connected components of $Q-e$ such that $v_i \in V(Q_i)$ for $i=1,2$. 
Then, 
$\vlabel{v}{Q} = \dist{r_2}{v_1}{Q} + \dist{v_1}{v}{Q-e} = \newweight{e_1} + \dist{v_1}{v}{Q-e}$ for $v \in V(Q_1)$ and 
$\vlabel{v}{Q} = \dist{r_1}{v_2}{Q} + \dist{v_2}{v}{Q-e} = \newweight{e_2} + \dist{v_2}{v}{Q-e}$ for $v \in V(Q_2)$. 
Therefore, for each $v \in V$, $Q - e + \{e_1, e_2\}$ contains an $r$-$v$ path $P_v$
whose length is $\vlabel{v}{Q}$. 
Furthermore, since $\tuple{r_1, r_2, Q}$ is good, $\vlabel{v}{Q} \le \vlabel{u}{Q} + \newweight{uv}$ holds for any $uv \in E$, 
by the correctness of the Bellman-Ford method (see e.g.~\cite[Section 8.3]{schrijver-book}, \cite[Section 7.1]{korte12}), 
$\vlabel{v}{Q}$ is equal to the shortest path length from $r$ to $v$ in $G^+$. 
That is $P_v$ is the unique shortest $r$-$v$ path in $G^+$, 
which shows the claim. 
\end{proof}

Since a subpath of a shortest path is also a shortest path, 
this claim implies that 
$Q[v_i, r_i]$ is the unique shortest $v_i$-$r_i$ path in $G$ for $i=1, 2$. 
Therefore, we obtain 
\begin{equation}
\label{eq:findQ01}
\newweight{e_1} = \dist{r_2}{v_1}{Q} = 
\begin{cases}
\frac{1}{2} \newweight{e}  & \mbox{if $r_2 = p_e$,} \\
\dist{r_2}{v_2}{G} + \newweight{e}  & \mbox{otherwise,} 
\end{cases}
\end{equation}
and 
\begin{equation}
\label{eq:findQ02}
\newweight{e_2} = \dist{r_2}{v_1}{Q} = 
\begin{cases}
\frac{1}{2} \newweight{e}  & \mbox{if $r_1 = p_e$,} \\
\dist{r_1}{v_1}{G} + \newweight{e}  & \mbox{otherwise.} 
\end{cases}
\end{equation}
Note that the right-hand sides of (\ref{eq:findQ01}) and (\ref{eq:findQ02}) can be computed 
without $Q$ if we have $v_1$ and $v_2$ in hand.

Therefore, the following algorithm correctly finds a desired spanning tree $Q$: 
guess two vertices $v_1$ and $v_2$ in $V$, 
construct $G^+$ where 
$\newweight{e_1}$ and $\newweight{e_2}$ are defined as (\ref{eq:findQ01}) and (\ref{eq:findQ02}), 
compute the shortest path tree $T$ from $r$ in $G^+$, define $Q:= T - \{e_1, e_2\} + e$, and 
check whether $Q$ satisfies the constraints or not. 
See Algorithm~\ref{alg:03} for a pseudocode of our algorithm. 
Since the number of choices of $v_1$ and $v_2$ is $O(|E|)$, this algorithm runs in polynomial time.

\begin{algorithm}
    \KwInput{A graph $G$ with two distinct points $r_1$ and $r_2$ in $V \cup R$. }
    \KwOutput{A spanning tree $Q$ such that $\reconf{r_1}{r_2}{Q}$ and $\tuple{r_1, r_2, Q}$ is good. }
    \For{$v_1, v_2 \in V$ with $e = v_1 v_2 \in E$}{
    Construct $G^+$\; 
    Define $\newweight{e_1}$ and $\newweight{e_2}$ as (\ref{eq:findQ01}) and (\ref{eq:findQ02})\;  
    Compute the shortest path tree $T$ from $r$ in $G^+$\; 
    $Q \coloneqq T - \{e_1, e_2\} + e$\; 
        \lIf{$Q$ is a spanning tree such that $\reconf{r_1}{r_2}{Q}$ and $\tuple{r_1, r_2, Q}$ is good}{
        \Return{$Q$}
        }
    }
    \Return{no solution exists}\;
    \caption{Algorithm for finding a spanning tree $Q$}\label{alg:03}
\end{algorithm}

\subsection{\texorpdfstring{Finding a pseudotree $Q$ with a cycle}{Finding a pseudotree Q with a cycle}}
\label{sec:Qhascycle}

We next consider the case when $Q$ contains a cycle. 
Suppose that $Q$ is a pseudotree with a cycle $\UniC{Q}$ such that $\reconf{r_1}{r_2}{Q}$ and $\tuple{r_1, r_2, Q}$ is good. 
Since $r_1$ and $r_2$ are on $\UniC{Q}$, there exist two edges 
$e= v_1 v_2$ and $e'= v'_1 v'_2$ in $E(\UniC{Q})$ such that 
$\dist{r_1}{v_1}{Q} < \dist{r_2}{v_1}{Q}$, $\dist{r_1}{v_2}{Q} > \dist{r_2}{v_2}{Q}$,  
$\dist{r_1}{v'_1}{Q} < \dist{r_2}{v'_1}{Q}$, and $\dist{r_1}{v'_2}{Q} > \dist{r_2}{v'_2}{Q}$. 
That is, each of $e$ and $e'$ is the edge containing the middle point of an $r_1$-$r_2$ path, where we note that 
$Q$ contains two $r_1$-$r_2$ paths. 
Let $Q_1$ and $Q_2$ be the connected components of $Q-\{e, e'\}$ such that $v_i, v'_i \in V(Q_i)$ for $i=1,2$. 
Define $G^+$ as the graph obtained from $G$ 
by adding a new vertex $r$ together with four new edges $e_1 = r v_1$, $e_2 = r v_2$, $e'_1 = r v'_1$, and $e'_2 = r v'_2$ (Figure~\ref{fig:702}). 
Set 
\begin{align}
&
\newweight{e_1} := 
\begin{cases}
\frac{1}{2} \newweight{e}  & \mbox{if $r_2 = p_e$,} \\
\dist{r_2}{v_2}{\UniC{Q}} + \newweight{e}  & \mbox{otherwise,} 
\end{cases}
\label{eq:defineell01}
\\
&
\newweight{e_2} := 
\begin{cases}
\frac{1}{2} \newweight{e}  & \mbox{if $r_1 = p_e$,} \\
\dist{r_1}{v_1}{\UniC{Q}} + \newweight{e}  & \mbox{otherwise,} 
\end{cases} 
\\
&
\newweight{e'_1} := 
\begin{cases}
\frac{1}{2} \newweight{e'}  & \mbox{if $r_2 = p_{e'}$,} \\
\dist{r_2}{v'_2}{\UniC{Q}} + \newweight{e'}  & \mbox{otherwise,} 
\end{cases}
\\
&
\newweight{e'_2} := 
\begin{cases}
\frac{1}{2} \newweight{e'}  & \mbox{if $r_1 = p_{e'}$,} \\
\dist{r_1}{v'_1}{\UniC{Q}} + \newweight{e'}  & \mbox{otherwise.} 
\end{cases} 
\label{eq:defineell04}
\end{align}
Note that $\newweight{e_1}$  (resp.~$\newweight{e'_1}$, $\newweight{e_2}$, $\newweight{e'_2}$) 
is the length of an $r_2$-$v_1$ path (resp.~an $r_2$-$v'_1$ path, an $r_1$-$v_2$ path, an $r_1$-$v'_2$ path) in $Q$.  
We now show the following claim. 

\begin{figure}[t!]
      \centering

\begin{tikzpicture}

\tikzset{tvertex/.style={auto=left,circle,fill=blue,minimum size=5pt,inner sep=0pt}} 
\tikzset{tleaf/.style={}} 
\tikzset{tedge/.style={blue}} 
    \pgfmathsetmacro{\width}{2.5}
    \pgfmathsetmacro{\height}{0.5*\width}
    \draw (0,0) ellipse ({\width} and {\height});
    \node[label = {[label distance=-2.5mm] -95:$r_1$}] at ($(0,0)+(-135:{\width} and {\height})$) (r1) {$\times$};
    \node[label = {[label distance=-2mm] -85:$r_2$}] at ($(0,0)+(-45:{\width} and {\height})$) (r2) {$\times$};

    \draw[thick, green] (0,0) [partial ellipse=80:100:{\width} and {\height}];
    \draw[thick, green] (0,0) [partial ellipse=-80:-100:{\width} and {\height}];

    \node[label = {[label distance=-1mm] -90:$e$}] at ($(0,0)+(-90:{\width} and {\height})$) (elabel)  {};
    \node[label = {[label distance=-1mm] 90:$e'$}] at ($(0,0)+(90:{\width} and {\height})$) (eplabel) {}; 

    \node[tvertex, label = {[label distance=-1mm] -95:$v_1$}] at ($(0,0)+(-100:{\width} and {\height})$) (v1) {}; 
    \node[tvertex, label = {[label distance=-1mm] -85:$v_2$}] at ($(0,0)+(-80:{\width} and {\height})$) (v2)  {};

    \node[tvertex, label = {[label distance=-1mm] 95:$v'_1$}] at ($(0,0)+(100:{\width} and {\height})$) (vp1) {}; 
    \node[tvertex, label = {[label distance=-1mm] 85:$v'_2$}] at ($(0,0)+(80:{\width} and {\height})$) (vp2)  {};

    \node[tvertex, label = {180:$r$}] at (0, 0) (r) {}; 

    \draw (r) -- node[midway, left] {$e_1$} (v1); 
    \draw (r) -- node[midway, right] {$e_2$} (v2); 
    \draw (r) -- node[midway, left] {$e'_1$} (vp1); 
    \draw (r) -- node[midway, right] {$e'_2$} (vp2); 

    \node[minimum size=0pt, inner sep=0pt] at ($(0,0)+( 50:{\width} and {\height})$) (t1) {}; 
    \node[minimum size=0pt, inner sep=0pt] at ($(0,0)+(150:{\width} and {\height})$) (t2) {}; 
    \node[minimum size=0pt, inner sep=0pt] at ($(0,0)+(210:{\width} and {\height})$) (t3) {}; 
    \node[minimum size=0pt, inner sep=0pt] at ($(0,0)+(300:{\width} and {\height})$) (t4) {}; 
    \node[minimum size=0pt, inner sep=0pt] at ($(t2) + (135:2em)$) (t2c) {}; 
    \node[minimum size=0pt, inner sep=0pt] at ($(t3) + (215:2em)$) (t3c) {}; 
    \node[minimum size=0pt, inner sep=0pt] at ($(t4) + (270:2em)$) (t4c) {}; 

    \draw (t1) -- +(30:2em); 
    \draw (t1) -- +(70:2em); 

    \draw  (t2) -- (t2c); 
    \draw (t2c) -- +( 90:1em); 
    \draw (t2c) -- +(210:1em); 

    \draw  (t3) -- (t3c); 
    \draw (t3c) -- +(190:1em); 
    \draw (t3c) -- +(250:1em); 

    \draw  (t4) -- (t4c); 
    \draw (t4c) -- +(260:1em); 
    \draw (t4c) -- +(340:1em); 
\end{tikzpicture}%

      \caption{Construction of $G^+$.}
      \label{fig:702}
\end{figure}


\begin{claim}
For each $v\in V$, it holds that $\vlabel{v}{Q} = \dist{r}{v}{G^+}$. 
Furthermore, $Q - \{e, e'\} + \{e_1, e_2, e'_1, e'_2\}$ contains the shortest path tree starting from $r$ in $G^+$.  
\end{claim}

\begin{proof}
Since 
$\vlabel{v}{Q} = \dist{r_2}{v}{Q} =  \min \{ \newweight{e_1} + \dist{v_1}{v}{Q_1},\  \newweight{e'_1} + \dist{v'_1}{v}{Q_1} \}$ for $v \in V(Q_1)$ and 
$\vlabel{v}{Q} = \dist{r_1}{v}{Q} = \min \{ \newweight{e_2} + \dist{v_2}{v}{Q_2},\  \newweight{e'_2} + \newweight{e'} + \dist{v'_2}{v}{Q_2} \}$ for $v \in V(Q_2)$,  
for each $v \in V$, $Q - \{e, e'\} + \{e_1, e_2, e'_1, e'_2\}$ contains an $r$-$v$ path $P_v$
whose length is $\vlabel{v}{Q}$. 
Furthermore, since $\vlabel{v}{Q} \le \vlabel{u}{Q} + \newweight{uv}$ for any $uv \in E$, 
$\vlabel{v}{Q}$ is equal to the shortest path length from $r$ to $v$ in $G^+$, that is, 
$P_v$ is the unique shortest $r$-$v$ path in $G^+$. 
Since the shortest path tree is the union of $P_v$'s, 
it is contained in $Q - \{e, e'\} + \{e_1, e_2, e'_1, e'_2\}$. 
\end{proof}

This claim shows that
there exist two edges $f_1= u_1 u'_1$ and $f_2= u_2 u'_2$ in $E(Q) \cup \{e_1, e_2, e'_1, e'_2\}$ such that 
$Q + \{e_1, e_2, e'_1, e'_2\} - \{e, e', f_1, f_2\}$ is 
the shortest path tree starting from $r$ in $G^+$. 
Without loss of generality, we may assume that $f_i \in E(Q_i) \cup \{e_i, e'_i\}$ for $i=1, 2$, because 
$E(Q_i) \cup \{e_i, e'_i\}$ contains a cycle.
Furthermore, for $i=1, 2$, we may assume that $v_i, u_i, u'_i$, and $v'_i$ appear in this order along $\UniC{Q}$ if $f_i \in E(Q_i)$. 
For $i=1,2$, we observe the following. 
\begin{itemize}
    \item If $f_i \in \{e_i, e'_i\}$, then $Q_i[v_i, v'_i]$ is the shortest $v_i$-$v'_i$ path in $G$, because it is a subpath of $P_{v_i}$ or $P_{v'_i}$, where 
    $P_v$ is the unique shortest $r$-$v$ path in $G^+$ for $v \in V$. 
    \item If $f_i \in E(Q_i)$, then $Q_i[v_i, u_i]$  is the shortest $v_i$-$u_i$ path in $G$ because it is a subpath of $P_{u_i}$,  
    and $Q_i[v'_i, u'_i]$ is the shortest $v'_i$-$u'_i$ path in $G$ because it is a subpath of $P_{u'_i}$ (Figure~\ref{fig:703}).
\end{itemize}

\begin{figure}[t!]
      \centering

\begin{tikzpicture}

    \tikzset{tvertex/.style={auto=left,circle,fill=blue,minimum size=5pt,inner sep=0pt}} 
    \tikzset{tleaf/.style={}} 
    \tikzset{tedge/.style={blue}} 
        \pgfmathsetmacro{\width}{2.5}
        \pgfmathsetmacro{\height}{0.5*\width}
        \draw (0,0) ellipse ({\width} and {\height});
        \node[label = {[label distance=-2mm] -100:$r_1$}] at ($(0,0)+(-135:{\width} and {\height})$) (r1) {$\times$};
        \node[label = {[label distance=-2mm] -80:$r_2$}] at ($(0,0)+(-45:{\width} and {\height})$) (r2) {$\times$};

        \node[tvertex, label={180:$r$}] at (0, 0) (r) {}; 
    
        \draw[thick, green] (0,0) [partial ellipse=160:260:{\width} and {\height}];
        \draw[thick, green] (0,0) [partial ellipse=-80:20:{\width} and {\height}];
        \draw[thick, green] (0,0) [partial ellipse=100:140:{\width} and {\height}];
        \draw[thick, green] (0,0) [partial ellipse=40:80:{\width} and {\height}];
        \draw[thick, green] ($(0,0)+(-100:{\width} and {\height})$) -- (r); 
        \draw[thick, green] ($(0,0)+(-80:{\width} and {\height})$) -- (r); 
        \draw[thick, green] ($(0,0)+(100:{\width} and {\height})$) -- (r); 
        \draw[thick, green] ($(0,0)+(80:{\width} and {\height})$)  -- (r);

        \node[tvertex, label={-90:$v_1$}] at ($(0,0)+(260:{\width} and {\height})$) (v1) {}; 
        \node[tvertex, label={-90:$v_2$}] at ($(0,0)+(280:{\width} and {\height})$) (v2)  {};
    
        \node[tvertex, label={90:$v'_1$}] at ($(0,0)+(100:{\width} and {\height})$) (vp1) {}; 
        \node[tvertex, label={90:$v'_2$}] at ($(0,0)+(80:{\width} and {\height})$) (vp2)  {};

        \node[tvertex, label={160:$u_1$}]  at ($(0,0)+(160:{\width} and {\height})$) (u1) {}; 
        \node[tvertex, label={140:$u'_1$}] at ($(0,0)+(140:{\width} and {\height})$) (up1) {}; 
        \node[label={[label distance=-5pt]-45:$f_1$}] at ($(0,0)+(150:{\width} and {\height})$) (f1label) {}; 
        \node[tvertex, label={40:$u'_2$}] at ($(0,0)+(40:{\width} and {\height})$)  (up2)  {};
        \node[tvertex, label={20:$u_2$}]  at ($(0,0)+(20:{\width} and {\height})$)  (u2)  {};
        \node[label={[label distance=-5pt]-135:$f_2$}] at ($(0,0)+(30:{\width} and {\height})$) (f2label) {};

        \draw[thick, green] (r) -- (v1); 
        \draw[thick, green] (r) -- (v2); 
        \draw[thick, green] (r) -- (vp1); 
        \draw[thick, green] (r) -- (vp2); 
    
        \node[label = {[label distance=-2mm] 100:$P_{u_2}$}] at ($(0,0)+(-50:{\width} and {\height})$) (pu2) {};
        \node[label = {[label distance=-2mm] 80:$P_{u_1}$}] at ($(0,0)+(-130:{\width} and {\height})$) (pu1) {};

        \node[label = {-90:$P_{u'_2}$}] at ($(0,0)+(75:{\width} and {\height})$) (pup2) {};
        \node[label = {-90:$P_{u'_1}$}] at ($(0,0)+(105:{\width} and {\height})$) (pup1) {};
    \end{tikzpicture}%

      \caption{Paths $P_{u_1}$, $P_{u'_1}$, $P_{u_2}$, and $P_{u'_2}$.}
      \label{fig:703}
\end{figure}

We are now ready to describe our algorithm. 
In our algorithm, we first guess four vertices $v_1$, $v_2$, $v'_1$, and $v_2$ in $V$ with $e=v_1 v_2 \in E$ and $e'=v'_1 v'_2 \in E$, construct $G^+$, and 
guess four vertices $u_1$, $u'_1$, $u_2$, and $u'_2$ in $V$ with $f_1= u_1 u'_1 \in E \cup \{e_1, e'_1\}$ and $f_2= u_2 u'_2 \in E \cup \{e_2, e'_2\}$. 
For $i=1, 2$, let $J_i$ be the $v_i$-$v'_i$ walk defined by  
\begin{equation}
    J_i := 
    \begin{cases}
    P^*_{v_i, v'_i} & \mbox{if $f_i \in \{e_i, e'_i\}$,} \\
    P^*_{v_i, u_i} \circ \{f_i\} \circ P^*_{u'_i, v'_i} & \mbox{if $f_i \in E$,}
    \end{cases} 
\label{eq:defJ}
\end{equation} 
where $P^*_{x, y}$ denotes the shortest $x$-$y$ path in $G$ for $x,y \in V$ and $\circ$ denotes the concatenation of walks. 
Define a closed walk $C$ as
\begin{equation}
    C := J_1 \circ \{e'\} \circ \overline{J_2} \circ \{e\}, 
\label{eq:defC}
\end{equation} 
where $\overline{J_2}$ is the reverse walk of $J_2$. 
Note that if a desired pseudotree $Q$ exists and if $v_1, v_2, v'_1, v'_2, u_1, u'_1, u_2$, and $u'_2$ are guessed correctly, then 
$C = \UniC{Q}$ holds by the above arguments. 
Therefore, it suffices to consider the case when $C$ is a cycle containing $r_1$ and $r_2$. 
We define $\ell(e_1), \ell(e_2), \ell(e'_1)$, and $\ell(e'_2)$ as 
(\ref{eq:defineell01})--(\ref{eq:defineell04}) in which $\UniC{Q}$ is replaced with $C$. 
Finally, we compute the shortest path tree $T$ from $r$ in $G^+$, define $Q:= (T + \{e, e', f_1, f_2\}) - \{e_1, e_2, e'_1, e'_2\}$, and 
check whether $Q$ satisfies the constraints or not. 
See Algorithm~\ref{alg:04} for a pseudocode of our algorithm. 

By the above arguments, this algorithm finds a desired pseudotree $Q$ if $v_1$, $v_2$, $v'_1$, $v'_2$, $u_1$, $u'_1$, $u_2$, and $u'_2$ are guessed correctly, 
which shows the correctness of the algorithm. 
Since the number of choices of these vertices is $O(|E|^4)$, this algorithm runs in polynomial time.

\begin{algorithm}
    \KwInput{ A graph $G$ with two distinct points $r_1$ and $r_2$ in $V \cup R$. }
    \KwOutput{ A pseudotree $Q$ with a cycle such that $\reconf{r_1}{r_2}{Q}$ and $\tuple{r_1, r_2, Q}$ is good. }
    \For{$v_1, v_2, v'_1, v'_2 \in V$ with $e = v_1 v_2 \in E $ and $e'= v'_1 v'_2 \in E$}{
        Construct $G^+$\;
        \For{$u_1, u'_1, u_2, u'_2 \in V$ with $f_1 = u_1 u'_1 \in E \cup \{e_1, e'_1\}$ and $f_2 = u_2 u'_2 \in E \cup \{e_2, e'_2\}$}{
            Define $J_1, J_2$, and $C$ as (\ref{eq:defJ}) and (\ref{eq:defC})\; 
            \If{$C$ is a cycle containing $r_1$ and $r_2$}{
                Define $\ell(e_1), \ell(e_2), \ell(e'_1)$, and $\ell(e'_2)$ as 
                    (\ref{eq:defineell01})--(\ref{eq:defineell04}) in which $\UniC{Q}$ is replaced with $C$\;  
                Compute the shortest path tree $T$ from $r$ in $G^+$\; 
                $Q:= (T + \{e, e', f_1, f_2\}) - \{e_1, e_2, e'_1, e'_2\}$\; 
                \lIf{$Q$ is a pseudotree with a cycle such that $\reconf{r_1}{r_2}{Q}$ and $\tuple{r_1, r_2, Q}$ is good}{
                    \Return{$Q$}
                }
            }
        }
    }
    \Return{no solution exists}\;
    \caption{Algorithm for finding a pseudotree $Q$}\label{alg:04}
\end{algorithm}

\section{Proof of Theorem~\ref{thm:good:sequence}}
\label{sec:goodsequence}

For points $r_1, r_2\in V \cup R$ and a pseudotree $Q$ with $\reconf{r_1}{r_2}{Q}$, we define $\f{r_1, r_2, Q} :=  (\newweight{\UniC{Q}}, \weight{r_1, r_2, Q})$, where $\weight{r_1, r_2, Q} = \sum_{v\in V}\vlabel{v}{r_1, r_2, Q}$.
Here, we denote $\newweight{\UniC{Q}} := \sum_{e \in E(\UniC{Q})} \newweight{e}$ if $Q$ contains a cycle and $\newweight{\UniC{Q}} := {\bf 0}$ otherwise.  
For two triplets $\tuple{r_1, r_2, Q}$ and $\tuple{r'_1, r'_2, Q'}$, we compare $\f{r_1, r_2, Q}$ and $\f{r'_1, r'_2, Q'}$ lexicographically, that is, 
we denote $\f{r_1, r_2, Q} < \f{r'_1, r'_2, Q'}$ if and only if either $\newweight{\UniC{Q}} < \newweight{\UniC{Q'}}$ or $\newweight{\UniC{Q}} = \newweight{\UniC{Q'}}$ and $\weight{r_1, r_2, Q} < \weight{r'_1, r'_2, Q'}$. 

To show Theorem~\ref{thm:good:sequence}, we apply the induction on $\f{r_1, r_2, Q}$ by using the 
following proposition, which roughly says that we can reduce to the case with smaller $\f{\cdot}$ if $\tuple{r_1, r_2, Q}$ is \emph{not} good.



\begin{proposition}
            \label{prop:getting:better}
      Let $r_1, r_2\in V \cup R$ be points and $Q$ be a pseudotree with $\reconf{r_1}{r_2}{Q}$. 
      If $\tuple{r_1, r_2, Q}$ is not good,
      then one of the followings holds:
      \newlist{enumcond}{enumerate}{1}
      \setlist[enumcond,1]{label=\textbf{B\arabic*.}, ref=\textbf{B\arabic*}}
      \begin{enumcond}[leftmargin=*]
            \item there exists a pseudotree $Q'$ such that $\reconf{r_1}{r_2}{Q'}$ with $\f{r_1, r_2, Q'} < \f{r_1, r_2, Q}$ or  \label{prop:stmt:case:1}
            \item there exist a point $r_0 \in V \cup R$ and  pseudotrees $Q_1, Q_2$ such that \label{prop:stmt:case:2}
                  \begin{itemize}
                        \item $\reconf{r_1}{r_0}{Q_1}$ with $\f{r_1, r_0, Q_1} < \f{r_1, r_2, Q}$ and
                        \item $\reconf{r_0}{r_2}{Q_2}$ with $\f{r_0, r_2, Q_2} < \f{r_1, r_2, Q}$.
                  \end{itemize}
      \end{enumcond}
\end{proposition}
\begin{proof}
      Suppose that $(r_1, r_2, Q)$ is not good.

      We first consider the case when $Q$ is a spanning tree.  
      Let $r$ be the middle point of $r_1$ and $r_2$ with respect to $\newweight{\cdot}$. Then, $\vlabel{v}{Q}= \dist{r_1}{r}{Q} + \dist{r}{v}{Q}$ for $v \in V$. 
      Since $(r_1, r_2, Q)$ is not good, there exists an edge $uv \in E$ such that $\vlabel{u}{Q} + \newweight{uv} < \vlabel{v}{Q}$. 
      Let $e$ be the unique edge on $Q[r,v]$ that is incident to $v$.  
      Then, $Q':= Q - \{ e \} + \{uv\}$ is a spanning tree such that $\dist{r}{x}{Q'} \le \dist{r}{x}{Q}$ for any $x \in V$. 
      Therefore, 
      $$
      \dist{r_i}{x}{Q'}  \le  \dist{r_i}{r}{Q'} + \dist{r}{x}{Q'}  \le \dist{r_i}{r}{Q} + \dist{r}{x}{Q} \\ 
                           =  \max\{\dist{r_1}{x}{Q},\ \dist{r_2}{x}{Q} \}  
      $$
      for any $x \in V$ and for $i=1,2$. 
      Since $\reconf{r_1}{r_2}{Q}$ implies that $\distorig{r_i}{x}{Q} \le \frac{d}{2}$, we obtain $\distorig{r_i}{x}{Q'} \le \frac{d}{2}$ for $x \in V$ and for $i=1,2$. 
      This shows that 
      $\reconf{r_1}{r_2}{Q'}$. 
      Furthermore, since $\vlabel{x}{Q'} \le \vlabel{x}{Q}$ for any $x \in V$ and $\vlabel{v}{Q'} < \vlabel{v}{Q}$, we have that $\weight{r_1, r_2, Q'} < \weight{r_1, r_2, Q}$ and $\f{r_1, r_2, Q'} < \f{r_1, r_2, Q}$. 
      Thus, $Q'$ satisfies the conditions in \cref{prop:stmt:case:1}.

      We next consider the case when $Q$ contains a cycle $\UniC{Q}$ and at least one of 
      $r_1$ and $r_2$ is not contained in $\UniC{Q}$. 
      Without loss of generality, we assume that $r_1 \notin V(\UniC{Q})$. 
      Define $T$ as the shortest path tree in $Q$ starting from $r_1$. 
      Let $\subT{Q}{r_1}$ be the connected component of $Q-E(\UniC{Q})$ that contains $r_1$. 
      If $r_2 \in V(\subT{Q}{r_1})$, then 
      $\vlabel{x}{T} = \vlabel{x}{Q}$ for any $x \in V$ and $\f{r_1, r_2, T} < \f{r_1, r_2, Q}$, and hence $Q':=T$ satisfies the conditions in \cref{prop:stmt:case:1}. 
      Otherwise, let $r_0$ be the unique vertex in $V(\subT{Q}{r_1}) \cap V(\UniC{Q})$, where we note that $r_0\not=r_1$ and $r_0 \not= r_2$ (Figure~\ref{fig:801}). 
      Since $\dist{r_0}{x}{Q} < \dist{r_2}{x}{Q}$ for $x \in V(\subT{Q}{r_1})$ and $\dist{r_0}{x}{Q} < \dist{r_1}{x}{Q}$ for $x \in V \setminus V(\subT{Q}{r_1})$, 
      it holds that $\dist{r_0}{x}{Q} < \max\{\dist{r_1}{x}{Q}, \ \dist{r_2}{x}{Q} \} = \vlabel{x}{Q}$ for $x \in V(\subT{Q}{r_1})$. 
      Therefore, $\distorig{r_0}{x}{Q} \le \frac{d}{2}$ for any $x \in V$, and hence $\reconf{r_1}{r_0}{Q}$ and $\reconf{r_0}{r_2}{Q}$. 
      Furthermore, since $\weight{r_1, r_0, Q} < \weight{r_1, r_2, Q}$ and $\weight{r_0, r_2, Q} < \weight{r_1, r_2, Q}$, 
      we obtain $\f{r_1, r_0, Q} < \f{r_1, r_2, Q}$ and $\f{r_0, r_2, Q} < \f{r_1, r_2, Q}$.  
      This shows that $Q_1:=Q_2:=Q$ and $r_0$ satisfy the conditions in \cref{prop:stmt:case:2}.
      
      \begin{figure}[t!]
      \centering

\begin{tikzpicture}

\tikzset{tvertex/.style={auto=left,circle,fill=blue,minimum size=5pt,inner sep=0pt}} 
\tikzset{tleaf/.style={}} 
\tikzset{tedge/.style={blue}} 
    \pgfmathsetmacro{\width}{2.5}
    \pgfmathsetmacro{\height}{0.5*\width}
    \draw (0,0) ellipse ({\width} and {\height});

    \node[fill=red, circle, minimum size=5pt, inner sep=0pt, label={0:$r_2$}] at ($(0,0)+( 50:{\width} and {\height})$) (t1) {}; 
    \node[minimum size=0pt, inner sep=0pt] at ($(0,0)+( 20:{\width} and {\height})$) (t2) {}; 
    \node[minimum size=0pt, inner sep=0pt] at ($(0,0)+(-70:{\width} and {\height})$) (t3) {};

    \node[fill=red, circle, minimum size=5pt, inner sep=0pt, label={0:$r_0$}] at ($(0,0)+(190:{\width} and {\height})$) (r0) {}; 
    \node[minimum size=0pt, inner sep=0pt] at ([shift=({180:1.5em})]r0) (r0c1) {}; 
    \node[fill=red, circle, minimum size=5pt, inner sep=0pt, label={[label distance=-0.5mm] -90:$r_1$}]  at ([shift=({135:1.5em})]r0c1) (r0c2) {}; 
    \node[minimum size=0pt, inner sep=0pt] at ([shift=({250:1.5em})]r0c1) (r0c3) {}; 

    \node[minimum size=0pt, inner sep=0pt] at ($(t1) + (50:2em)$) (t1c) {}; 
    \node[minimum size=0pt, inner sep=0pt] at ($(t2) + (20:2em)$) (t2c) {}; 
    \node[minimum size=0pt, inner sep=0pt] at ($(t3) + (-60:1em)$) (t3c) {}; 

    \node[minimum size=0pt, inner sep=0pt,label={-30:$C_Q$}] at ($(0,0) + (-30:{\width} and {\height})$) (cqlabel) {}; 

    \draw  (t1) -- (t1c); 
    \draw (t1c) -- +(90:1em); 
    \draw (t1c) -- +(20:1em); 

    \draw  (t2) -- (t2c); 
    \draw (t2c) -- +( 45:1em); 
    \draw (t2c) -- +(-10:1em); 

    \draw  (t3) -- (t3c); 
    \draw (t3c) -- +(-30:1em); 
    \draw (t3c) -- +(-100:1em); 

    \node[minimum size=0pt, inner sep=0pt] at ([shift=({80:1.5em})]r0c2) (r0l1) {}; 
    \node[minimum size=0pt, inner sep=0pt] at ([shift=({175:1.5em})]r0c2) (r0l2) {}; 
    \node[minimum size=0pt, inner sep=0pt] at ([shift=({225:1.5em})]r0c3) (r0l3) {}; 
    \node[minimum size=0pt, inner sep=0pt] at ([shift=({290:1.5em})]r0c3) (r0l4) {}; 
    \draw (r0) -- (r0c1); 
    \draw (r0c1) -- (r0c2); 
    \draw (r0c1) -- (r0c3); 
    \draw (r0c2) -- +(80:1em); 
    \draw (r0c2) -- +(175:1em); 
    \draw (r0c3) -- +(225:1em); 
    \draw (r0c3) -- +(290:1em); 

    \node[minimum size=0pt, inner sep=0pt] at ([shift=({315:1.5em})]r0l4) (qr1label) {$Q_{r_1}$}; 
    \draw[dashed] \convexpath{r0,r0l4,r0l3,r0l2,r0l1}{4pt}; 
\end{tikzpicture}%

      \caption{Case when $r_1 \notin V(\UniC{Q})$ and $r_2 \not\in V(\subT{Q}{r_1})$.}
      \label{fig:801}
      \end{figure}
      
      In what follows, we consider the case when $Q$ contains a cycle $\UniC{Q}$ and both $r_1$ and $r_2$ are contained in $\UniC{Q}$.
      Since $(r_1, r_2, Q)$ is not good, there exists an edge $uv \in E(Q)$ such that $\vlabel{v}{Q} > \vlabel{u}{Q} + \newweight{uv}$. 
      If $v \notin V(\UniC{Q})$, then 
      $Q' := Q - \set{e} + \set{uv}$ satisfies the conditions in \cref{prop:stmt:case:1}, where $e$ is the unique edge on $Q[r_1, v]$ that is incident to $v$. 
      This is because $\UniC{Q} = \UniC{Q'}$ and 
      $\vlabel{x}{Q'} < \vlabel{x}{Q}$ for any $x \in V$ with $v \in V(Q[r_1, x])$ (see Figure~\ref{fig:802}). 
      Thus, it suffices to consider the case when $v \in V(\UniC{Q})$. 
      For $i=1,2$, let $q_i$ be the farthest point on $\UniC{Q}$ from $r_i$, that is, $\dist{r_i}{q_i}{Q} = \frac{1}{2} \newweight{\UniC{Q}}$, 
      and let $e_i \in E(Q)$ be the edge containing $q_i$. 
      Then, $Q-e_i$ is the shortest path tree in $Q$ starting from $r_i$. 
      We may assume that $e_1\not= e_2$, since otherwise $Q':=Q-e_1$ satisfies the conditions in \cref{prop:stmt:case:1}.
      Let $g_1$ and $g_2$ be the points in $\UniC{Q}$ such that $\dist{g_i}{r_1}{Q} = \dist{g_i}{r_2}{Q} =: \ell_i$ for $i = 1, 2$ (Figure~\ref{fig:803}). 
      Then, we observe that 
      \begin{equation}
      \label{eq:labelbymiddlepoint}
          \vlabel{x}{Q} = \max \{\dist{r_1}{x}{Q},\ \dist{r_2}{x}{Q} \} =  \min \{ \ell_1 + \dist{g_1}{x}{Q},\ \ell_2 + \dist{g_2}{x}{Q} \}
      \end{equation}
      for $x \in V$.
      Let $w$ be the vertex in $V(\subT{Q}{u}) \cap V(\UniC{Q})$, where $\subT{Q}{u}$ is the connected component of $Q-E(\UniC{Q})$ containing $u$. 
      Note that if $u$ is in $V(\UniC{Q})$, then $w = u$.
      We now show the following claim, which is useful in our case analysis. 

\begin{figure}[t!]
      \centering

\begin{tikzpicture}

\tikzset{tvertex/.style={auto=left,circle,fill=blue,minimum size=5pt,inner sep=0pt}} 
\tikzset{tleaf/.style={}} 
\tikzset{tedge/.style={blue}} 
    \pgfmathsetmacro{\width}{2.5}
    \pgfmathsetmacro{\height}{0.5*\width}
    \draw (0,0) ellipse ({\width} and {\height});

    \node[label = {[label distance=-2mm] -100:$r_1$}] at ($(0,0)+(-135:{\width} and {\height})$) (r1) {$\times$};
    \node[label = {[label distance=-2mm] -80:$r_2$}] at ($(0,0)+(-45:{\width} and {\height})$) (r2) {$\times$};

    \node[minimum size=0pt, inner sep=0pt] at ($(0,0)+(190:{\width} and {\height})$) (r0) {}; 
    \node[minimum size=0pt, inner sep=0pt] at ([shift=({180:1.5em})]r0) (r0c1) {}; 
    \node[fill=black, circle, minimum size=5pt, inner sep=0pt, label={[label distance=-1mm] 100:$v$}] at ([shift=({115:1.5em})]r0c1) (r0c2) {}; 
    \node[minimum size=0pt, inner sep=0pt] at ([shift=({135:1.5em})]r0c2) (r0c4) {}; 
    \node[minimum size=0pt, inner sep=0pt] at ([shift=({200:1.5em})]r0c1) (r0c3) {};

    \node[minimum size=0pt, inner sep=0pt,label={-30:$C_Q$}] at ($(0,0) + (-30:{\width} and {\height})$) (cqlabel) {};

    \node[minimum size=0pt, inner sep=0pt] at ([shift=({80:1.5em})]r0c2) (r0l1) {}; 
    \node[minimum size=0pt, inner sep=0pt] at ([shift=({175:1.5em})]r0c2) (r0c22) {}; 
        \node[minimum size=0pt, inner sep=0pt] at ([shift=({80:1.5em})]r0c22) (r0l21) {}; 
        \node[minimum size=0pt, inner sep=0pt] at ([shift=({165:1.5em})]r0c22) (r0l22) {}; 
    \node[minimum size=0pt, inner sep=0pt] at ([shift=({225:1.5em})]r0c3) (r0l3) {}; 
    \node[minimum size=0pt, inner sep=0pt] at ([shift=({290:1.5em})]r0c3) (r0l4) {}; 
    \draw (r0) -- (r0c1); 
    \draw [red] (r0c1) -- node[midway, left]{$e$} (r0c2); 
    \draw (r0c1) -- (r0c3); 
    \draw (r0c2) -- (r0l1); 
    \draw (r0c2) -- (r0c22); 
    \draw (r0c3) -- (r0l3); 
    \draw (r0c3) -- (r0l4); 

    \draw (r0c22) -- (r0l21); 
    \draw (r0c22) -- (r0l22); 

    \node[fill=blue, circle, minimum size=5pt, inner sep=0pt, label={[text=blue] 0:$u$}] at ($(0,0)+(170:{\width} and {\height})$) (u) {}; 
    \draw[dashed, blue] (r0c2) -- (u); 

\end{tikzpicture}%

      \caption{Case when $v \notin V(\UniC{Q})$.}
      \label{fig:802}
\end{figure}
\begin{figure}[t!]
      \centering
    \begin{tikzpicture}
    \basetree
\end{tikzpicture}%

      \caption{Positions of $q_i$, $e_i$, and $g_i$.}
      \label{fig:803}
\end{figure}
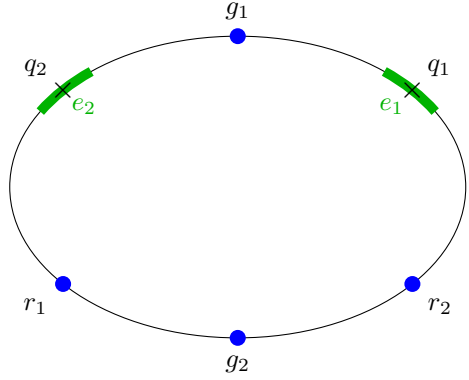

      \begin{claim}
            \label{prop:getting:better:claim:for:cond:two}
            Suppose that $u, v, e_1$, and $e_2$ are defined as above. If $Q' := Q - \set{e_1, e_2} + \set{uv}$ is connected,
            then \cref{prop:stmt:case:1} or \cref{prop:stmt:case:2} holds. 
      \end{claim}
      \begin{claimproof}
            \begin{figure}[t!]
                  \newcommand{\picscale}{.55}
                  \centering
                  \begin{minipage}{0.33\hsize}
                        \centering
                         \scalebox{\picscale}{ %
    \begin{tikzpicture}
    \coordinate (c) at (0, 0); 
    \def\cRad{3cm and 2cm} 

    \def\cRone{220} 
    \def\cRtwo{320} 

    \def\cGone{90}
    \def\cGtwo{270} 

    \def\cEone{40} 
    \def\cEtwo{140} 

    \def\EdgeLength{10} 

    \draw (c) [partial ellipse=0:360:\cRad];

    \draw[myedge] (c) [partial ellipse=\cEone - \EdgeLength:\cEone + \EdgeLength:\cRad];
    \node[label={[text=black!30!green,label distance=-5pt]\cEone-180:{$e_1$}}] (e1) at ($(c)+(\cEone:\cRad)$) {};
    \draw[myedge] (c) [partial ellipse=\cEtwo - \EdgeLength:\cEtwo + \EdgeLength:\cRad];
    \node[label={[text=black!30!green,label distance=-5pt]\cEtwo-180:{$e_2$}}] (e2) at ($(c)+(\cEtwo:\cRad)$) {};

    \node[label={[label distance=-7pt]\cEone:{$q_1$}}] at (e1) {\large $\times$}; 
    \node[label={[label distance=-7pt]\cEtwo:{$q_2$}}] at (e2) {\large $\times$}; 

    \node[myvertex, label=\cRone:{$r_1$}] (r1) at ($(c)+(\cRone:\cRad)$) {};
    \node[myvertex, label=\cRtwo:{$r_2$}] (r2) at ($(c)+(\cRtwo:\cRad)$) {};

    \def\cV{120} 
    \def\cW{300} 
    \def\lU{0} 

    \draw[transform canvas={yshift=0.3mm}, green]  (c) [partial ellipse=270:\cW:\cRad];
    \draw[transform canvas={yshift=-0.3mm}, green] (c) [partial ellipse=90:\cV:\cRad];
    \node[myv, label=\cV:{$v$}] (v) at ($(c)+(\cV:\cRad)$) {};
    \node[myw, label=\cW:{$w$}] (w) at ($(c)+(\cW:\cRad)$) {};
    \path[dotted part of curve] (w)  .. controls ++(100:3) and ++(300:3) .. (v); 
    \draw[transform canvas={xshift=0.3mm}, green] (w) .. controls ++(100:3) and ++(300:3) .. (v); 

    \coordinate (r0) at (0.9, -0.1); 
    \node[myvertex, label=above:{$r_0$}] at (r0) {};

    \node[myvertex, label=\cGone:{$g_1$}] (g1) at ($(c)+(\cGone:\cRad)$) {};
    \node[myvertex, label=\cGtwo:{$g_2$}] (g2) at ($(c)+(\cGtwo:\cRad)$) {};
    \node[myu, label=\lU:{$u$}] at (u) {};
\end{tikzpicture}%
}
                  \end{minipage}
                  \begin{minipage}{0.33\hsize}
                        \centering
                         \scalebox{\picscale}{ %
    \begin{tikzpicture}

    \coordinate (c) at (0, 0); 
    \def\cRad{3cm and 2cm} 

    \def\cRone{150} 
    \def\cRtwo{30} 

    \def\cGone{90}
    \def\cGtwo{270} 

    \def\cEone{320} 
    \def\cEtwo{220} 

    \def\EdgeLength{10} 

    \draw (c) [partial ellipse=0:360:\cRad];

    \draw[myedge] (c) [partial ellipse=\cEone - \EdgeLength:\cEone + \EdgeLength:\cRad];
    \node[label={[text=black!30!green,label distance=-5pt]\cEone-180:{$e_1$}}] (e1) at ($(c)+(\cEone:\cRad)$) {};
    \draw[myedge] (c) [partial ellipse=\cEtwo - \EdgeLength:\cEtwo + \EdgeLength:\cRad];
    \node[label={[text=black!30!green,label distance=-5pt]\cEtwo-180:{$e_2$}}] (e2) at ($(c)+(\cEtwo:\cRad)$) {};

    \node[label={[label distance=-8pt]\cEone:{$q_1$}}] at (e1) {\large $\times$}; 
    \node[label={[label distance=-8pt]\cEtwo:{$q_2$}}] at (e2) {\large $\times$}; 

    \node[myvertex, label=\cRone:{$r_1$}] (r1) at ($(c)+(\cRone:\cRad)$) {};
    \node[myvertex, label=\cRtwo:{$r_2$}] (r2) at ($(c)+(\cRtwo:\cRad)$) {};

    \def\cV{120} 
    \def\cW{300} 
    \def\lU{0} 

    \draw[transform canvas={yshift=0.3mm}, green]  (c) [partial ellipse=270:\cW:\cRad];
    \draw[transform canvas={yshift=-0.3mm}, green] (c) [partial ellipse=90:\cV:\cRad];
    \node[myv, label=\cV:{$v$}] (v) at ($(c)+(\cV:\cRad)$) {};
    \node[myw, label=\cW:{$w$}] (w) at ($(c)+(\cW:\cRad)$) {};
    \path[dotted part of curve] (w)  .. controls ++(100:3) and ++(300:3) .. (v); 
    \draw[transform canvas={xshift=0.3mm}, green] (w) .. controls ++(100:3) and ++(300:3) .. (v); 

    \node[myvertex, label=\cGone:{$g_1$}] (g1) at ($(c)+(\cGone:\cRad)$) {};
    \node[myvertex, label=\cGtwo:{$g_2$}] (g2) at ($(c)+(\cGtwo:\cRad)$) {};
    \node[myu, label=\lU:{$u$}] at (u) {};

    \coordinate (r0) at (-0.6, 0.6); 
    \node[myvertex, label=below:{$r_0$}] at (r0) {};

\end{tikzpicture}%
}
                  \end{minipage}
                        \caption{Proof of Claim~\ref{prop:getting:better:claim:for:cond:two}.}
                        \label{fig:804and805}
            \end{figure}

            Let $Q_1 := Q - \set{e_1} + \set{uv}$ and $Q_2 := Q - \set{e_2} + \set{uv}$.
            It is obvious that $Q_1$ and $Q_2$ are pseudotrees. 

            By changing the roles of $g_1$ and $g_2$ if necessary, we may assume that $v$ is on the $q_1$-$q_2$ path in $Q$ that contains $g_1$. 
            Since $Q' = Q - \set{e_1, e_2} + \set{uv}$ is connected, $w$ is on the $q_1$-$q_2$ path in $Q$ that contains $g_2$ (Figure~\ref{fig:804and805}).

            We first consider the lengths of $\UniC{Q_1}$ and $\UniC{Q_2}$.
            Since $\ell_1 + \dist{g_1}{v}{Q} = \vlabel{v}{Q} \le \frac{1}{2} \newweight{\UniC{Q}} = \ell_1 + \ell_2$, it holds that $\dist{g_1}{v}{Q} \le \ell_2$. 
            Since $\vlabel{u}{Q} = \ell_2 + \dist{g_2}{u}{Q}$, $\vlabel{v}{Q} \le \ell_1 + \ell_2$, and $\vlabel{u}{Q} + \newweight{uv} < \vlabel{v}{Q}$, it holds that 
            $\dist{g_2}{u}{Q} + \newweight{uv} < \ell_1$. 
            By combining these two inequalities, we obtain 
            $\dist{g_2}{u}{Q} + \newweight{uv} + \dist{g_1}{v}{Q} < \ell_1 + \ell_2$. 
            Let $P$ be the unique $g_1$-$g_2$ path in $Q - \{e_1, e_2\} +\set{uv}$. 
            Then, this inequality shows that $\newweight{P} < \ell_1 + \ell_2$.
            Therefore, 
            $\newweight{\UniC{Q_i}} \le \ell_1 + \ell_2 + \newweight{P} < 2(\ell_1 + \ell_2) = \newweight{\UniC{Q}}$ for $i=1,2$, that is, 
            the lengths of $\UniC{Q_1}$ and $\UniC{Q_2}$  are shorter than that of $\UniC{Q}$.

            Recall that $\bar{\ell_1}$ (resp.~$\bar{\ell_2}$) denotes the first coordinate of $\ell_1$ (resp.~$\ell_2$), 
            which corresponds to the length before the perturbation. 
            If $\bar \ell (P) \ge \bar{\ell}_1$, then 
            let $r_0$ be the point on $P$ such that $\distorig{r_0}{g_1}{P} = \bar{\ell_1}$. 
            Then, $r_0$ is in $V \cup R$, because $\distorig{r_0}{v}{P}$ is equal to $\distorig{r_1}{v}{Q}$ or $\distorig{r_2}{v}{Q}$, which is half-integral. 
            If $\bar \ell (P) < \bar{\ell}_1$, then
            let $r_0$ be the point on $P$ that is closest to $g_2$ subject to $r_0 \in V \cup R$. 
            This construction shows that there exists a point $r_0 \in V\cup R$ on $P$ such that
            $\distorig{r_0}{g_1}{P} \le \bar{\ell_1}$ and
            $\distorig{r_0}{g_2}{P}\le \bar{\ell_2}$. 
            Let $V_1 \subseteq V$ (resp.~$V_2 \subseteq V$) be the vertex set of the connected component of $Q-\set{e_1, e_2}$ that contains $v$ (resp.~$w$). 
            For $i=1,2$ and for $x \in V_1$, we obtain 
            \begin{align*}
            \distorig{r_0}{x}{Q_i} &\le \distorig{r_0}{g_1}{P} + \distorig{g_1}{x}{Q}  
                                   \le \bar{\ell_1} + \distorig{g_1}{x}{Q} \\
                                   &= \max \{\distorig{r_1}{x}{Q},\ \distorig{r_2}{x}{Q} \} \le \frac{d}{2},     
            \end{align*}
            where we use (\ref{eq:labelbymiddlepoint}) in the equality. 
            Similarly, for $i=1,2$ and for $x \in V_2$, it holds that
            \begin{align*}
            \distorig{r_0}{x}{Q_i} &\le \distorig{r_0}{g_2}{P} + \distorig{g_2}{x}{Q} 
                                   \le \bar{\ell_2} + \distorig{g_2}{x}{Q} \\ 
                                   &= \max \{\distorig{r_1}{x}{Q},\ \distorig{r_2}{x}{Q} \} \le \frac{d}{2}. 
            \end{align*}
            The above inequalities show that $r_0 \in \inner{Q_i}$ for $i=1,2$. 
            Since $r_i \in \inner{Q_i}$ for $i=1,2$, we obtain $\reconf{r_1}{r_0}{Q_1}$ if $r_0 \not= r_1$ and $\reconf{r_0}{r_2}{Q_2}$ if $r_0 \not= r_2$.

            Since the lengths of $\UniC{Q_1}$ and $\UniC{Q_2}$  are shorter than that of $\UniC{Q}$, we have that 
            $\f{r_1, r_0, Q_1} < \f{r_1, r_2, Q}$ and $\f{r_0, r_2, Q_2} < \f{r_1, r_2, Q}$. 
            Therefore, if $r_0 \not\in \{r_1, r_2\}$, then $r_0$, $Q_1$, and $Q_2$ satisfy the conditions in \cref{prop:stmt:case:2}. 
            If $r_0 = r_1$ (resp.~$r_0=r_2$), then $Q_2$ (resp.~$Q_1$) satisfies the conditions in \cref{prop:stmt:case:1}. 
            This completes the proof. 
      \end{claimproof}

     In what follows, by changing the roles of $g_1$ and $g_2$ if necessary, we may assume that $\dist{g_1}{r_1}{Q} > \dist{g_2}{r_1}{Q}$.
     Then, $g_1$, $q_2$, $r_1$, $g_2$, $r_2$, and $q_1$ appear in this order along $\UniC{Q}$. 
     Let $V_2 \subseteq V$ be the vertex set of the connected component of $G-\{e_1, e_2\}$ that contains $r_1$ and $r_2$. Let $V_1 := V \setminus V_2$. 
     Then, we see that a vertex $x\in V$ is in $V_1$ if and only if $\ell_1 + \dist{g_1}{x}{Q} < \ell_2 + \dist{g_2}{x}{Q}$.
     By the symmetry of $r_1$ and $r_2$, we may assume that $v$ is contained in one of  
     $Q[g_1, q_2]$, $Q[q_2, r_1]$, and $Q[r_1, g_2]$. 
     In the remaining of the proof, we consider each case separately. 

      \newlist{caseenum}{enumerate}{2}
      \setlist[caseenum,1]{label=\textbf{Case \arabic*.}, ref=\textbf{Case \arabic*.}}
      \setlist[caseenum,2]{label=\textbf{\alph*.}, ref=\thecaseenumi\textbf{\alph*.}}
      \begin{caseenum}[leftmargin=*]

                  \begin{figure}[t!]
                        \newcommand{\picscale}{.55}
                        \centering
                        \begin{minipage}{0.3\hsize}
                              \centering
                               \scalebox{\picscale}{ %
    \begin{tikzpicture}
    \basetree; 

    \def\cV{120}; 
    \def\cW{100}; 
    \def\lU{300}; 

    \def\eLength{5}
    \def\cE{\cV-\eLength}; 
    \draw[myedge] (c) [partial ellipse=\cE-\eLength:\cE+\eLength:\cRad];
    \node[label=\cE+180:{$e$}] (e) at ($(c)+(\cE:\cRad)$) {};

    \node[myv, label=\cV:{$v$}] (v) at ($(c)+(\cV:\cRad)$) {};
    \node[myw, label=\cW:{$w$}] (w) at ($(c)+(\cW:\cRad)$) {};
    \path[dotted part of curve] (w)  .. controls ++(270:3) and ++(270:3) .. (v); 
    \node[myu, label=\lU:{$u$}] at (u) {};

\end{tikzpicture}%
}
                              \caption{$w$ is on $Q[g_1, v]$.}
                              \label{fig:getting:better:2:1}
                        \end{minipage}
                        \begin{minipage}{0.3\hsize}
                              \centering
                               \scalebox{\picscale}{ %
\begin{tikzpicture}
    \basetree; 

    \def\cV{120}; 
    \def\cW{70}; 
    \def\lU{300}; 

    \def\eLength{5}
    \def\cE{\cV-\eLength}; 
    \draw[myedge] (c) [partial ellipse=\cE-\eLength:\cE+\eLength:\cRad];
    \node[label=\cE+165:{$e$}] (e) at ($(c)+(\cE:\cRad)$) {};

    \node[myv, label=90:{$v$}] (v) at ($(c)+(\cV:\cRad)$) {};
    \node[myv, label=90:{$v'$}] (vp) at ($(c)+(\cV-2*\eLength:\cRad)$) {};
    \node[myw, label=\cW:{$w$}] (w) at ($(c)+(\cW:\cRad)$) {};
    \path[dotted part of curve] (w)  .. controls ++(270:3) and ++(270:3) .. (v); 
    \node[myu, label=\lU:{$u$}] at (u) {};

\end{tikzpicture}%
}
                              \caption{$w$ is on $Q[q_1, g_1]$.}
                              \label{fig:getting:better:2:7}
                        \end{minipage}
                        \begin{minipage}{0.3\hsize}
                              \centering
                              \scalebox{\picscale}{ %
    \begin{tikzpicture}
    \basetree; 

    \node[myv, label=105:{$v$}] (v) at ($(c)+(105:\cRad)$) {};
    \node[myw, label=120:{$w$}] (w) at ($(c)+(120:\cRad)$) {};
    \path[dotted part of curve] (w)  .. controls ++(270:3) and ++(270:3) .. (v); 
    \node[myu, label=180:{$u$}] at (u) {}; 

\end{tikzpicture}%
}
                              \caption{$w$ is on $Q[v, g_2]$.}
                              \label{fig:getting:better:2:2}
                        \end{minipage}
                  \end{figure}
            \item \label{prop:getting:better:2} 
                  Suppose that $v$ appears on $Q[g_1, q_2]$.
                  We consider the following cases separately. 
                  \begin{caseenum}
                        \item \label{prop:getting:better:2:1}
                              Suppose that $w$ appears on $Q[g_1, v]$ or $Q[q_1, g_1]$ (Figures~\ref{fig:getting:better:2:1} and~\ref{fig:getting:better:2:7}). 
                              Let $e$ be the edge on $Q[g_1, v]$ that is incident to $v$ and  
                              let $Q' := Q - \set{e} + \set{uv}$, which is a pseudotree. 
                              Since $\vlabel{u}{Q} + \newweight{uv} < \vlabel{v}{Q}$, $\vlabel{u}{Q}= \ell_1 + \dist{g_1}{u}{Q}$, and $\vlabel{v}{Q}= \ell_1 + \dist{g_1}{v}{Q}$, 
                              it holds that $\dist{g_1}{u}{Q} + \newweight{uv} < \dist{g_1}{v}{Q}$. 
                              Hence, $\newweight{\UniC{Q'}} < \newweight{\UniC{Q}}$ holds. 
                              For $i=1,2$ and for $x \in V_1$, we obtain     
                                \begin{align*}
                                \dist{r_i}{x}{Q'} &\le \dist{r_i}{g_1}{Q'} + \dist{g_1}{x}{Q'}  \\ 
                                       &\le  \ell_1 + \dist{g_1}{x}{Q} = \max \{\dist{r_1}{x}{Q},\ \dist{r_2}{x}{Q} \}
                                \end{align*}
                              by (\ref{eq:labelbymiddlepoint}), which implies that $\distorig{r_i}{x}{Q'} \le \frac{d}{2}$. 
                              Similarly, for $i=1,2$ and for $x \in V_2$, we obtain     
                                \begin{align*}
                                \dist{r_i}{x}{Q'} &\le \dist{r_i}{g_2}{Q'} + \dist{g_2}{x}{Q'}  \\ 
                                       &=  \ell_2 + \dist{g_2}{x}{Q} = \max \{\dist{r_1}{x}{Q},\ \dist{r_2}{x}{Q} \}, 
                                \end{align*}
                               which implies that $\distorig{r_i}{x}{Q'} \le \frac{d}{2}$. 
                              This shows that $\reconf{r_1}{r_2}{Q'}$, and hence $Q'$ satisfies the conditions in \cref{prop:stmt:case:1}. 

                        \item \label{prop:getting:better:2:2}
                              Suppose that $w$ appears on $Q[v, q_2]$ (Figure~\ref{fig:getting:better:2:2}). 
                              Then,  $\vlabel{v}{Q} < \vlabel{w}{Q} \le \vlabel{u}{Q}$, which contradicts the assumption that
                              $\vlabel{v}{Q} > \vlabel{u}{Q} + \newweight{uv}$. 

                        \item \label{prop:getting:better:2:3-6}
                              Suppose that $w$ appears on $Q[q_2, r_1]$, $Q[r_1, g_2]$, $Q[g_2, r_2]$, or $Q[r_2, q_1]$.
                              In this case, \cref{prop:stmt:case:1} or \cref{prop:stmt:case:2} holds by \cref{prop:getting:better:claim:for:cond:two}.
                  \end{caseenum}
            \item \label{prop:getting:better:3} 
                 Suppose that $v$ appears on on $Q[q_2, r_1]$.
                 We consider the following cases separately. 

                  \begin{figure}[t!]
                        \newcommand{\picscale}{.55}
                        \centering
                        \begin{minipage}{0.4\hsize}
                              \centering
                               \scalebox{\picscale}{ %
\begin{tikzpicture}
    \basetree;

    \def\cV{160}; 
    \def\cW{200}; 
    \def\lU{270}; 

    \def\eLength{5}
    \def\cE{\cV+\eLength}; 
    \draw[myedge] (c) [partial ellipse=\cE-\eLength:\cE+\eLength:\cRad];
    \node[label=180:{$e$}] (e) at ($(c)+(\cE:\cRad)$) {};

    \node[myv, label=\cV:{$v$}] (v) at ($(c)+(\cV:\cRad)$) {};
    \node[myw, label=\cW:{$w$}] (w) at ($(c)+(\cW:\cRad)$) {};

    \path[dotted part of curve] (w)  .. controls ++(0:2) and ++(360:2) .. (v); 
    \node[myu, label=\lU:{$u$}] at (u) {};


\end{tikzpicture}%
}
                              \caption{$w$ is on $Q[v, r_1]$ }
                              \label{fig:getting:better:3:3}
                        \end{minipage}
                        \begin{minipage}{0.4\hsize}
                              \centering
                               \scalebox{\picscale}{ %
    \begin{tikzpicture}
    \basetree;

    \def\cV{180}; 
    \def\cW{250}; 
    \def\lU{270}; 

    \def\cEthree{190}; 
    \draw[myedge] (c) [partial ellipse=\cEthree-\EdgeLength:\cEthree+\EdgeLength:\cRad];
    \node[label=\cEthree:{$e$}] (e) at ($(c)+(\cEthree:\cRad)$) {};

    \node[myv, label=\cV:{$v$}] (v) at ($(c)+(\cV:\cRad)$) {};
    \node[myw, label=\cW:{$w$}] (w) at ($(c)+(\cW:\cRad)$) {};

    \path[opacity=0, get path length=\b] (w)  .. controls ++(90:2) and ++(0:2) .. (v); 
    \draw[opacity=0, get path length=\rw] (c) [partial ellipse=\cRone:\cW:\cRad]; 
    \getlength{\cb}{\b}; 
    \getlength{\crw}{\rw}; 
    \pgfmathsetmacro\posRthree{\crw/\cb}; 

    \path[dotted part of curve] (w)  .. controls ++(90:2) and ++(0:2) .. (v) 
        node[myvertex, label=0:{$r_0$}, pos=\posRthree] {}; 
    \node[myu, label=\lU:{$u$}] at (u) {};

\end{tikzpicture}%
}
                              \caption{$w$ is on $Q[r_1, g_2]$ }
                              \label{fig:getting:better:3:4}
                        \end{minipage}
                  \end{figure}

                  \begin{figure}[t!]
                        \newcommand{\picscale}{.55}
                        \centering
                        \begin{minipage}{0.4\hsize}
                               \scalebox{\picscale}{ %
    \begin{tikzpicture}
    \basetree;


    \def\cV{180}; 
    \def\cW{300}; 
    \def\lU{270}; 

    \def\cEthree{190}; 
    \draw[myedge] (c) [partial ellipse=\cEthree-\EdgeLength:\cEthree+\EdgeLength:\cRad];
    \node[label=\cEthree:{$e$}] (e) at ($(c)+(\cEthree:\cRad)$) {};

    \node[myv, label=\cV:{$v$}] (v) at ($(c)+(\cV:\cRad)$) {};
    \node[myw, label=\cW:{$w$}] (w) at ($(c)+(\cW:\cRad)$) {};

    \path[opacity=0, get path length=\b] (w)  .. controls ++(90:2) and ++(0:2) .. (v); 
    \draw[opacity=0, get path length=\rw] (c) [partial ellipse=\cRone:\cV:\cRad]; 
    \getlength{\cb}{\b}; 
    \getlength{\crw}{\rw}; 
    \pgfmathsetmacro\posRthree{1-\crw/\cb}; 

    \path[dotted part of curve] (w)  .. controls ++(90:2) and ++(0:2) .. (v) 
        node[myvertex, label=90:{$r_0$}, pos=\posRthree] {}; 
    \node[myu, label=\lU:{$u$}] at (u) {};
\end{tikzpicture}%
}
                              \caption{$w$ is on $Q[g_2, r_2]$ }
                              \label{fig:getting:better:3:5}
                        \end{minipage}
                        \begin{minipage}{0.4\hsize}
                               \scalebox{\picscale}{ %
    \begin{tikzpicture}
    \basetree;


    \def\cV{180}; 
    \def\cW{350}; 
    \def\lU{270}; 

    \def\cEthree{190}; 
    \draw[myedge] (c) [partial ellipse=\cEthree-\EdgeLength:\cEthree+\EdgeLength:\cRad];
    \node[label=\cEthree:{$e$}] (e) at ($(c)+(\cEthree:\cRad)$) {};

    \node[myv, label=\cV:{$v$}] (v) at ($(c)+(\cV:\cRad)$) {};
    \node[myw, label=\cW:{$w$}] (w) at ($(c)+(\cW:\cRad)$) {};

    \path[opacity=0, get path length=\b] (w)  .. controls ++(180:2) and ++(0:2) .. (v); 
    \draw[opacity=0, get path length=\rw] (c) [partial ellipse=\cRone:\cV:\cRad]; 
    \getlength{\cb}{\b}; 
    \getlength{\crw}{\rw}; 
    \pgfmathsetmacro\posRthree{1-\crw/\cb}; 

    \path[dotted part of curve] (w)  .. controls ++(180:2) and ++(0:2) .. (v) 
        node[myvertex, label=90:{$r_0$}, pos=\posRthree] {}; 
    \node[myu, label=\lU:{$u$}] at (u) {};
\end{tikzpicture}%
}
                              \caption{$w$ is on $Q[r_2, q_1]$ }
                              \label{fig:getting:better:3:6}
                        \end{minipage}
                  \end{figure}

                  \begin{caseenum}
                        \item \label{prop:getting:better:3:1:7}
                              Suppose that $w$ appears on $Q[q_1, g_1]$ or $Q[g_1, q_2]$.
                              In this case, \cref{prop:stmt:case:1} or \cref{prop:stmt:case:2} holds by \cref{prop:getting:better:claim:for:cond:two}.

                        \item \label{prop:getting:better:3:2}
                              Suppose that $w$ appears on $Q[q_2, v]$. 
                              Then,  $\vlabel{v}{Q} < \vlabel{w}{Q} \le \vlabel{u}{Q}$, which contradicts the assumption that
                              $\vlabel{v}{Q} > \vlabel{u}{Q} + \newweight{uv}$. 

                        \item \label{prop:getting:better:3:4}
                              Suppose that $w$ appears on $Q[v, r_1]$, $Q[r_1, g_2]$, $Q[g_2, r_2]$, or $Q[r_2, q_1]$ (Figures~\ref{fig:getting:better:3:3}--\ref{fig:getting:better:3:6}). 
                              Let $e$ be the edge on $Q[g_2, v]$ that is incident to $v$. 
                              Let $Q_0 := Q - \set{e} + \set{uv}$ and $Q_1 := Q - \set{e_1} + \set{uv}$.  
                              Since $\vlabel{u}{Q} + \newweight{uv} < \vlabel{v}{Q}$, $\vlabel{u}{Q}= \ell_2 + \dist{g_2}{u}{Q}$, and $\vlabel{v}{Q}= \ell_2 + \dist{g_2}{v}{Q}$, 
                              it holds that $\dist{g_2}{u}{Q} + \newweight{uv} < \dist{g_2}{v}{Q}$. 
                              Hence, $\newweight{\UniC{Q_0}} < \newweight{\UniC{Q}}$ holds. 
                              We also have that $\newweight{\UniC{Q_1}} \le  \dist{g_2}{v}{Q} + \dist{g_2}{u}{Q} + \newweight{uv} < 2 \dist{g_2}{v}{Q} < \newweight{\UniC{Q}}$. 
                              Therefore, the lengths of $\UniC{Q_0}$ and $\UniC{Q_1}$  are shorter than that of $\UniC{Q}$.

                              Let $P$ be the unique $g_2$-$v$ path in $Q - \{e, e_1\} + \set{uv}$. 
                              Since $\size{E(P)} \le \distorig{g_2}{v}{Q} = \bar \ell_2 + \distorig{r_1}{v}{Q}$,
                              there exists a point $r_0 \in V\cup R$ on $P$ such that
                              $\distorig{r_0}{g_2}{P} \le \bar{\ell}_2$ and $\distorig{r_0}{v}{P}\le \distorig{r_1}{v}{Q}$.
                              Note that such $r_0$ exists in $V \cup R$, because $\distorig{r_1}{v}{Q}$ is half-integral. 
                              
                              Then, we see that $r_0 \in \inner{Q_i}$ for $i=0,1$,  
                              $r_2 \in \inner{Q_0}$, and $r_1 \in \inner{Q_1}$ by the following observations. 
                              \begin{itemize}
                                  \item 
                                  For $i=0,1$ and for $x \in V_1$,  
                                  \begin{align*}
                                     \distorig{r_0}{x}{Q_i} &\le \distorig{r_0}{v}{P} + \distorig{v}{g_1}{Q} + \distorig{g_1}{x}{Q}  \\ 
                                                            &\le \distorig{r_1}{v}{Q} + \distorig{v}{g_1}{Q} + \distorig{g_1}{x}{Q} \\ 
                                                            &=  \bar{\ell_1} + \distorig{g_1}{x}{Q} = \max\{\distorig{r_1}{x}{Q},\ \distorig{r_2}{x}{Q} \} \le \frac{d}{2}.      
                                  \end{align*}

                                  \item
                                  For $i=0,1$ and for $x \in V_2$, 
                                  \begin{align*}
                                     \distorig{r_0}{x}{Q_i} &\le \distorig{r_0}{g_2}{P} + \distorig{g_2}{x}{Q_i} \\ 
                                         &\le \bar{\ell}_2 + \distorig{g_2}{x}{Q} = \max\{\distorig{r_1}{x}{Q},\ \distorig{r_2}{x}{Q} \} \le \frac{d}{2}.     
                                  \end{align*}
                                  
                                  \item
                                  For $x \in V_j$ with $j \in \{1, 2\}$, 
                                  \begin{align*}
                                     \distorig{r_2}{x}{Q_0} &\le \distorig{r_2}{g_j}{Q_0} + \distorig{g_j}{x}{Q_0}  \\ 
                                                            &\le \bar{\ell_j} + \distorig{g_j}{x}{Q} = \max\{\distorig{r_1}{x}{Q},\ \distorig{r_2}{x}{Q} \} \le \frac{d}{2}.      
                                  \end{align*}

                                  \item
                                  For $x \in V$, it holds that $\distorig{r_1}{x}{Q_1} \le \distorig{r_1}{x}{Q - e_1} = \distorig{r_1}{x}{Q} \le \frac{d}{2}$, because $Q - e_1$ is the shortest path tree starting from $r_1$ in $Q$.  
                              \end{itemize}
                              Therefore, if $r_0 \not\in \{r_1, r_2\}$, then we obtain $\reconf{r_1}{r_0}{Q_1}$ and $\reconf{r_0}{r_2}{Q_0}$, and hence
                              $r_0$, $Q_1$, and $Q_0$ satisfy the conditions in \cref{prop:stmt:case:2}. 
                              If $r_0 = r_1$ (resp.~$r_0=r_2$), then $Q_0$ (resp.~$Q_1$) satisfies the conditions in \cref{prop:stmt:case:1}. 
                  \end{caseenum}
                  
            \item \label{prop:getting:better:4} 
                   Suppose that $v$ appears on on $Q[r_1, g_2]$.
                   We consider the following cases separately. 
                  \begin{figure}[t!]
                        \centering
                        \newcommand{\picscale}{.55}
                        \begin{minipage}{0.4\hsize}
                               \scalebox{\picscale}{ %
\begin{tikzpicture}
    \basetree;


    \def\cV{230}; 
    \def\cW{250}; 
    \def\lU{0}; 

    \def\eLength{5}
    \def\cE{\cV+\eLength}; 
    \draw[myedge] (c) [partial ellipse=\cE-\eLength:\cE+\eLength:\cRad];
    \node[label=270:{$e$}] (e) at ($(c)+(\cE:\cRad)$) {};

    \node[myv, label=\cV:{$v$}] (v) at ($(c)+(\cV:\cRad)$) {};
    \node[myw, label=\cW+20:{$w$}] (w) at ($(c)+(\cW:\cRad)$) {};

    \path[dotted part of curve] (w)  .. controls ++(90:2.5) and ++(90:2.5) .. (v); 
    \node[myu, label=\lU:{$u$}] at (u) {};

\end{tikzpicture}%
}
                              \caption{$w$ is on $Q[v, g_2]$.}
                              \label{fig:getting:better:4:4}
                        \end{minipage}
                        \begin{minipage}{0.4\hsize}
                               \scalebox{\picscale}{ %
\begin{tikzpicture}
    \basetree;


    \def\cV{240}; 
    \def\cW{300}; 
    \def\lU{0}; 

    \def\eLength{5}
    \def\cE{\cV+\eLength}; 
    \draw[myedge] (c) [partial ellipse=\cE-\eLength:\cE+\eLength:\cRad];
    \node[label=270:{$e$}] (e) at ($(c)+(\cE:\cRad)$) {};

    \node[myv, label=\cV:{$v$}] (v) at ($(c)+(\cV:\cRad)$) {};
    \node[myw, label=\cW+20:{$w$}] (w) at ($(c)+(\cW:\cRad)$) {};

    \path[dotted part of curve] (w)  .. controls ++(90:2.5) and ++(90:2.5) .. (v); 
    \node[myu, label=\lU:{$u$}] at (u) {};

\end{tikzpicture}%
}
                              \caption{$w$ is on $Q[g_2, r_2]$.}
                              \label{fig:getting:better:4:5}
                        \end{minipage}
                  \end{figure}
                  \begin{caseenum}
                        \item \label{prop:getting:better:4:1:7}
                              Suppose that $w$ appears on $Q[q_1, g_1]$ or $Q[g_1, q_2]$.
                              In this case, \cref{prop:stmt:case:1} or \cref{prop:stmt:case:2} holds by \cref{prop:getting:better:claim:for:cond:two}.

                        \item \label{prop:getting:better:4:2:3}
                              Suppose that $w$ appears on $Q[q_2, v]$. 
                              Then,  $\vlabel{v}{Q} < \vlabel{w}{Q} \le \vlabel{u}{Q}$, which contradicts the assumption that
                              $\vlabel{v}{Q} > \vlabel{u}{Q} + \newweight{uv}$. 

                        \item \label{prop:getting:better:4:4}
                              Suppose that $w$ appears on $Q[v, g_2]$ or $Q[g_2, r_2]$  (Figures~\ref{fig:getting:better:4:4} and~\ref{fig:getting:better:4:5}). 
                              Let $e$ be the edge on $Q[v, g_2]$ that is incident to $v$ 
                              and let $Q' := Q - \set{e} + \set{uv}$.  
                              Since $\vlabel{u}{Q} + \newweight{uv} < \vlabel{v}{Q}$, 
                              $\vlabel{u}{Q}= \ell_2 + \dist{g_2}{u}{Q}$, and $\vlabel{v}{Q}= \ell_2 + \dist{g_2}{v}{Q}$, 
                              it holds that $\dist{g_2}{u}{Q} + \newweight{uv} < \dist{g_2}{v}{Q}$. 
                              Hence, $\newweight{\UniC{Q'}} < \newweight{\UniC{Q}}$ holds. 
                              For $i=1,2$ and for $x \in V_1$, we obtain     
                                \begin{align*}
                                \dist{r_i}{x}{Q'} &\le \dist{r_i}{g_1}{Q'} + \dist{g_1}{x}{Q'}  \\ 
                                                  &=  \ell_1 + \dist{g_1}{x}{Q} = \max \{\dist{r_1}{x}{Q},\ \dist{r_2}{x}{Q} \}
                                \end{align*}
                              by (\ref{eq:labelbymiddlepoint}), which implies that $\distorig{r_i}{x}{Q'} \le \frac{d}{2}$. 
                              Similarly, for $i=1,2$ and for $x \in V_2$, we obtain     
                                \begin{align*}
                                \dist{r_i}{x}{Q'} &\le \dist{r_i}{g_2}{Q'} + \dist{g_2}{x}{Q'}  \\ 
                                       &\le  \ell_2 + \dist{g_2}{x}{Q} = \max \{\dist{r_1}{x}{Q},\ \dist{r_2}{x}{Q} \}, 
                                \end{align*}
                               which implies that $\distorig{r_i}{x}{Q'} \le \frac{d}{2}$. 
                              This shows that $\reconf{r_1}{r_2}{Q'}$, and hence $Q'$ satisfies the conditions in \cref{prop:stmt:case:1}.

                        \item \label{prop:getting:better:4:6}
                              Suppose that $w$ appears on $Q[r_2, q_1]$. 
                              Then,  $\vlabel{v}{Q} \le \vlabel{r_1}{Q} = \vlabel{r_2}{Q} \le \vlabel{w}{Q} \le \vlabel{u}{Q}$, which contradicts the assumption that
                              $\vlabel{v}{Q} > \vlabel{u}{Q} + \newweight{uv}$. 
                  \end{caseenum}
      \end{caseenum}
        By the above case analysis, the proposition holds. 
\end{proof}

By this proposition, we can derive Theorem~\ref{thm:good:sequence}.

\begin{proof}[Proof of Theorem~\ref{thm:good:sequence}]
      It suffices to prove the following statement by the induction on $\f{r_1, r_2, Q}$: for a triplet $\tuple{r_1, r_2, Q}$ with $\reconf{r_1}{r_2}{Q}$, $\recgraph'$ contains an $r_1$-$r_2$ path. 
      
      Since $\f{\cdot}$ can take only finitely many values, there exists a triplet $\tuple{r^*_1, r^*_2, Q^*}$ with $\reconf{r^*_1}{r^*_2}{Q^*}$ that minimizes $\f{r^*_1, r^*_2, Q^*}$.  
      If $\tuple{r^*_1, r^*_2, Q^*}$ is not good, then we can apply \cref{prop:getting:better} to obtain a tripe with smaller $\f{\cdot}$, which contradicts the minimality of $\f{r^*_1, r^*_2, Q^*}$. 
      Therefore, $\tuple{r^*_1, r^*_2, Q^*}$ is good, and hence $r^*_1 r^*_2 \in E(\recgraph')$, which shows the base case of the statement. 

     To show the induction step, suppose that $\tuple{r_1, r_2, Q}$ is a triplet with $\reconf{r_1}{r_2}{Q}$. 
     If $\tuple{r_1, r_2, Q}$ is good, then $r_1 r_2 \in E(\recgraph')$, and hence we are done. 
     Otherwise, we apply \cref{prop:getting:better} to obtain \cref{prop:stmt:case:1} or \cref{prop:stmt:case:2}. 
     When we obtain \cref{prop:stmt:case:1}, since $\reconf{r_1}{r_2}{Q'}$ and $\f{r_1, r_2, Q'} < \f{r_1, r_2, Q}$, $\recgraph'$ contains an $r_1$-$r_2$ path by the induction hypothesis. 
     When we obtain \cref{prop:stmt:case:2}, the induction hypothesis shows that $\recgraph'$ contains an $r_1$-$r_0$ path and an $r_0$-$r_2$ path, which shows the existence of an $r_1$-$r_2$ path in $\recgraph'$.
\end{proof}


\section{Concluding Remarks}
\label{sec:conclusion}

   In this paper, we have investigated the computational complexity of \textsc{RST with Small ({\rm or} Large) Maximum Degree} and {\sc RST with Small ({\rm or} Large) Diameter}.
   
    We have proved in Theorem~\ref{thm:smd_hard} that {\sc RST with Small Maximum Degree} is PSPACE-complete for $d \ge 3$. 
    One can naturally ask what happens for the case of maximum degree at most $2$. 
    In this case, the problem becomes the \textsc{Hamiltonian Path Reconfiguration} problem, in which a feasible solution is a Hamiltonian path. 
    We were not able to determine the complexity of this problem and we left it as an open problem. 
    Note that \textsc{Hamiltonian Path Reconfiguration} problem can be also seen as a special case of {\sc RST with Large Diameter} in which the lower bound on the diameter is $|V(G)|-1$. 
    Note also that, for the Hamiltonian \emph{cycle} case, the \textsc{Hamiltonian Cycle Reconfiguration} problem is known to be PSPACE-complete~\cite{Takaoka18}, in which two edge flips are executed in one step.



    We have proved in Theorem~\ref{thm:ldiam_hard} that {\sc RST with Large Diameter} is NP-hard, but it is unclear whether this problem belongs to the class NP. 
    We conjecture that the problem is PSPACE-complete, and left this question as another open problem.

\newpage
\bibliography{ref}

\end{document}